\documentclass[conference]{IEEEtran}
\ifCLASSINFOpdf
\else
\fi

\usepackage{amsfonts}
\usepackage{amssymb}
\usepackage{amsmath}
\usepackage{amsthm}
\usepackage{mathtools}
\usepackage{amstext}
\usepackage{array}

\usepackage{cite}

\usepackage{misc}
\usepackage{xspace}
\usepackage{pgf}
\usepackage{tikz}
\usetikzlibrary{fit,shapes}

\newcommand{\itatomich}{\textit{(atom$_\textit{h}$)}\xspace}
\newcommand{\itbackh}{\textit{(back$_\textit{h}$)}\xspace}
\newcommand{\itforthh}{\textit{(forth$_\textit{h}$)}\xspace}

\newcommand{\itsimb}{\textit{(sim)}\xspace}

\newcommand{\atomf}{\textit{atom}\!$_f$\!}
\newcommand{\atomh}{\textit{atom}$_h$\!}

\newcommand{\forth}{\textit{forth}}
\newcommand{\forthh}{\textit{forth}$_h$\!}
\newcommand{\backh}{\textit{back}$_h$\!}

\newcommand{\atomgh}{\textit{atom}$^g_h$}
\newcommand{\forthgh}{\textit{forth}$^g_h$}
\newcommand{\forthggh}{\textit{forth}$^{gg}_h$}
\newcommand{\forthg}{\textit{forth}$^g$}
\newcommand{\backgh}{\textit{back}$^g_h$}
\newcommand{\simgh}{\textit{sim}$^g_h$}
\newcommand{\atomg}{\textit{atom}$^g$}

\newcommand{\hALC}{\ensuremath{\textit{horn}\mathcal{ALC}}}
\newcommand{\hGF}{\ensuremath{\textit{horn}\textsl{GF}}}

\theoremstyle{plain}
\newtheorem{theorem}{Theorem}
\newtheorem{lemma}{Lemma}

\newtheorem{definition}{Definition}
\newtheorem{example}{Example}

\begin{document}
%
\title{Model Comparison Games for Horn Description Logics}

\author{
\IEEEauthorblockN{Jean Christoph Jung}
\IEEEauthorblockA{Universit\"at Bremen\\
Germany\\
jeanjung@uni-bremen.de}
\and
\IEEEauthorblockN{Fabio Papacchini and Frank Wolter}
\IEEEauthorblockA{Department of Computer Science\\
University of Liverpool, UK\\
\{papacchf,wolter\}@liverpool.ac.uk}
\and
\IEEEauthorblockN{Michael Zakharyaschev}
\IEEEauthorblockA{Department of Computer Science\\and Information Systems\\
Birkbeck, University of London, UK\\
michael@dcs.bbk.ac.uk}
}


%

\IEEEoverridecommandlockouts
\IEEEpubid{\makebox[\columnwidth]{978-1-7281-3608-0/19/\$31.00~
\copyright2019 IEEE \hfill} \hspace{\columnsep}\makebox[\columnwidth]{ }}

\maketitle


\begin{abstract}
Horn description logics are syntactically defined fragments of standard description logics that fall within the Horn fragment of first-order logic and for which ontology-mediated query answering is in \textsc{PTime} for data complexity. They were independently introduced in modal logic to capture the intersection of Horn first-order logic with modal logic. In this paper, we introduce model comparison games for the basic Horn description logic \hALC{} (corresponding to the basic Horn modal logic) and use them to obtain an Ehrenfeucht-Fra\"{i}ss\'e type definability result and a van Benthem style expressive completeness result for \hALC{}. We also establish a finite model theory version of the latter.  
The Ehrenfeucht-Fra\"{i}ss\'e type definability result is used to show that checking \hALC{} indistinguishability of models is \textsc{ExpTime}-complete, which is in sharp contrast to $\mathcal{ALC}$ indistinguishability (i.e., bisimulation equivalence) checkable in \textsc{PTime}. In addition, we explore the behavior of Horn fragments of more expressive description and modal logics by defining a Horn guarded fragment of  first-order logic and introducing model comparison games for it.
\end{abstract}


%
\IEEEpeerreviewmaketitle

\section{Introduction}
Description logics (DLs)~\cite{Baader-et-al-03b,DBLP:books/daglib/0041477} have been introduced as knowledge representation formalisms supported by efficient reasoning tools. The basic DL, called $\mathcal{ALC}$, is a notational variant of the classical multi-modal logic. In fact, numerous applications have led to the development of a large family of DLs with different features. DLs serve as the logical underpinning of ontologies, finitely axiomatized theories known as TBoxes. Two main  reasoning problems have to be solved efficiently for TBoxes, often containing thousands of axioms:
\begin{description}
\item[\emph{Deduction}:] \hspace*{0.6cm} does a formula follow from a TBox?
 
\item[\emph{Query answering}:] \hspace*{1.6cm} is a tuple $\dbf$ in a database $D$ a certain answer to an ontology-mediated query $(\Tmc,\boldsymbol{q}(\xbf))$ comprising a TBox $\Tmc$ and a query $\boldsymbol{q}(\xbf)$? In other words, does $\boldsymbol{q}(\dbf)$ follow from $\Tmc\cup D$?
\end{description}
As the data is typically much larger than both TBox and query, the
appropriate efficiency measure for ontology-mediated query answering
is \emph{data complexity}, under which the database is the only input
to the problem, while the TBox and query are regarded as
fixed~\cite{DBLP:conf/stoc/Vardi82}. For $\mathcal{ALC}$, deduction is
\ExpTime-complete and ontology-mediated query answering with
conjunctive queries is \coNP-complete~\cite{DBLP:books/daglib/0041477}.

Horn DLs have been introduced as syntactically defined fragments of standard DLs that fall within the Horn fragment of first-order logic (henceforth Horn FO) and for which ontology-mediated query answering is in \PTime for data complexity~\cite{DBLP:conf/ijcai/HustadtMS05,DBLP:journals/jar/HustadtMS07}. The Horn fragment of $\mathcal{ALC}$ is called \hALC{}. The modal logic corresponding to \hALC{} was introduced independently, actually five years earlier, with the aim of capturing the intersection of Horn FO and modal logic~\cite{Sturm00}.\!\footnote{The results obtained in this paper could have been presented as a contribution to modal rather than description logic. The only reason why we have chosen the DL environment is that the impact of Horn fragments in description logic has so far been much more significant than in modal logic.}
The introduction of Horn DLs had an enormous impact on description logic research and applications: while the weaker Horn DLs of the $\mathcal{EL}$~\cite{BaBrLu-IJCAI-05} and \textsl{DL-Lite}~\cite{Romans,DBLP:journals/jair/ArtaleCKZ09} families gave rise to two Web Ontology Language OWL~2 profiles (trading expressiveness for high efficiency), the more expressive Horn DLs starting at \hALC{} have also been used extensively, and investigated in depth for ontology-mediated query answering~\cite{DBLP:conf/jelia/EiterGOS08,ijcai2011s,DBLP:conf/rweb/BienvenuO15,DBLP:conf/kr/GrauHKKMMW12,DBLP:conf/ijcai/Bienvenu0LW16,DBLP:conf/aaai/GlimmKT17}. Moreover, despite the fact that deduction in many expressive Horn DLs, including \hALC{}, is \ExpTime-hard, it turned out that efficient reasoners capable of coping with very large real-world TBoxes could be developed~\cite{conf/ijcai/Kazakov09,DBLP:series/ssw/Krotzsch10}. The complexity of reasoning in various types of Horn DLs has been investigated in~\cite{DBLP:journals/tocl/KrotzschRH13}. It also turned out that basic questions relevant for ontology-mediated query answering, such as query-inseparability and conservative extensions, query emptiness, and query by example, admit more elegant solutions and are easier to solve computationally for Horn DLs than in the classical case~\cite{emptiness,DBLP:journals/ai/BotoevaKRWZ16,GuJuSa-IJCAI18,Botoevaetal19}. The relationship between Horn DLs 
and \PTime query answering is by now well-understood~\cite{DBLP:journals/lmcs/LutzW17,HLPW-IJCAI18}. 

In contrast, the model theory for expressive Horn DLs remains largely
undeveloped. Even very basic questions such as whether \hALC{} indeed
captures the intersection of $\mathcal{ALC}$ (or modal logic) and Horn
FO are still unanswered. The aim of this paper is to lay foundations
for a model-theoretic understanding of Horn DLs \textup{(}and Horn
modal logic\textup{)} by developing model comparison games and using
them to obtain Ehrenfeucht-Fra\"{i}ss\'e type definability and van
Benthem style expressive completeness results. In a first application
of these results, we show that concept learning and model
indistinguishability in \hALC{} are \ExpTime-complete and that \hALC{} does not capture the intersection of $\mathcal{ALC}$ and Horn FO. 

The original definition of Horn DLs~\cite{DBLP:conf/ijcai/HustadtMS05}
was based on the polarity of concepts, as used in automated theorem
proving. The equivalent definition given for modal
logic~\cite{Sturm00} (and also for
DLs~\cite{DBLP:journals/lmcs/LutzW17}) is more similar to the classical definition of Horn FO as the closure of Horn clauses $\varphi_{1}\wedge \cdots \wedge \varphi_{n}\rightarrow \varphi_{n+1}$, with atomic $\varphi_{i}$, under $\wedge$, $\exists$, and $\forall$. The obvious analogue of this definition in description (modal) logic is the closure of Horn clauses $A_{1}\sqcap \cdots \sqcap A_{n}\rightarrow A_{n+1}$, with concept names (propositional variables in modal logic) $A_{i}$, under $\sqcap$, $\exists R$, $\forall R$ (respectively, $\wedge$, $\Diamond$, $\Box$).
However, in contrast to the first-order case, this definition leaves out the substitution instances $C_{1} \sqcap \dots \sqcap C_{n}\rightarrow C$ with positive existential $C_i$ and Horn $C$ (understood recursively), which have thus been explicitly included in \hALC. It is easy to show that \hALC{} is a fragment of Horn FO under the standard translation. 


As the first contribution of the paper, we introduce model comparison
games for characterizing \hALC. These \emph{Horn simulation games}
differ from standard bisimulation or Ehrenfeucht-Fra\"{i}ss\'e games
in the following respects: 
%
\begin{enumerate}
\item[(1)] the Horn simulation relations underlying Horn simulation games are non-symmetric (which reflects
  that Horn languages are not closed under negation); 

\item[(2)] positions in the games consist of pairs $(X,b)$ with a \emph{set} $X$ of nodes and a node $b$ (which reflects that Horn languages are not closed under disjunction);

\item[(3)] Horn simulation games use as a subgame the basic simulation game for checking indistinguishability by positive existential $\mathcal{ALC}$ formulas (which reflects that the left-hand side of \hALC{} implications are such).
\end{enumerate}
Both~(2) and (3) have important consequences. The latter means that
Horn simulation games are \emph{modular} as far as the
characterization of the left-hand side of implications is concerned.
For example, by dropping the subgames entirely, we characterize the
positive fragment of $\mathcal{ALC}$ and by restricting it to depth
$0$, we characterize the weaker Horn fragment of $\mathcal{ALC}$
discussed above. We will use this modularity to characterize 
a proper extension, \hALC$_{\nabla}$, of \hALC{} with the operators $\nabla R.C = \exists R.\top \sqcap \forall R.C$ (or $\nabla p=\Diamond \top \wedge \Box p$ in modal logic) on the left-hand side of \hALC{} 
implications, which also lies in Horn FO.

The consequences of~(2) are three-fold. First, using sets rather than nodes in positions implies that the obvious algorithm checking the existence of Horn simulations containing a pair $(\{a\},b)$ of nodes runs in exponential time. Thus, using Horn simulation games to check whether two nodes $a$ and $b$ satisfy the same \hALC-concepts or whether two models satisfy the same TBox axioms yields exponential time algorithms. We show that this is unavoidable by proving corresponding \ExpTime lower bounds. 
The \ExpTime-completeness results are in sharp contrast to the typical
complexity of indistinguishability in modal-like languages. For
example, as bisimilarity of nodes can be checked in polynomial time,
deciding whether two nodes satisfy the same $\mathcal{ALC}$-concepts
is in \PTime; similarly, since one can check in polynomial time
whether there is a standard simulation between two nodes, deciding
whether they satisfy the same $\mathcal{EL}$-concepts is in \PTime as well. 
Thus, \hALC{} sits between languages for which definability questions are
computationally and model-theoretically much more
straightforward.

Second, as player~2 does \emph{not} have a winning strategy in position
$(X,b)$ in the Horn simulation game if, and only if, there exists a
\hALC-concept that is true at all nodes in $X$ but not true at $b$,
our complexity results are directly applicable to the \emph{concept
learning by example \textup{(}CBE\textup{)} problem}: given a data
set, and sets $P$ and $N$ of positive and negative examples, does
there exist a \hALC-concept $C$ separating $P$ from $N$ over the
data? The goal of this supervised learning problem is to
automatically derive new concept descriptions from labelled data. It
has been investigated before in
DL~\cite{DBLP:journals/jmlr/Lehmann09,DBLP:conf/ilp/BadeaN00, DBLP:conf/kr/CohenH94} and
for many logical languages, in particular in
databases~\cite{DBLP:conf/icdt/CateD15,DBLP:conf/icdt/Barcelo017,DBLP:journals/tods/ArenasD16,DBLP:conf/lics/GroheR17}.
Horn DLs are of particular interest as target languages for CBE as
they can be regarded as `maximal DLs without disjunction,\!' and the
unlimited use of disjunction in derived concept descriptions is
undesirable as it leads to \emph{overfitting}: learnt concepts enumerate the positive
examples rather than generalize from the
examples~\cite{DBLP:conf/ilp/BadeaN00}. The complexity
analysis for Horn simulation games shows that the CBE problem for
\hALC{} is \ExpTime-complete, again in contrast to \ALC, where CBE is
in \PTime. We regard the increased complexity as the price for obtaining
proper generalizations.

Finally, the presence of sets in positions of the Horn simulation
games has an impact on the standard infinitary saturated model
approach to proving van Benthem style expressive completeness
results~\cite{goranko20075,DBLP:books/daglib/p/Gradel014}. For
example, we aim to prove that an FO-formula with one free variable is
equivalent to (the standard translation of) a \hALC-concept just in
case it is preserved under Horn simulations. 
Then, for the infinitary proof method, not only do the structures
showing that non-equivalence to a \hALC{}-concept implies
non-preservation under Horn simulations have to satisfy appropriate
saturatedness conditions, but also the \emph{substructures} induced by
the sets $X$ chosen by the players have to be saturated. However, saturated structures do not enjoy the latter property for arbitrary subsets $X$ of
their domain.  In fact, it currently seems that the only way to obtain
expressive completeness results with an infinitary approach is to
restrict the moves of players to `saturated sets,' say sets definable
as the intersection of FO-definable sets.

In this paper, we prove van Benthem style expressive completeness
results for $\hALC$-concepts and TBoxes via Horn simulation games by 
developing appropriate finitary methods which do not require saturated structures. 
As a consequence, the results hold both in the classical and the
finite model theory setting, and without any restrictions on the moves of players. In fact, we show that preservation under $\ell$-round Horn simulation games coincides 
with preservation under infinitary Horn simulation games for $\mathcal{ALC}$-concepts and TBoxes of nesting depth $\ell$. We thus also obtain 
decidability of the problem whether an $\mathcal{ALC}$-concept or TBox is equivalent to a \hALC-concept or TBox, respectively. 
The finitary approach to van Benthem style expressive completeness results was first used by Rosen~\cite{DBLP:journals/jolli/Rosen97} to 
obtain a bisimulation characterization of $\mathcal{ALC}$ in the
finite model theory setting, and has been further developed and applied with great success by 
Otto et al.~\cite{DBLP:journals/apal/Otto04,DBLP:journals/apal/DawarO09,DBLP:journals/jacm/Otto12,DBLP:books/daglib/p/Gradel014}. The lifting of our results from expressive
completeness for \hALC{} within $\mathcal{ALC}$ to expressive completeness for \hALC{} within FO relies on these earlier results. 
%

It is straightforward to extend the Horn simulation games for \hALC{} to games providing Ehrenfeucht-Fra\"{i}ss\'e type
definability results for the Horn fragments of many popular extensions of $\mathcal{ALC}$, such as the extension by inverse roles
or the universal role. Instead of going through those extensions step-by-step, however, we consider the guarded fragment, GF, of
FO and introduce its Horn fragment, \hGF{}, by generalizing the definition of Horn DLs in the obvious way. We remind the
reader that GF has been introduced as an extension of multi-modal logic to predicates of arbitrary arity, which still has many
of the fundamental properties of modal and description logics~\cite{ANvB98,DBLP:journals/jsyml/Gradel99,goranko20075,DBLP:journals/corr/BaranyGO13,pods17}. 
Like \hALC{}, \hGF{} is contained in Horn FO and ontology-mediated query answering using conjunctive queries is in \PTime for data complexity. 
The latter can be shown by establishing a close relationship between \hGF{} and guarded tuple-generating dependencies (guarded tgds), a popular member of the 
Datalog$^\pm$ family for which query answering is in \PTime~\cite{DBLP:journals/jair/CaliGK13,DBLP:conf/pods/CaliGL09}.
In fact, guarded tgds can be seen as normal forms for \hGF{}, and deduction and query answering in \hGF{} can both be polynomially reduced to 
the same problem for guarded tgds, and vice versa.
To study the model theory of \hGF{}, we generalize Horn simulations to guarded Horn simulations, and show an Ehrenfeucht-Fra\"{i}ss\'e type
definability result for \hGF{}. This result is used to prove an \ExpTime upper bound for model indistinguishability in \hGF{}
and to explore the expressive power of \hGF. In particular, we show
that \hGF{} captures more of the intersection of $\mathcal{ALC}$ and Horn FO
than \hALC{} but does not capture the intersection of GF and Horn FO. We then show expressive completeness of \hGF{}: an FO-formula is equivalent to a \hGF{}-formula just
in case it is preserved under guarded Horn simulations. Our proof uses infinitary methods and thus the moves of player~1 are restricted
to intersections of FO-definable sets. It remains open whether the expressive completeness holds without this restriction and whether it holds in the finite model theory setting.  

The emerging landscape of the fragments of Horn FO and GF we considered in this paper is discussed in the conclusion.



\smallskip

{\bf Related Work}. Here we briefly review the related work not yet discussed. 
The definition of Horn simulations is inspired by games used to provide van Benthem style characterizations 
of concepts in weak DLs such as $\mathcal{FL}^{-}$~\cite{DBLP:journals/ai/KurtoninaR99}. Van Benthem style
characterizations of DLs in the $\mathcal{EL}$ and \textsl{DL-Lite} families are given  in~\cite{TBoxpaper}. Bisimulations
have been studied for the guarded fragment and many variations~\cite{ANvB98,DBLP:books/daglib/p/Gradel014,DBLP:journals/jacm/BaranyCS15}.
Bisimulations have also been studied recently for coalgebraic modal logics~\cite{DBLP:conf/calco/GorinS13,DBLP:journals/logcom/SchroderPL17} 
and fuzzy modal logics~\cite{DBLP:conf/lics/WildSP018}.

This paper contributes to the model theory of languages obtained by taking the
intersection of Horn FO with modal and description logic. Horn FO was originally introduced in classical 
model theory~\cite{DBLP:journals/jsyml/McKinsey43,DBLP:journals/jsyml/Horn51} with the aim of understanding FO-formulas 
that are preserved under products of models. As it turned out, an FO-formula is equivalent to a Horn formula iff 
it is preserved under \emph{reduced} products; for details consult~\cite{ChangKeisler,Hodges93}. A complicated recursive
characterization of FO-sentences preserved under direct products is given in~\cite{Weinstein65}. 

There have been other attempts to define Horn modal and temporal logics~\cite{Nguyen04,CerroP87,ChenLin93,BresolinMS16} with
a focus on the complexity of reasoning and not on model theory.

\section{Preliminaries}
\label{sec:characterisation}

Description logics (DLs) are fragments of first-order logic with unary and binary predicates. 
However, the standard notation for DL `formulas' is more succinct and does not use individual 
variables explicitly~\cite{Baader-et-al-03b,DBLP:books/daglib/0041477}. Let $\tau$ be a vocabulary
consisting of unary and binary predicate names only. In DL parlance, they are called \emph{concept names} (denoted $A$, $B$, etc.) and \emph{role names} (denoted $R$, $S$, etc.), respectively.
The \emph{$\mathcal{ALC}[\tau]$-concepts}, $C$, are defined by the following grammar:
\begin{multline*}
C,D ~::=~ A \ \mid \ \top \ \mid \ \bot \ \mid \ \neg C \ \mid \ C \sqcup D \ \mid \ 
C \sqcap D \ \mid\\
C \to D \ \mid \ \exists R.C \ \mid \ \forall R.C,
\end{multline*}
where $A\in \tau$ is unary, $R\in \tau$ binary, $\top$ is the universal and $\bot$ the empty concept. If not relevant, we drop $\tau$ and simply talk about $\mathcal{ALC}$-concepts.
An $\mathcal{ALC}[\tau]$-\emph{concept inclusion} (or CI) takes the form $C\sqsubseteq D$, where $C$ and $D$
are $\mathcal{ALC}[\tau]$-concepts. An $\mathcal{ALC}[\tau]$-\emph{TBox}, $\Tmc$, is a finite set of 
$\mathcal{ALC}[\tau]$-CIs. 

$\mathcal{ALC}[\tau]$ is interpreted in usual $\tau$-structures 
$$
\Amf=(\text{dom}(\Amf),(A^{\Amf})_{A\in \tau},(R^{\Amf})_{R\in \tau})
$$
with $\text{dom}(\Amf) \ne \emptyset$, $A^{\Amf}\subseteq \text{dom}(\Amf)$ and
$R^{\Amf}\subseteq\text{dom}(\Amf)^{2}$. The semantics of $\mathcal{ALC}$ can be defined 
via the \emph{standard translation} $^\dag$ of $\mathcal{ALC}$-concepts to FO-formulas with one  free variable $x$:
\begin{align*}
& A^{\dag}  = A(x), \quad  \top^{\dag} =  (x=x), \quad \bot^{\dag} = \neg(x=x), \\
& ^\dag\text{commutes with the Booleans (changing $\sqcap$ to $\land$ and $\sqcup$ to $\lor$),}\\ 
& (\exists R.C)^\dag = \exists y \, (R(x,y) \land C^\dag[y/x]),\\ 
& (\forall R.C)^\dag = \forall y \, (R(x,y) \to C^\dag[y/x]).
\end{align*}
The \emph{extension} $C^{\Amf}$ of a concept $C$ in a structure $\Amf$ is \mbox{defined as}
$$
C^{\Amf} = \{ a\in \text{dom}(\Amf) \mid \Amf\models C^{\dag}(a)\},
$$
and the CI $C\sqsubseteq D$ is regarded as a shorthand for the FO-sentence 
$\forall x \, (C^\dag(x) \to D^\dag(x))$. We write $\Tmc\models C \sqsubseteq D$ to 
say that the CI $C\sqsubseteq D$ \emph{follows from} the TBox $\Tmc$, that is, 
$C^\Amf \subseteq D^\Amf$ holds in every model $\Amf$ of $\Tmc$. 
Concepts $C$ and $D$ are \emph{equivalent} if $\emptyset \models C\sqsubseteq D$ and $\emptyset \models D\sqsubseteq C$.

$\mathcal{ALC}$-concepts that are built from concept names using $\top$, $\sqcap$, $\sqcup$, and $\exists R.C$ only are called $\mathcal{ELU}$-\emph{concepts}; $\mathcal{ELU}$-concepts without $\sqcup$ are called $\mathcal{EL}$-\emph{concepts}. The FO-translation $C^\dag$ of any $\mathcal{ELU}$-concept $C$ is clearly 
a positive existential formula.

\begin{definition}[\bf Horn $\mathcal{ALC}$-concept]\label{defHornALC}\em 
We define \emph{\hALC$[\tau]$-concepts}, $H$, by the grammar
$$
H,H' ~::=~ \bot \mid \top \mid A \mid H \sqcap H' \mid L \to H \mid \exists R . H \mid \forall R.  H,
$$
where $L$ is an $\mathcal{ELU}[\tau]$-concept. 
A \emph{\hALC-CI} takes the form $L \sqsubseteq H$.
A \emph{\hALC-TBox} is a finite set of \hALC-CIs.
\end{definition} 

Our definition of \hALC-concepts is
from~\cite{DBLP:journals/lmcs/LutzW17}. We show in
the appendix that both polarity-based definition of \hALC-concepts
from~\cite{DBLP:conf/ijcai/HustadtMS05} and Horn \emph{modal} formulas
(appropriately adapted to the DL syntax)
defined in~\cite{Sturm00} are equivalent to our definition.

To put \hALC-concepts into the context of classical Horn FO-formulas,
we recall that a \emph{basic Horn formula} is a disjunction
$\varphi_{1}\vee \dots \vee \varphi_{n}$ of FO-formulas, with at most
one of them being an atom and the remaining ones  negations of
atoms~\cite{ChangKeisler}. A \emph{Horn formula} is constructed from
basic Horn formulas using $\wedge$, $\exists$, and $\forall$. 

\begin{theorem}\label{thm:11}
$(i)$ Every \hALC{}-concept is equivalent to a Horn formula with one free variable
and every \hALC-CI is equivalent to a Horn sentence.

$(ii)$ There exists an $\mathcal{ALC}$-concept \textup{(}TBox\textup{)} that is equivalent to a Horn formula \textup{(}Horn sentence\textup{)}, but not equivalent to any \hALC-concept \textup{(}TBox\textup{)}.
\end{theorem}
\begin{proof}
$(i)$ is proved by a straightforward induction on the construction of \hALC-concepts. To prove $(ii)$, consider first the $\mathcal{ALC}$-concept 
$$
C_{\nabla} = (\exists R.\top \sqcap \forall R.A) \rightarrow B.
$$
It is not hard to check that $C_{\nabla}^\dag$ has the same models as 
$$
\exists y\, R(x,y) \to \exists z\, (R(x,z) \land (A(z) \to B(x))),
$$
which is equivalent to a Horn formula. 
Example~\ref{exam:1} below shows that $C_{\nabla}$ is not equivalent to any \hALC{}-concept.

Next, consider the $\mathcal{ALC}$-TBox $\Tmc_{\textit{horn}}$ with the following CIs:
\begin{align*}
& E \sqsubseteq A_1 \sqcup A_2 \sqcup \exists R. (\neg B_1 \sqcap \neg B_2), \quad \exists R. (B_1\sqcap B_2) \sqsubseteq \bot, \\
& E \sqsubseteq \exists R. \top, \quad \exists R. B_1 \sqsubseteq \exists R. B_2, \quad \exists R. B_2 \sqsubseteq \exists R. B_1.
\end{align*}
The FO-translations of all of them but the first one are obviously (equivalent to) Horn sentences. We take a conjunction of these translations together with the sentence  
\begin{multline*}
\forall x\, \big[E(x) \to \exists y\, ( R(x,y) \land{}\\
(B_1(y) \to A_1(x)) \land (B_2(y) \to A_2(x)))\big],
\end{multline*}
which is also equivalent to a Horn one. One can check that the resulting sentence is equivalent to $\Tmc_{\textit{horn}}$. 
On the other hand, Example~\ref{exam:1} below shows that $\Tmc_{\textit{horn}}$ is not equivalent to any \hALC{}-TBox.
\end{proof}
Given Theorem~\ref{thm:11}, a natural question arises whether it is possible to design a syntactic extension of \hALC{} that captures the intersection of 
$\mathcal{ALC}$ and Horn FO. We discuss this problem in the conclusion of this paper.


We remind the reader of two usual operations on structures. The \emph{product} $\prod_{i\in I}\Amf_{i}$ of a family of $\tau$-structures $\Amf_{i}$, $i\in I$, is defined as follows: its domain $\text{dom}(\prod_{i\in I}\Amf_{i})$ is the set of functions $f\colon I\rightarrow \bigcup_{i\in I}\text{dom}(\Amf_{i})$ with $f(i)\in \text{dom}(\Amf_{i})$, for $i\in I$, and 
\begin{align*}
& A^{\prod_{i\in I}\Amf_{i}} = \{ f\in \text{dom}(\prod_{i\in I}\Amf_{i}) \mid \forall i\in I\,f(i)\in A^{\Amf_{i}}\},\\
& R^{\prod_{i\in I}\Amf_{i}} = \{ (f,g)\in (\text{dom}(\prod_{i\in I}\Amf_{i}))^{2} \mid{}\\[-4pt] 
& \hspace*{4.9cm} \forall i\in I\,(f(i),g(i))\in R^{\Amf_{i}}\}.
\end{align*}
Horn formulas are \emph{preserved under products} in the sense that 
$$
\forall i\in I\, \Amf_{i}\models \varphi(f_{1}(i),\ldots,f_{n}(i)) \ \Rightarrow \ 
\prod_{i\in I}\Amf_{i} \models \varphi(f_{1},\ldots,f_{n})
$$
for all Horn formulas $\varphi$. Note that an FO-formula is equivalent to a Horn formula iff it is preserved under the more general \emph{reduced} products (modulo filters over $I$)~\cite{ChangKeisler}.

The \emph{disjoint union} $\Amf$ of a family $\Amf_{i}$, $i\in I$, of structures has domain
$\bigcup_{i\in I}\text{dom}(\Amf_{i})\times \{i\}$ and 
\begin{align*}
& A^{\Amf} = \{ (a,i) \mid a\in A^{\Amf_{i}}\},\\
& R^{\Amf} = \{ ((a,i),(b,i)) \mid (a,b)\in R^{\Amf_{i}}\}.
\end{align*}
$\mathcal{ALC}$-TBoxes $\Tmc$ are \emph{invariant under disjoint
unions}, that is:
$$
\forall i \in I\ \Amf_{i}\models \Tmc \quad \Leftrightarrow \quad \Amf\models \Tmc.
$$
%
By the \emph{depth} of a concept $C$ we mean the maximal number of nestings of $\exists R$ and $\forall R$ in $C$.
For example, the concepts $\exists R.\exists R.A$ and $\exists R.\forall R.A$ are of depth 2.
The \emph{depth} of a TBox is the maximum over the depths of the concepts occurring in it.
By a \emph{pointed structure} we mean 
a pair $\Amf,X$ with a structure $\Amf$ and a non-empty set 
$X\subseteq \text{dom}(\Amf)$. If $X=\{a\}$, we simply write $\Amf,a$.
\begin{definition}[\bf DL indistinguishability]\em
For any DL $\Lmc$, $\tau$-structures $\Amf$ and $\Bmf$, $a\in \text{dom}(\Amf)$, $X\subseteq \text{dom}(\Amf)$, $b\in \text{dom}(\Bmf)$, and $\ell<\omega$, we write:
\begin{itemize}
\item[--] $\Amf,X \leq_{\mathcal{\Lmc}}^{(\ell)} \Bmf,b$ if $X\subseteq C^{\Amf}$ implies $b\in C^{\Bmf}$, for any $\Lmc$-concept $C$ (of depth $\leq \ell$); 


\item[--] $\Amf,a \equiv_{\Lmc}^{(\ell)} \Bmf,b$ if $\Amf,a\leq^{(\ell)}_{\Lmc}\Bmf,b$ and $\Bmf,b\leq^{(\ell)}_{\Lmc}\Amf,a$;

\item[--] $\Amf\leq_{\Lmc}\Bmf$ if $\Amf\models C\sqsubseteq D$ implies $\Bmf\models C\sqsubseteq D$, for any $\Lmc$-CI $C \sqsubseteq D$;

\item[--] $\Amf\equiv_{\Lmc}\Bmf$ if $\Amf \leq_{\Lmc}\Bmf$ and $\Bmf\leq_{\Lmc}\Amf$. 
 \end{itemize}
\end{definition}
%
%
%
%
We now recall the model comparison games for indistinguishability in $\mathcal{ELU}$ that will be 
required as part of the model comparison games for \hALC{}. The $\mathcal{ELU}$ case is rather straightforward 
and folklore~\cite{LuWo09}, but 
it will remind the reader of the basics of model comparison games used
in this paper. 

\begin{definition}[\bf simulation]\label{def:sim}\em 
A relation $S\subseteq \text{dom}(\Amf)\times \text{dom}(\Bmf)$ is a \emph{simulation} between 
$\tau$-structures $\Amf$ and $\Bmf$ if the following conditions hold:
\begin{description}
\item[\hspace*{-1.5mm}(\atomf)] for any $A\in \tau$, if $(a,b)\in S$ and $a\in A^\Amf$, then $b\in A^\Bmf$,
\item[\hspace*{-1.5mm}(\forth)] for any $R\in \tau$, if $(a,b)\in S$ and $(a,a')\in R^{\Amf}$, then there is 
$b'$ with $(b,b')\in R^{\Bmf}$ and $(a',b')\in S$.
\end{description}
We write $\Amf,a\preceq_{\textit{sim}}\Bmf,b$ if there
exists a simulation $S$ between $\Amf$ and $\Bmf$ with $(a,b)\in S$.
\end{definition}
Simulations can be equivalently described as games between two players on the disjoint union 
of $\Amf$ and $\Bmf$. A position in the \emph{simulation game} is a pair of nodes 
$(a,b)\in \text{dom}(\Amf)\times \text{dom}(\Bmf)$, marked by pebbles. The players move, in turns,   the pebbles
along binary relations in $\tau$. The first player chooses an \mbox{$R\in \tau$} and moves
the pebble in $\Amf$ along $R^{\Amf}$, the second player must respond in $\Bmf$ complying with 
(\atomf) and (\forth). The second player wins a game starting at $(a,b)$ if she can
always respond to the first player's moves, \emph{ad infinitum}. One can show that the second player 
has a winning strategy iff $\Amf,a\preceq_{\textit{sim}}\Bmf,b$.

Besides the infinitary simulation games corresponding to Definition~\ref{def:sim}, we consider games with a fixed number $\ell$ of moves. 
We write $\Amf,a\preceq_{\textit{sim}}^{\ell}\Bmf,b$ if the second player has a winning strategy in the simulation game with $\ell$ rounds 
starting from $(a,b)$.
We write $\Amf,a\preceq_{\textit{sim}}^{\omega}\Bmf,b$ if $\Amf,a\preceq_{\textit{sim}}^{\ell}\Bmf,b$ for every $\ell < \omega$.

\begin{theorem}[\bf Ehrenfeucht-Fra\"{i}ss\'e game for $\mathcal{ELU}$]\label{thm:elu}
For any finite vocabulary $\tau$, pointed $\tau$-structures $\Amf,a$ and $\Bmf,b$, and any $\ell < \omega$, we have
$$
\Amf,a \leq_{\mathcal{ELU}}^{\ell} \Bmf,b \quad \text{iff} \quad 
\Amf,a \preceq_{\textit{sim}}^{\ell}\Bmf,b.
$$
Thus, $\Amf,a \leq_{\mathcal{ELU}} \Bmf,b$ iff $\Amf,a \preceq_{\textit{sim}}^{\omega}\Bmf,b$. 
If $\Amf$ and $\Bmf$ are finite, then
$$
\Amf,a \leq_{\mathcal{ELU}} \Bmf,b \quad \text{iff} \quad 
\Amf,a \preceq_{\textit{sim}}\Bmf,b.
$$
\end{theorem}
%


In some proofs, we shall also be using bisimulations. Recall that a relation $S$ between $\Amf$ and $\Bmf$ is a \emph{bisimulation} if $S$ is a simulation between $\Amf$ and $\Bmf$, and its inverse is a simulation between $\Bmf$ and $\Amf$. The notion of \emph{$\ell$-bisimilarity} is defined by restricting the corresponding \emph{bisimulation game} to $\ell$ moves. 
This notion characterizes indistinguishability in $\mathcal{ALC}$:
pointed structures $\Amf,a$ and $\Bmf,b$ are $\ell$-bisimilar iff
$\Amf,a \equiv_{\mathcal{ALC}}^{\ell}
\Bmf,b$~\cite{DBLP:books/daglib/p/Gradel014}.

%



\section{Simulations for \hALC{}}

We now define a new model comparison game, the \emph{Horn simulation
game}, and prove that it provides an Ehrenfeucht-Fra\"{i}ss\'e
characterization of the relation $\Amf,a \leq_{\hALC}\Bmf,b$. 
As \hALC{} is not closed under negation, Horn simulations will be
non-symmetric. Moreover, since \hALC{} is not closed under
disjunction, Horn simulations relate non-empty \emph{sets} of elements
from $\text{dom}(\Amf)$ with elements from $\text{dom}(\Bmf)$.  hus,
rather than characterizing $\Amf,a \leq_{\hALC} \Bmf,b$ only, we
actually characterize the relation $\Amf,X \leq_{\hALC} \Bmf,b$.  To
define the relation between subsets of $\text{dom}(\Amf)$ along which
the pebble is moved in the Horn simulation game, we set $X
R^{\uparrow}Y$, for a binary relation $R$ and sets $X,Y$, if for any
$a\in X$, there exists $b\in Y$ with $(a,b)\in R$, and we  set $X
R^{\downarrow}Y$ if, for any $b\in Y$, there exists $a\in X$ with
$(a,b)\in R$.

 \begin{definition}[\bf Horn simulation]\label{def:hornsim}\em
 A   \emph{Horn simulation} between $\tau$-structures $\Amf$ and $\Bmf$ is a relation $Z \subseteq
  \mathcal{P}(\text{dom}(\Amf))\times \text{dom}(\Bmf)$ such that $(X,b) \in Z$ implies  $X\not=\emptyset$
  and the following hold:
\begin{description}
\item[\hspace*{-1.5mm}(\atomh)] for any $A\in \tau$, if $(X,b) \in Z$ and $X\subseteq A^{\Amf}$, then $b\in A^{\Bmf}$;
 
\item[\hspace*{-1.5mm}(\forthh)] for any $R\in \tau$, if $(X,b) \in Z$ and $X R^{\Amf\uparrow} Y$, then there exist $Y' \subseteq Y$ and $b'\in \text{dom}(\Bmf)$ with $(b,b')\in R^{\Bmf}$ and $(Y',b') \in Z$;
       
\item[\hspace*{-1.5mm}(\backh)] for any $R\in \tau$, if $(X,b) \in Z$ and $(b,b')\in R^{\Bmf}$,
    then there is $Y \subseteq \text{dom}(\Amf)$ with $X R^{\Amf\downarrow} Y$ and $(Y,b') \in Z$;
   
\item[\hspace*{-1.5mm}(\emph{sim})] if $(X,b) \in Z$, then $\Bmf,b\preceq_{\textit{sim}}\Amf,a$ for every $a\in X$.
\end{description}
We write $\Amf,X \preceq_{\textit{horn}}\Bmf,b$ if there exists a Horn simulation $Z$ between $\Amf$ and $\Bmf$ such that $(X,b) \in Z$.
\end{definition}
Condition (\atomh) ensures that concept names are preserved under Horn simulations, and conditions (\forthh), (\backh), and (\emph{sim}) ensure, recursively, the preservation of concepts of the form $\exists R.H$, $\forall R.H$ and $L\rightarrow H$, respectively. Note that (\emph{sim}) implies that the converse of (\atomh) holds as well, and so $\Amf,X \preceq_{\textit{horn}}\Bmf,b$ entails $X \subseteq A^{\Amf}$ iff $b\in A^{\Bmf}$, for all $A\in \tau$, which reflects that $A\rightarrow \bot$ is a \hALC-concept.

As we intend Horn simulations to characterize \hALC-concepts, which are $\mathcal{ALC}$-concepts, Horn simulations should subsume bisimulations. The following lemma states that this is indeed the case. It also shows that having \emph{sets} as the first component of positions in the Horn simulation games is the defining difference between Horn simulations and bisimulations. The (straightforward) proof is instructive to understand Horn simulations.
\begin{lemma}\label{lem:hornbisim}
$(i)$  If $Z$ is a bisimulation between $\tau$-structures $\Amf$ and $\Bmf$, then 
$\{(\{a\},b) \mid (a,b)\in Z\}$ is a Horn simulation between $\Amf$ and $\Bmf$.
$(ii)$ Conversely, if $Z$ is a Horn simulation with a singleton $X$ in every $(X,b)\in Z$, then
$\{(a,b) \mid (\{a\},b) \in Z\}$ is a bisimulation between $\Amf$ and $\Bmf$. 
\end{lemma}

As the FO-translations of \hALC-concepts are Horn
formulas, and the Horn formulas are (almost) characterized as the fragment of FO preserved under products, 
one could expect products to be closely related to Horn simulations. We now show this to be the case. Consider a family of $\tau$-structures $\Amf_{i}$, $i\in I$, and let $\Amf$ be the disjoint union of the $\Amf_{i}$.  Define a relation $Z$ between $\mathcal{P}(\text{dom}(\Amf))$ and $\prod_{i\in I}\Amf_{i}$ by setting $(Y,f)\in Z$ if $Y\subseteq \text{dom}(\Amf)$, $f\in \text{dom}(\prod_{i\in I}\Amf_{i})$, and $\text{dom}(\Amf_{i})\cap Y=\{f(i)\}$ for all $i\in I$. 
The proof of the following is again straightforward and instructive. 
\begin{lemma}\label{lem:product}
$Z$ is a Horn simulation between $\Amf$ and $\prod_{i\in I}\Amf_{i}$.
\end{lemma}
The following examples illustrate that Horn simulations can be seen as
a proper generalization of both bisimulations and products.

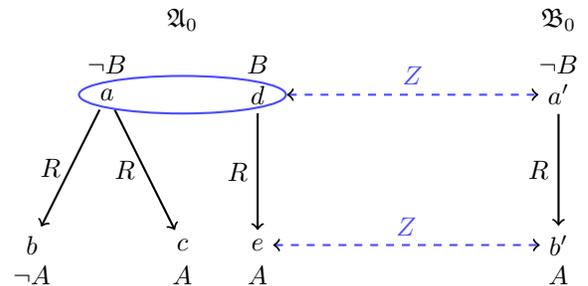
\begin{tikzpicture}
  [ every circle node/.style={draw, fill=black, inner sep=0pt,
    minimum size=.1cm},%
  xscale= 1,%
  yscale= 1,%
  ]

  \node (a) at (0,0) {$a$};
  \draw (a) ++(-90:2cm) ++(-180:1cm) node (b) {$b$};
  \draw (a) ++(-90:2cm) ++(0:1cm) node (c) {$c$};
  \node (d) [node distance=2cm, right of=a] {$d$};
  \node (e) [node distance=2cm, below of=d] {$e$};

  \node (aLabel) [above of=a, node distance=.4cm] {$\neg B$};
  \node (dLabel) [above of=d, node distance=.4cm] {$B$};
  \node [below of=b, node distance=.4cm] {$\neg A$};
  \node [below of=c, node distance=.4cm] {$A$};
  \node [below of=e, node distance=.4cm] {$A$};

  \draw[->, thick] (a) to node [left] {$R$} (b);
  \draw[->, thick] (a) to node [left] {$R$} (c);
  \draw[->, thick] (d) to node [left] {$R$} (e);
  
  \draw (a) ++(90:1cm) ++(0:1cm) node {$\Amf_0$};

  \node (a1) [node distance=6cm, right of=a] {$a'$};
  \node (b1) [node distance=2cm, below of=a1] {$b'$};;

  \node (a1Label) [above of=a1, node distance=.4cm] {$\neg B$};
  \node (b1Label) [below of=b1, node distance=.4cm] {$A$};

  \draw[->, thick] (a1) to node [left] {$R$} (b1);

  \node [node distance=1cm, above of=a1] {$\Bmf_{0}$};

  \node (fit) [ellipse, minimum width=2.75cm, minimum height=.5cm,
  draw=blue!70!white, thick] at (1,0) {};
  
  \draw [<->, thick, dashed, draw=blue!70!white] (fit) to
  node[above] {\color{blue!70!white}$Z$} (a1);

  \draw [<->, thick, dashed, draw=blue!70!white] (e) to
  node[above] {\color{blue!70!white}$Z$} (b1);

\end{tikzpicture}

\begin{example}\label{exam:1}\em
(\emph{i}) Let $\Amf_{0}$ and $\Bmf_{0}$ be the structures shown above.
$Z= \{(\{a,d\},a'),(\{e\},b')\}$ is a Horn simulation between $\Amf_{0}$ and $\Bmf_{0}$.
For the concept $C_{\nabla}$ from Theorem~\ref{thm:11}, we have $\{a,d\}\subseteq C_{\nabla}^{\Amf_{0}}$ 
but $a'\not\in C_{\nabla}^{\Bmf_{0}}$. Thus, by Lemma~\ref{lem:diagram} below, $C_{\nabla}$ is not equivalent 
to any \hALC{}-concept.


(\emph{ii}) Let $\Amf_{1}$ and $\Bmf_{1}$ be the structures below. 
Then the relation 
$$
Z= \{(\{a_1,a_2\},b),(\{c_1\},e),(\{c_2\},e),(\{d_1\},f),(\{d_2\},f)\} 
$$
is a surjective Horn simulation between $\Amf_{1}$ and $\Bmf_{1}$.
Moreover, we have $\Amf_{1}\models \Tmc_{\textit{horn}}$ but
$\Bmf_{1}\not\models \Tmc_{\textit{horn}}$ for the TBox $\Tmc_{\textit{horn}}$ from Theorem~\ref{thm:11}.
By Theorem~\ref{thm:ehrenhorn-tbox} below, $\Tmc_{\textit{horn}}$ is not equivalent to any \hALC{}-TBox.

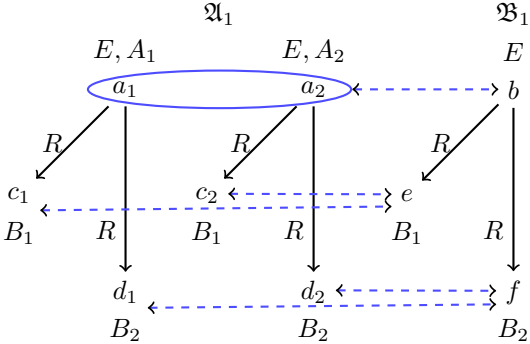
\begin{tikzpicture}
  [ every circle node/.style={draw, fill=black, inner sep=0pt,
    minimum size=.1cm},%
  xscale= 1,%
  yscale= 1,%
  ]

  \node (a1)  at (0,0) {$a_1$};
  \draw (a1) ++(-135:2cm) node (c1) {$c_1$};
  \draw (a1) ++(-90:2.7cm) node (d1) {$d_1$};

  \node [node distance=.5cm, above of=a1] {$E,A_1$};
  \node [node distance=.5cm, below of=c1] {$B_1$};
  \node [node distance=.5cm, below of=d1] {$B_2$};

  \draw [->, thick] (a1) to node[left] {$R$} (c1);
  \draw [->, thick] (a1) to node[left, near end] {$R$} (d1);

  \node (a2) [node distance=2.5cm, right of=a1] {$a_2$};
  \draw (a2) ++(-135:2cm) node (c2) {$c_2$};
  \draw (a2) ++(-90:2.7cm) node (d2) {$d_2$};

  \node [node distance=.5cm, above of=a2] {$E,A_2$};
  \node [node distance=.5cm, below of=c2] {$B_1$};
  \node [node distance=.5cm, below of=d2] {$B_2$};

  \draw [->, thick] (a2) to node[left] {$R$} (c2);
  \draw [->, thick] (a2) to node[left, near end] {$R$} (d2);

  \node at (1.25, 1) {$\Amf_1$};

  \node (b) [node distance=.3\linewidth, right of=a2] {$b$};
  \draw (b) ++(-135:2cm) node (e) {$e$};
  \draw (b) ++(-90:2.7cm) node (f) {$f$};

  \draw [->, thick] (b) to node[left] {$R$} (e);
  \draw [->, thick] (b) to node[left,near end] {$R$} (f);

  \node [node distance=.5cm, above of=b] {$E$};
  \node [node distance=.5cm, below of=e] {$B_1$};
  \node [node distance=.5cm, below of=f] {$B_2$};

  \node [above of=b, node distance=1cm] {$\Bmf_1$};

  \node (fit1) [ellipse, minimum width=3.5cm, minimum height=.5cm,
  draw=blue!70!white, thick] at (1.25,0) {};

  
  
  \draw [<->, thick, dashed, draw=blue!70!white] (fit1) -- (b);
  


  \draw [<->, thick, dashed, draw=blue!70!white] (c1.-35) to (e.-145);

  \draw [<->, thick, dashed, draw=blue!70!white] (c2.5) to (e.175);

  \draw [<->, thick, dashed, draw=blue!70!white] (d1.-35) to (f.-145);

  \draw [<->, thick, dashed, draw=blue!70!white] (d2.5) to (f.175);

\end{tikzpicture}
\end{example}

Given the notion of Horn simulation, we next introduce the \emph{Horn simulation game} between $\tau$-structures $\Amf$ and $\Bmf$
in the expected way. In the infinite case, it consists of the following nested games. Using simulation games, one can check
whether condition (\emph{sim}) holds for a pair $(X,b)$. Then, in the main Horn simulation game, the
second player must respond with pairs $(Y',b')$ for (\forthh) and sets $Y$ for (\backh) 
such that the new position satisfies ($\emph{sim}$) and the remaining conditions of Definition~\ref{def:hornsim}. 
In the Horn simulation game with $\ell$ rounds, some care must be taken: as we want to characterize the 
depth $\ell$ fragment of \hALC, 
we have to decompose condition (\emph{sim}). 
Thus, define pairs $(X,b)$ satisfying (\emph{sim}$^{\ell}$)
as those for which player 2 has a winning strategy for the $\ell$-round simulation game for all 
$(b,a)$ with $a\in X$.
Then, inductively, player 2 has a winning strategy in the $(\ell+1)$-round Horn simulation game at 
position $(X,b)$ if $(X,b)$ satisfies (\emph{sim}$^{\ell+1}$) and player 2 can
react to player 1's first move by choosing $Y$ in such a way that condition (\emph{sim}$^{\ell}$) 
holds for the resulting position 
and she has a winning strategy in the resulting $\ell$-round game. A formal definition is given in the appendix. 
We write $\Amf,X \preceq_{\textit{horn}}^{\ell}\Bmf,b$ if player 2 has a winning strategy in the $\ell$-round game.
\begin{theorem}[\bf Ehrenfeucht-Fra\"{i}ss\'e game for \hALC]\label{thm:ehrenhorn}
For any finite vocabulary $\tau$, pointed $\tau$-structures $\Amf,a$ and $\Bmf,b$,
and any $\ell < \omega$, we have
$$
\Amf,a \leq_{\textit{horn}\mathcal{ALC}}^{\ell} \Bmf,b \quad \text{iff} \quad 
\Amf,a \preceq_{\textit{horn}}^{\ell}\Bmf,b.
$$
Thus, $\Amf,a \leq_{\textit{horn}\mathcal{ALC}} \Bmf,b$ iff $\Amf,a \preceq_{\textit{horn}}^{\omega}\Bmf,b$. 
If $\Amf$ and $\Bmf$ are finite, then
$$
\Amf,a \leq_{\textit{horn}\mathcal{ALC}} \Bmf,b \quad \text{iff} \quad \Amf,a \preceq_{\textit{horn}}\Bmf,b.
$$
\end{theorem}
As discussed above, to prove Theorem~\ref{thm:ehrenhorn}, we actually show the following
stronger statement:
\begin{lemma}\label{lem:diagram}
For any finite vocabulary $\tau$, pointed $\tau$-structures $\Amf,X$ and $\Bmf,b$,
and any $\ell < \omega$,
$$
\Amf,X \leq_{\textit{horn}\mathcal{ALC}}^{\ell} \Bmf,b \quad \text{iff} \quad 
\text{$\exists X_{0}\subseteq X\ \Amf,X_{0} \preceq_{\textit{horn}}^{\ell}\Bmf,b$}.
$$
If $\Amf$ and $\Bmf$ are finite, then
$$
\Amf,X \leq_{\textit{horn}\mathcal{ALC}} \Bmf,b \quad \text{iff} \quad 
\text{$\exists X_{0}\subseteq X\ \Amf,X_{0} \preceq_{\textit{horn}}\Bmf,b$}.
$$
\end{lemma}
\begin{proof} (sketch)
The second claim follows directly from the first one, which we prove here. 
For any $\ell < \omega$ and any pointed $\tau$-structure $\Amf,a$, 
let $\lambda_{\Amf,\ell,a}$ be an $\mathcal{ELU}[\tau]$-concept of depth $\leq \ell$ such that, 
for any pointed $\tau$-structure $\Bmf,b$,
\begin{equation}\label{eq:1}
b \in \lambda_{\Amf,\ell,a}^{\Bmf} \quad \text{iff} \quad \Amf,a\preceq_{\textit{sim}}^{\ell}\Bmf,b.
\end{equation}
The existence of $\lambda_{\Amf,\ell,a}$ follows immediately from the fact that 
there are only finitely-many non-equivalent $\mathcal{ELU}[\tau]$-concepts of any fixed depth $\ell$. 
Similarly, fix a finite set Horn$_{\ell}$ of \hALC-concepts of depth $\leq \ell$ such that every \hALC-concept
of depth $\leq \ell$ is equivalent to some concept in Horn$_{\ell}$.
For a pointed $\tau$-structure $\Amf,X$, let $\rho_{\Amf,\ell,X}$ be
the conjunction of all concepts $C$ in Horn$_{\ell}$ with $X\subseteq C^{\Amf}$.
Clearly, we have
\begin{equation}
b\in \rho_{\Amf,\ell,X}^{\Bmf} \quad \text{iff} \quad 
\Amf,X \leq_{\textit{horn}\mathcal{ALC}}^{\ell} \Bmf,b.
\end{equation}

To prove the implication $(\Rightarrow)$ of the first claim, we 
define relations $Z_{\ell} \subseteq \mathcal{P}(\text{dom}(\Amf)) \times \text{dom}(\Bmf)$, $\ell < \omega$, by setting 
$(X,b)\in Z_{\ell}$ if $X\not=\emptyset$ and the following two conditions hold:
\begin{itemize}
\item[$(i)$] $b\in \rho_{\Amf,\ell,X}^{\Bmf}$,
\item[$(ii)$] $X\subseteq \lambda_{\Bmf,\ell,b}^{\Amf}$.
\end{itemize}
\smallskip
\noindent
\emph{Claim}~1. For any $\ell<\omega$, $\emptyset \ne X\subseteq \text{dom}(\Amf)$ and $b\in \text{dom}(\Bmf)$, 
if $(X,b)\in Z_{\ell}$, then $\Amf,X \preceq_{\textit{horn}}^{\ell}\Bmf,b$.

\smallskip
\noindent
\emph{Proof of claim}. We proceed by induction on $\ell<\omega$. The basis $\ell=0$ holds by definition. So suppose that Claim~1 has been proved for $\ell$ and that $(X,b)\in Z_{\ell+1}$. We show $\Amf,X \preceq_{\textit{horn}}^{\ell+1}\Bmf,b$. Condition~(\atomh) holds by definition. For (\forthh), suppose player~1 moves the pebble to $Y$ with $X R^{\Amf\uparrow}Y$. Then  $X\subseteq (\exists R.\rho_{\Amf,\ell,Y})^{\Amf}$. By the definition of $Z_{\ell+1}$, $b\in (\exists R.\rho_{\Amf,\ell,Y})^{\Bmf}$. Let player~2 respond with $b'$ and $Y'$ such that $(b,b')\in R^{\Bmf}$, 
$b'\in \rho_{\Amf,\ell,Y}^{\Bmf}$, and 
$$
Y'= Y \cap \lambda_{\Bmf,\ell,b'}^{\Amf}.
$$
We show that $(Y',b')$ is as required for (\forthh). By IH, it suffices to prove that $(Y',b')\in Z_{\ell}$. To show $Y'\ne\emptyset$, suppose otherwise. Then $Y\subseteq (\lambda_{\Bmf,\ell,b'}\rightarrow \bot)^{\Amf}$, so $(\lambda_{\Bmf,\ell,b'}\rightarrow \bot)$ is equivalent to a conjunct of $\rho_{\Amf,\ell,Y}$. By the construction of $b'$, we have $b'\in (\lambda_{\Bmf,\ell,b'}\rightarrow \bot)^{\Bmf}$. On the other hand, $b'\in \lambda_{\Bmf,\ell,b'}^{\Bmf}$, which is impossible.

Condition $(i)$ is proved similarly and condition $(ii)$ holds by the definition of $Y'$.

For (\backh), suppose player 1 moves the pebble to $b'$ with $(b,b')\in R^{\Bmf}$. For every $C \in \text{Horn}_{\ell}$ with $b'\not\in C^{\Bmf}$, take some $a_{C}\in X$ and $a_{C}'$ with $(a_{C},a_{C}')\in R^{\Amf}$ such that $a_{C}'\in (\lambda_{\Bmf,\ell,b'}\sqcap \neg C)^{\Amf}$. They exist since otherwise we would have $X\subseteq (\forall R.(\lambda_{\Bmf,\ell,b'}\rightarrow C))^{\Amf}$ but \mbox{$b\not\in(\forall R.(\lambda_{\Bmf,\ell,b'}\rightarrow C))^{\Bmf}$}, which contradicts the definition of $Z_{\ell+1}$. Now let player~2 respond with the set $Y$ of all such $a_{C}'$. Then $X R^{\Amf\downarrow} Y$ and $(Y,b')\in Z_{\ell}$ as $b'\in \rho_{\Amf,\ell,Y}^{\Bmf}$ and $Y\subseteq \lambda_{\Bmf,\ell,b'}^{\Bmf}$ hold by the construction of $Y$. By IH, $\Amf,Y \preceq_{\textit{horn}}^{\ell}\Bmf,b'$, as required.

Finally,~(\emph{sim}$^{\ell+1}$) follows from~\eqref{eq:1}, which completes the proof of Claim~1.

\smallskip

Now assume that $\Amf,X \leq_{\textit{horn}\mathcal{ALC}}^{\ell} \Bmf,b$. Then it suffices to
prove that if $b\in \rho_{\Amf,\ell,X}^{\Bmf}$, then there exists $X_{0}\subseteq X$ with
$(X_{0},b)\in Z_{\ell}$. But for $X_{0}= X \cap \lambda_{\Bmf,\ell,b}^{\Amf}$ this can be proved 
in the same way as Claim~1 above.

The proof of the implication $(\Leftarrow)$ of the first claim is by induction on $\ell<\omega$.
\end{proof}
Horn simulations can also characterize \hALC{}-TBoxes. 
For $\ell <\omega$, we write $\Amf \preceq_{\textit{horn}}^{(\ell)}\Bmf$ if, for every $b\in \text{dom}(\Bmf)$, 
there exists $X\subseteq \text{dom}(\Amf)$ such that $\Amf,X \preceq_{\textit{horn}}^{(\ell)}\Bmf,b$.
\begin{theorem}\label{thm:ehrenhorn-tbox}
For any finite vocabulary $\tau$, $\tau$-structures $\Amf$ and $\Bmf$,
and $\ell < \omega$, we have
$$
\Amf \leq_{\hALC}^{\ell} \Bmf \quad \text{iff} \quad 
\Amf \preceq_{\textit{horn}}^{\ell}\Bmf.
$$
If $\Amf$ and $\Bmf$ are finite, then
$$
\Amf \leq_{\hALC} \Bmf \quad \text{iff} \quad \Amf \preceq_{\textit{horn}}\Bmf.
$$
\end{theorem}
\section{Complexity of Model Indistinguishability}\label{sec:complexity}

We next study the complexity of deciding the relations $\leq_{\hALC}$
and $\equiv_{\hALC}$ and their restrictions $\leq^{\ell}_{\hALC}$ and
$\equiv^{\ell}_{\hALC}$ on the level of concepts. The related problems
\emph{on the TBox level} (cf.\ Theorem~\ref{thm:ehrenhorn-tbox}) have the
same complexity as shown in the appendix. We refer to the respective decision
problems as \emph{$(\ell)$-entailment} and
\emph{$(\ell)$-equivalence}; for instance, entailment is the problem
of deciding whether $\Amf,a\leq_{\hALC} \Bmf,b$ for input
$\Amf,\Bmf,a,b$, and $\ell$-equivalence is the problem of deciding
whether $\Amf,a\equiv_{\hALC}^{\ell} \Bmf,b$ for input
$\Amf,\Bmf,a,b,\ell$.

As a second application of the Ehrenfeucht-Fra\"{i}ss\'e
results we investigate \emph{concept learning by
example \textup{(}CBE\textup{)}}. CBE is a supervised learning
problem with applications in knowledge engineering for automatically deriving new and potentially interesting concept
descriptions from labelled data.  Intuitively, given some relational
data and sets of positive and negative examples, the goal is to find
a concept that generalizes the positive examples, but avoids the
negative ones. The associated decision problem is formally
defined as follows:
\begin{description}

  \item[\hspace*{-1mm}\textbf{Input:}] structure $\Amf$, positive and negative examples $P,N$.

  \item[\hspace*{-1mm}\textbf{Question:}] \quad\, is there a \hALC-concept $C$ with
    $P\subseteq C^\Amf$ and $N\cap C^\Amf=\emptyset$?

\end{description}
The connection to Horn simulations is given by
Lemma~\ref{lem:diagram}: an input $\Amf,P,N$ is a yes-instance
of CBE iff $\Amf,P\not\preceq_{\textit{horn}}\Amf,b$ for all
$b\in N$. We denote by $\ell$-CBE the variant of CBE
restricted to \hALC-concepts of depth~$\ell$. This variant is 
important in practice as the user is interested in small
separating concepts. 

Our main result in this section is the following theorem:
\begin{theorem}\label{thm:complexity}
  Entailment, equivalence, and CBE are \ExpTime-complete. Moreover,
  $\ell$-entailment, $\ell$-equivalence and $\ell$-CBE are
  \ExpTime-complete if $\ell$ is given in binary and \PSpace-complete
  if $\ell$ is given in unary.
\end{theorem}
It is worth mentioning the striking contrast between this \ExpTime
result and the fact that the same problems for 
\ALC are in \PTime. On the
one hand, the \ExpTime lower bounds provide evidence that the use of sets in the notion of
Horn simulations is inevitable. On the other hand, observe that CBE for \ALC is
in \PTime because there is an \ALC-concept separating the
positive and negative examples iff $\Amf,a$ and $\Amf,b$ are not
bisimilar, for any $a\in P$ and $b\in N$. Consequently, the positive
examples can be treated essentially separately and a naive application
of this leads to overfitting, that is, the intended generalization of
the positive examples is not taking
place~\cite{DBLP:conf/ilp/BadeaN00}. Thus, we can regard the \ExpTime
result for CBE in \hALC{} as the price for obtaining real
generalizations. 

To prove Theorem~\ref{thm:complexity}, we focus on the
unrestricted case.  The following lemma gives complexity-theoretic
reductions between the mentioned problems and the problem
\emph{HornSim} of deciding whether
$\Amf,X\preceq_{\textit{horn}}\Bmf,b$ for input $\Amf,\Bmf,X,b$. 
%
\begin{lemma} \label{lem:reductions}
  \begin{enumerate}

    \item[(1)] $\text{CBE}\leq_T^P\text{HornSim}$;

    \item[(2)] $\overline{\text{HornSim}}\leq_m^P\text{CBE}$;

    \item[(3)] $\text{HornSim}\leq_m^P \text{Entailment}$;

    \item[(4)] $\text{Entailment}\leq_m^P \text{HornSim}$;

    \item[(5)] $\text{Equivalence}\leq_T^P \text{Entailment}$;

    \item[(6)] $\text{Entailment}\leq_m^P\text{Equivalence}$.

  \end{enumerate}
\end{lemma}

\begin{proof} Here, we only show the most interesting reduction (3). 

  Let $\Amf,\Bmf,X,b$ be the input to \textit{HornSim}. Define $\Amf'$ by
  adding a new $R$-predecessor $a$ to all nodes in $X$. Further,
  define $\Bmf'$ by taking the disjoint union of $\Amf$ and $\Bmf$ and
  adding a new $R$-predecessor $d$ to $b$, and making $d$ also a
  predecessor of all nodes in (the copy of) $X$. Then we have
  \begin{align*}
    \Amf,X\preceq_{\textit{horn}}\Bmf,b \quad\text{iff\quad}
    \Amf',a\preceq_{\textit{horn}}\Bmf',d,
  \end{align*}
  which is equivalent to $\Amf',a\leq_{\hALC} \Bmf',d$ by
  Theorem~\ref{thm:ehrenhorn}.\! 
\end{proof}
Thus, to prove Theorem~\ref{thm:complexity}, it suffices to
show the following:
\begin{lemma} \label{lem:hornsim}
  HornSim is \ExpTime-complete.
\end{lemma}
For the upper bound, we observe that the Horn
simulation game can be implemented by an alternating Turing machine
(ATM) using only polynomial space. 
%
%
%
%
For the lower bound, we carefully adapt a strategy
from~\cite{HarelKV02} for proving that the simulation problem between
two structures \Amf and \Bmf is \ExpTime-hard when \Amf is given as a
\emph{fair concurrent transition system}, that is, a certain
synchronized product of structures. More precisely, we reduce the word
problem of polynomially space-bounded ATMs. Let $M$ be an $s(n)$-space
bounded ATM and $w$ an input of length $n$. We construct structures
\Amf, \Bmf, $X\subseteq \text{dom}(\Amf)$, and $b\in \text{dom}(\Bmf)$
such that
$$M\text{ accepts }w\text{\quad
  iff\quad$\Amf,X\preceq_{\textit{horn}}\Bmf,b$.}$$
The structure $\Amf$ can be thought of as the disjoint union of $s(n)$
structures $\Amf_1,\ldots,\Amf_{s(n)}$ and a single copy of \Bmf (plus
some connections from the $\Amf_i$ to \Bmf).  Intuitively, each
sub-structure $\Amf_i$ is responsible for tape cell $i$ of one of $M$'s
configurations on input $w$; thus, the domain of each $\Amf_i$
consists of the possible contents of a single cell.  As usual, the
challenge is the synchronization. Here, different tape cells are
synchronized via the simulation conditions using different role names:
one role name $R_{q,a,i,d}$ for every possible state $q$
of $M$, current head position $i$, read symbol $a$, and branching
direction $d$ (we assume that $M$ has binary branching).  The
extension of such a role name $R_{q,a,i,d}$ in a structure $\Amf_j$
is defined in the obvious way, respecting $M$'s transition relation.
The set $X$ consists of $M$'s initial configuration on input $w$. 

The role of \Bmf (as the second structure) is to control $M$'s
computation; its domain is independent of $M$ and consists of 20
elements only. It manages both the
switch between universal and existential states and the intended
acceptance value of the current configuration ($1$ if the
configuration leads to acceptance, $0$ otherwise). Universal and
existential elements are fully interconnected with the mentioned role names
$R_{q,a,i,d}$.  The initial element $b$ for the reduction is a
universal element (without loss of generality $M$'s initial state is
universal), and corresponds to acceptance value~$1$.

Now, conditions~(\textit{atom$_h$}) and (\textit{forth$_h$}) are
responsible for simulating $M$'s computation on input $w$
and~(\textit{atom$_h$}) ensures that every reached set $X$ in $\Amf$
corresponds to a valid configuration.  Conditions~(\textit{sim})
and~(\textit{back$_h$}) do not have a real purpose for the reduction,
but need to be reflected in the mentioned inclusion of (a copy of)
$\Bmf$ in \Amf.

\section{Expressive Completeness for \hALC{}}
In this section, we prove that an FO-formula with one free variable is
equivalent to a \hALC{}-concept just in case it is preserved under
Horn simulations, and that an FO-sentence is equivalent to a
\hALC{}-TBox just in case it is invariant under disjoint unions and
preserved under \emph{global} (that is, surjective) Horn simulations.
We prove these results both in the classical setting defined above and
in the finite model theory setting where the notions of equivalence,
preservation, and invariance are relativized to finite models. 

In the concept case, by the van Benthem-Rosen characterization of $\mathcal{ALC}$-concepts
as the bisimulation invariant fragment of FO and since preservation under Horn simulations implies
invariance under bisimulations (by Lemma~\ref{lem:hornbisim}), it suffices to prove that
an $\mathcal{ALC}$-concept is equivalent to a \hALC{}-concept iff it is preserved under Horn simulations
(in the classical and finite model theory setting). We will, therefore, formulate the expressive completeness result
within $\mathcal{ALC}$. In the TBox case, it is known that an FO-sentence is equivalent to an $\mathcal{ALC}$-TBox just in case it is 
invariant under disjoint unions and preserved under bisimulations\cite{TBoxpaper,DBLP:journals/apal/Otto04}; 
thus, again by Lemma~\ref{lem:hornbisim}, it suffices to show that an $\mathcal{ALC}$-TBox is equivalent to a \hALC-TBox iff it is preserved under global Horn simulations. 
In both proofs, we employ Otto's finitary method~\cite{DBLP:journals/apal/Otto04} and show that, for $\mathcal{ALC}$-concepts 
(TBoxes) of depth $\le\ell$, preservation under (global) Horn simulations is equivalent to 
preservation under (global) $\ell$-Horn simulations, 
which is the same as equivalence to a \hALC{}-concept (TBox) of depth $\leq \ell$.

We start with the concept case. Say that an $\mathcal{ALC}$-concept $C$
is \emph{preserved under $(\ell)$-Horn simulations} if, for any pointed structures $\Amf,X$ and $\Bmf,b$, whenever 
$X\subseteq C^{\Amf}$ and $\Amf,X\preceq_{\textit{horn}}^{(\ell)}\Bmf,b$ then $b\in C^{\Bmf}$.
%
%

\begin{theorem}[\bf expressive completeness: \hALC-concepts]\label{thm:mainconcept}
Let $C$ be an $\mathcal{ALC}$-concept of depth $\ell$. Then the following
conditions are equivalent \textup{(}in the classical and finite model theory setting\textup{)}:
\begin{enumerate}
\item[(1)] $C$ is equivalent to a \hALC{}-concept,
\item[(2)] $C$ is preserved under Horn simulations,
\item[(3)] $C$ is preserved under $\ell$-Horn simulations,
\item[(4)] $C$ is equivalent to a \hALC{}-concept of depth $\leq \ell$.
\end{enumerate}
\end{theorem}
\begin{proof} (sketch)
$(1) \Rightarrow (2)$ follows from Lemma~\ref{lem:diagram}; \mbox{$(4) \Rightarrow (1)$} is trivial; (3) $\Rightarrow$ (4) is straightforward and proved in the appendix. We thus focus on
(2) $\Rightarrow$ (3).

A structure $\Amf$ is called \emph{tree-shaped} if the directed graph $G_{\Amf}=(\text{dom}(\Amf),E)$ with
$E=\bigcup_{R\in \tau}R^{\Amf}$ is a directed tree and $R^{\Amf}\cap S^{\Amf}=\emptyset$
for all distinct role names $R$ and $S$. The root of $G_{\Amf}$ is called the \emph{root} of $\Amf$.
The \emph{depth} of $a\in \text{dom}(\Amf)$ is the length of the path from the root of $\Amf$
to $a$; the root of $\Amf$ has depth $0$. The disjoint union of tree-shaped structures
is a \emph{forest}.
Recall that every pointed $\Amf,a$ can be unravelled into a tree-shaped structure
$\Amf^{\ast}$ with root $a$ such that $\Amf,a$ and $\Amf^{\ast},a$ are bisimilar~\cite{goranko20075}.
Note that $\Amf^{\ast}$ is infinite (even for finite $\Amf$) if $G_{\Amf}$ contains a cycle.
The finite model theory version of Theorem~\ref{thm:mainconcept} is
not affected as one only needs 
the unravelled tree-shaped structures up to a finite depth $\ell$.
%

Suppose $C$ is an $\mathcal{ALC}$-concept of depth $\leq \ell$
preserved under Horn simulations.
Let $\Amf,X$ and $\Bmf,b$ be pointed structures such that $\Amf,X \preceq_{\textit{horn}}^{\ell} \Bmf,b$
and $X\subseteq C^{\Amf}$. We have to show that $b\in C^{\Bmf}$. For every $a\in X$, take a tree-shaped 
pointed structure $\Amf_{a},a$ bisimilar to $\Amf,a$. Let $\Bmf',b$ be a tree-shaped pointed structure
bisimilar to $\Bmf,b$.
Then $\Amf',X\preceq_{\textit{horn}}^{\ell} \Bmf',b$ for the disjoint union $\Amf'$ of $\Amf_{a}$, $a\in X$.
By bisimulation invariance of $\mathcal{ALC}$-concepts, we have $X\subseteq C^{\Amf'}$ and it suffices 
to prove that $b\in C^{\Bmf'}$.
Remove from $\Amf'$ and $\Bmf'$ all nodes of depth $>\ell$ and denote the resulting
structures by $\Amf''$ and $\Bmf''$, respectively. 
As $C$ is of depth $\leq \ell$, we have $X\subseteq C^{\Amf''}$ and it suffices 
to prove that $b\in C^{\Bmf''}$. Using $\Amf',X \preceq_{\textit{horn}}^{\ell} \Bmf',b$, it is straightforward 
to show that $\Amf'',X\preceq_{\textit{horn}} \Bmf'',b$.
Then $b\in C^{\Bmf''}$ follows from the preservation of $C$ under Horn simulations.

For the finite model theory setting, observe that $\Amf''$ and $\Bmf''$ are finite if $\Amf$ and $\Bmf$ are finite.
\end{proof}

We now consider the TBox case.
We say that an $\mathcal{ALC}$-TBox is \emph{preserved under global ($\ell$)-Horn simulations} if 
$\Amf\models \Tmc$ and $\Amf\preceq_{\textit{horn}}^{(\ell)} \Bmf$ 
imply $\Bmf\models \Tmc$.

\begin{theorem}[\bf expressive completeness: \hALC-TBoxes]\label{thm:maintbox}
For any $\mathcal{ALC}$-TBox $\Tmc$ of depth $\ell$, the following
conditions are equivalent \textup{(}in the classical and finite model theory setting\textup{)}:
\begin{enumerate}
\item[(1)] $\Tmc$ is equivalent to a \hALC-TBox;
\item[(2)] $\Tmc$ is preserved under global Horn simulations;
\item[(3)] $\Tmc$ is preserved under global $\ell$-Horn simulations;
\item[(4)] $\Tmc$ is equivalent to a \hALC-TBox of depth $\leq \ell$.
\end{enumerate}
\end{theorem}
\begin{proof} (sketch)
We use the notation from the previous proof and focus on (2) $\Rightarrow$ (3), showing (3) $\Rightarrow$ (4) in the appendix.


We require \emph{injective $\ell$-Horn simulations}, which are defined as follows. 
Let $\Amf$ be a forest and $\Bmf$ a tree-shaped structure. A sequence $H^{0},\ldots,H^{\ell}$ of relations between $\mathcal{P}(\text{dom}(\Amf))$
and $\text{dom}(\Bmf)$ is called an \emph{injective $\ell$-Horn simulation} if for each $(X,b)\in H^{i}$ all $a\in X$ are
of depth $i$ in $\Amf$ and $b$ is of depth $i$ in $\Bmf$, and the following conditions hold:
\begin{enumerate}
\item[--] if $(X,b)\in H^{i}$, then $\Amf,X \preceq_{\textit{horn}}^{i}\Bmf,b$, for $0\leq i\leq \ell$; 
\item[--] if $(X,b)\in H^{i}$ and $X R^{\Amf\uparrow} Y$, then there are $Y' \subseteq Y$ and $b'\in\text{dom}(\Bmf)$ with 
        $(b,b')\in R^{\Bmf}$ and $(Y',b')\in H^{i+1}$, for all $R\in \tau$ and $0\leq i<\ell$;
\item[--] if $(X,b)\in H^{i}$ and $(b,b')\in R^{\Bmf}$,
    then there exists $Y \subseteq \text{dom}(\Amf)$ such that $X R^{\Amf\downarrow} Y$ and 
    $(Y',b')\in H^{i+1}$, for all $R\in \tau$ and $0\leq i<\ell$;
\item[--] if $(X_{0},b),(X_{1},b)\in H^{i}$, then $X_{0}=X_{1}$, for $0\leq i \leq \ell$.
\end{enumerate}
If $\Amf,X \preceq_{\text{horn}}^{\ell}\Bmf,b$, we can take,  
for $a\in X$, a tree-shaped pointed structure $\Amf_{a},a$ bisimilar
to $\Amf,a$ and a tree-shaped pointed structure $\Bmf',b$ bisimilar to $\Bmf,b$.
Then $\Amf',X\preceq_{\text{horn}}^{\ell} \Bmf',b$ for the disjoint union $\Amf'$ of the 
$\Amf_{a}$, $a\in X$. By duplicating successors in $\Bmf'$ sufficiently often
(possibly exponentially many times), we obtain a tree-shaped pointed structure 
$\Bmf'',b$ bisimilar to $\Bmf',b$ such that
there is an injective $\ell$-Horn simulation $H^{0},\ldots,H^{\ell}$ between $\Amf'$ and $\Bmf''$ with
$(X,b)\in H^{0}$.

Now suppose $\Tmc$ is preserved under global Horn simulations.
Let $\Bmf$ be a structure such that there exists a model $\Amf$ of $\Tmc$ with
$\Amf \preceq^{\ell}_{\textit{horn}} \Bmf$. We have to show that $\Bmf$ is a model of $\Tmc$. 
Let $b_{0}\in \text{dom}(\Bmf)$ be arbitrary. 
It suffices to show $b_{0}\in (\neg C \sqcup D)^{\Bmf}$ for all $C \sqsubseteq D\in \Tmc$.
Since $\Amf\preceq^\ell_{\textit{horn}}\Bmf$, there is a set
$X\subseteq \text{dom}(\Amf)$ with $\Amf,X\preceq_{\textit{horn}}^{\ell} \Bmf,b_{0}$. 
For every $a\in X$, take a tree-shaped pointed structure $\Amf_{a},a$ bisimilar
to $\Amf,a$. By the observation above, we can take a tree-shaped pointed interpretation 
$\Bmf',b_{0}$ bisimilar to $\Bmf,b_{0}$ and the disjoint union $\Amf'$ of the  
$\Amf_{a}$, $a\in X$, such that there is an injective $\ell$-Horn simulation $H^{0},\ldots,H^{\ell}$ between 
$\Amf'$ and $\Bmf'$ with $(X,b_{0})\in H^{0}$. By bisimulation invariance of $\mathcal{ALC}$-concepts, 
$\Amf'$ is a model of $\Tmc$, and so it is enough to show that $b_{0}\in (\neg C \sqcup D)^{\Bmf'}$ for all 
$C \sqsubseteq D\in \Tmc$.

Let $\Bmf'|_{\ell}$ be the structure obtained from $\Bmf'$
by dropping all nodes of depth $>\ell$. We hook to every leaf $b\in \Bmf'|_{\ell}$ of depth $\ell$ a 
structure $\Bmf_{b}$ so that $\Amf',X \preceq_{\textit{horn}} \Bmf'',b_{0}$ 
for the resulting structure $\Bmf''$. As $\Tmc$ is preserved under global Horn simulations, 
$\Bmf''$ is a model of $\Tmc$. As $\Tmc$
has depth $\leq \ell$, we have $b_{0}\in (\neg C \sqcup D)^{\Bmf'}$ for all $C \sqsubseteq D\in \Tmc$, as required.
We come to the construction of the $\Bmf_{b}$ for $b$ a leaf of depth $\ell$ in $\Bmf'|_{\ell}$. 
Since $H^{0},\ldots,H^{\ell}$ is injective, there is a unique non-empty $X_{b}\subseteq \text{dom}(\Amf')$
such that $(X_{b},b)\in H^{\ell}$. Observe that from $(X_{b},b)\in H^{\ell}$ it follows that $X_{b} \subseteq A^{\Amf'}$ iff 
$b\in A^{\Bmf'}$ for any $A\in \tau$.
Let $\Amf_{a}'$ be the tree-shaped substructure of $\Amf'$ rooted at $a$, for $a\in X_{b}$. 
Then we hook to $b$ the structure $\Bmf_{b}=\prod_{a\in X_{b}}\Amf_{a}'$ by identifying $(a \mid a\in X_{b})\in
\prod_{a\in X_{b}}\Amf_{a}'$ with $b$. Using Lemma~\ref{lem:product}, it is readily checked that
the resulting structure is as required. 

In the finite model theory setting, we consider finite $\Amf$ and $\Bmf$.
Then we can assume that the structures $\Bmf'|_{\ell}$ and $\Amf'|_{\ell}$ are finite. Now, rather than hooking
the (possibly infinite) $\Bmf_{b}=\prod_{a\in X_{b}}\Amf_{a}'$ to every leaf $b$ of depth $\ell$ in $\Bmf'|_{\ell}$ we (\emph{i}) replace all 
$\Amf_{a}'$ with $a$ of depth $\ell$ in $\Amf'$ by \emph{finite} models $\Amf''_{a}$ of $\Tmc$ satisfying the same subconcepts of 
$\Tmc$ as $\Amf'_{a}$ in $a$ and (\emph{ii}) hook $\prod_{a\in X_{b}}\Amf_{a}''$ to $b$. Then $\Amf'',X\preceq_{\textit{horn}}\Bmf''',b_{0}$
for the resulting finite structures $\Amf''$ and $\Bmf'''$, and $\Amf''$ is a model of $\Tmc$. 
\end{proof}
As there is only a finite number of \hALC{}-concepts and TBoxes of bounded depth in
a finite vocabulary $\tau$, it follows from Theorems~\ref{thm:mainconcept}
and~\ref{thm:maintbox} that it is decidable whether an $\mathcal{ALC}$-concept or TBox
is equivalent to a \hALC{}-concept or TBox, respectively.

\section{Horn guarded fragment \hGF{} of FO}
\label{sec:horn_GFO}
We extend \hALC{} to the Horn fragment, \hGF{}, of the guarded fragment of FO
in the obvious way. \hGF{} contains numerous popular Horn DLs including those
extending \hALC{} with inverse roles, the universal role, and role inclusions~\cite{DBLP:conf/ijcai/HustadtMS05,DBLP:journals/jar/HustadtMS07}.
We then generalize the Horn simulation games to guarded Horn simulation games for \hGF{} and prove 
an Ehrenfeucht-Fra\"{i}ss\'e type definability theorem and a van Benthem style expressive
completeness result for \hGF{}. Applications include an \ExpTime upper
bound for model indistinguishability.

Let $\tau$ be a vocabulary of predicate names $R$ of arbitrary arity $r_{R}\ge 0$. The \emph{guarded fragment} GF$[\tau]$ of FO is defined by the following rules:
\begin{itemize}
\item[--] GF$[\tau]$ contains the constants $\top$ (truth) and $\bot$ (falsehood);
\item[--] GF$[\tau]$ contains the \emph{atomic formulas} $R(\xbf)$ and $x=y$ with $R\in \tau$;
\item[--] GF$[\tau]$ is closed under the connectives $\wedge$, $\vee$, and $\neg$;
\item[--] if $\varphi(\xbf\ybf)$ is in GF$[\tau]$ with free
variables among $\xbf\ybf$ and $G(\xbf\ybf)$ is an atomic formula containing all the variables in
$\xbf\ybf$, then
$$
\forall \ybf\, (G(\xbf\ybf)\rightarrow \varphi(\xbf\ybf)),
\qquad
\exists \ybf\, (G(\xbf\ybf)\wedge \varphi(\xbf\ybf))
$$
are in GF$[\tau]$ (these are called the universal and existential \emph{guarded quantifiers} of GF$[\tau]$).
\end{itemize}
If the particular vocabulary $\tau$ is not relevant, we simply write GF for GF$[\tau]$.
The \emph{nesting depth of guarded quantifiers} in a formula $\varphi$ in GF, or simply the \emph{depth} of $\varphi$,
is defined as the number of nestings of guarded quantifiers in $\varphi$.
The formulas of the \emph{positive existential guarded fragment} GF$^{\exists}[\tau]$ of GF$[\tau]$
are constructed from atomic formulas using $\wedge$, $\vee$, and the 
guarded existential quantifiers.

\begin{definition}[\bf \hGF{}]\em
The fragment \hGF$[\tau]$ of GF$[\tau]$ is given by the following grammar:
\begin{multline*}
\varphi,\varphi' ~::=~ \bot \ \mid \ \top \ \mid x=y \ \mid \ R(\xbf) \ \mid \ \varphi
                                                       \wedge
                                                       \varphi' \ 
                                                       \mid \ 
                                                       \lambda
                                                       \to \varphi\\
                                                     \mid \ \exists\, \ybf
                                                       (G(\xbf\ybf)
                                                       \land
                                                       \varphi(\xbf\ybf)) \ 
                                                       \mid \ 
                                                       \forall\, \ybf(G(\xbf\ybf) \to
                                                       \varphi(\xbf\ybf)),
\end{multline*}
where $R\in \tau$, $G(\xbf\ybf)$ are atomic formulas containing all the variables in
$\xbf\ybf$, and $\lambda\in \text{GF}^{\exists}[\tau]$.
\end{definition}
%
%

\hGF{} is closely related to guarded tuple-generating dependencies
(guarded tgds), a member of the Datalog$^\pm$ family of ontology
languages for which query answering is in
\PTime~\cite{DBLP:conf/pods/CaliGL09}. \emph{Guarded tgds} are 
FO-formulas of the form $\forall \xbf\forall \ybf\,
(\psi(\xbf\ybf)\rightarrow \exists \zbf\, \varphi(\xbf\zbf))$ with
conjunctions of atoms $\psi(\xbf\ybf)$ and $\varphi(\xbf\zbf)$ 
such that $\psi$ contains an atom $G(\xbf\ybf)$ guarding all the 
variables in $\xbf\ybf$. Thus, in contrast to \hGF{},
guarded tgds have no quantifier alternation and can be regarded as a
normal form for \hGF{}. In the appendix, we give
a polynomial time reduction of
deduction and query answering in \hGF{} to the respective problems for
guarded tgds by introducing fresh predicate names for complex formulas.
We also provide a polynomial reduction in the converse direction. We
note that satisfiability in \hGF{} has the same complexity as
satisfiability in GF~\cite{DBLP:journals/jsyml/Gradel99}:
\ExpTime-complete if the arity of predicates is bounded and
\TwoExpTime-complete otherwise. 
%
\begin{example}\em The Horn formulas equivalent to the concept
  $C_{\nabla}$ and TBox $\Tmc_{\textit{horn}}$ from the proof of
  Theorem~\ref{thm:11} are in \hGF{}. Thus, there are
  $\mathcal{ALC}$-concepts and TBoxes that are not equivalent to any
  \hALC{}-concepts or TBoxes, respectively, but nevertheless are
  equivalent to formulas in \hGF.  \end{example}
\begin{theorem}\label{thm:guarded11} $(i)$ Every formula in \hGF{} is
  equivalent to a Horn formula.

$(ii)$ There exists a sentence in GF---in fact, an $\mathcal{ALC}$-TBox---that is equivalent 
to a Horn sentence, but not equivalent to any \hGF{} sentence.
\end{theorem}
\begin{proof}
The proof of $(i)$ is by a straightforward induction.
For $(ii)$, consider the TBox 
$$
\Tmc_{\textit{guard}} ~=~ \{\, E \sqsubseteq \exists R . \top \sqcap \exists S.\top, \ 
E \sqcap \forall R.A \sqcap \forall S.B \sqsubseteq D \, \}.
$$
It is equivalent to the Horn sentence
\begin{multline*}
    \forall x\, (E(x) \to \exists y_{1}y_{2}\, (R(x,y_{1}) \land S(x,y_{2})\wedge{} \\
                        ((A(y_{1}) \wedge B(y_{2})) \rightarrow D(x))))
\end{multline*}
%
but, as shown in Example~\ref{ex:guardedexx} below, $\Tmc_{\textit{guard}}$ is not equivalent to any \hGF-sentence.
%
\end{proof}

Let $\Amf=(\dom(\Amf),(R^{\Amf})_{R\in \tau})$ be a $\tau$-structure. Denote by $\abf$ a \emph{tuple} $a_{1}\dots a_{n}$ of elements of $\Amf$ and set 
$[\abf]=\{a_{1},\dots,a_{n}\}$. A set $X \subseteq \text{dom}(\Amf)$ is \emph{guarded} in $\mathfrak{A}$ if $X$ is a singleton or $R^\Amf(\abf)$ for some $R\in\tau$
and $X = [\abf]$. A tuple $\abf$ is \emph{guarded} if $[\abf]$ is guarded. 

Let $\Amf$ and $\Bmf$ be structures. If $\abf$ and $\bbf$ are tuples
of the same length in $\Amf$ and $\Bmf$, respectively, we use
$p\colon\abf \mapsto \bbf$ to denote the map from $[\abf]$ to $[\bbf]$
with $p(a_{i})=b_{i}$. If $\abf'=a_{i_{1}}\dots a_{i_{k}}$ is a
subtuple of $\abf$, then $p(\abf')$ denotes the subtuple
$p(a_{i_{1}})\dots p(a_{i_{k}})$ of $\bbf$. The map $p$ is a
\emph{homomorphism} from $\Amf|_{[\abf]}$ to $\Bmf|_{[\bbf]}$ if
$\cbf \in R^{\Amf}$ implies $p(\cbf)\in R^{\Bmf}$
for all $R\in \tau$ and $\cbf$ with $[\cbf]\subseteq [\abf]$. 

In this section, by a \emph{pointed structure} we mean a pair $\Amf,X$ where $X\subseteq \text{dom}(\Amf)$ is a nonempty set of \emph{guarded tuples}, all of the same positive length. We again write $\Amf,\abf$ for $\Amf,\{\abf\}$.
We give the straightforward Ehrenfeucht-Fra\"{i}ss\'e type characterization for GF$^{\exists}$ needed for the
charachterization of \hGF{}. It is
obtained from the standard guarded bisimulation characterization of GF~\cite{DBLP:books/daglib/p/Gradel014}
by replacing partial isomorphisms by homomorphisms and dropping the backward condition.
\begin{definition}[\bf guarded simulation]\em
For $\tau$-structures $\Amf$ and $\Bmf$, a set $Z$ of maps from guarded sets in $\Amf$ to guarded sets in $\Bmf$ is called a \emph{guarded simulation} if the following conditions hold  for all $p \colon\abf \mapsto \bbf$ in $Z$: 
\begin{description}
\item[\hspace*{-1mm}(\atomg)]\ $p:\Amf|_{[\abf]}\rightarrow \Bmf|_{[\bbf]}$ is a homomorphism;
\item[\hspace*{-1mm}(\forthg)]\ for every guarded tuple $\abf'$ in $\Amf$, there exist  
a guarded tuple $\bbf'$ in $\Bmf$ and $p'$ such that $p'\colon\abf' \rightarrow \bbf'$
is in $Z$ and $p|_{[\abf]\cap [\abf']}=p'|_{[\abf] \cap [\abf']}$. 
\end{description} 
We write $\Amf,\abf\preceq_{\textit{gsim}}\Bmf,\bbf$ if there exists a guarded simulation between $\Amf$ and $\Bmf$ containing $p\colon\abf\mapsto \bbf$.
\end{definition}

In the same way as for $\mathcal{ELU}$, one can capture guarded simulations by 
\emph{guarded simulation games} between two players such that player~2 has a winning strategy 
(can respond to any move of player~1) iff $\Amf,\abf\preceq_{\textit{gsim}}\Bmf,\bbf$.
We write $\Amf,\abf \preceq_{\textit{gsim}}^{\ell} \Bmf,\bbf$ if player~2 has a winning 
strategy for the guarded simulation game with $\ell$ rounds.

We write $\Amf,\abf \leq_{\text{GF}^{\exists}}^{(\ell)} \Bmf,\bbf$ if, for all formulas $\lambda$ in GF$^{\exists}$ 
(of depth $\leq \ell$), $\Amf \models \lambda(\abf)$ implies $\Bmf \models \lambda(\bbf)$.

\begin{theorem}[\bf Ehrenfeucht-Fra\"{i}ss\'e game for GF$^{\exists}$]\label{thm:gfexists}
For any finite vocabulary $\tau$, pointed $\tau$-structures $\Amf,\abf$ and $\Bmf,\bbf$, and any $\ell < \omega$, we have\footnote{Here and in what follows the assumption that the tuples considered in Ehrenfeucht-Fra\"{i}ss\'e characterizations are guarded is not essential. It is straightforward to modify the model comparison games in such a way that the characterizations hold for arbitrary tuples.}
$$
\Amf,\abf  \leq_{\textit{GF}^{\exists}}^{\ell} \Bmf,\bbf \quad \text{ iff }
\quad \Amf,\abf \preceq_{\textit{gsim}}^{\ell} \Bmf,\bbf.
$$
If $\Amf$ and $\Bmf$ are finite, then
$$
\Amf,\abf \leq_{\textit{GF}^{\exists}} \Bmf,\bbf \quad \text{ iff }
\quad \Amf,\abf \preceq_{\textit{gsim}} \Bmf,\bbf.
$$
\end{theorem}


\subsection{Simulations for \hGF{}}
We introduce guarded Horn simulation games and prove an Ehrenfeucht-Fra\"{i}ss\'e type definability 
result for \hGF{}. A \emph{link} between structures $\Amf$ and $\Bmf$ is a pair $(P,\bbf)$ with $\bbf$ a guarded tuple in $\Bmf$ 
and $P$ a nonempty set of mappings $p\colon \bbf \mapsto p(\bbf)$ such that each $p$ is a homomorphism
from $\Bmf|_{[\bbf]}$ to $\Amf|_{[p(\bbf)]}$. We denote by $P[\bbf]$ the set $\{ p(\bbf) \mid p\in P\}$ and define the analogue of $X R^{\uparrow} Y$ for guarded Horn simulations.
If $(P,\bbf)$ is a link between $\Amf$ and $\Bmf$, and $\Amf,Y$ a pointed structure, 
$R(\xbf_{0}\ybf)$ an atomic formula of the same arity as the tuples in $Y$, and $\bbf_{0}$ a subtuple of $\bbf$ of the same length as $\xbf_{0}$, 
then we say that \emph{$Y$ is an $R(\bbf_{0}\ybf)$-successor of
$(P,\bbf)$} when, for any $p\in P$, there exists a tuple $\abf$ with $p(\bbf_{0})\abf\in Y$ and 
$\Amf\models R(p(\bbf_{0})\abf)$.
 
\begin{definition}[\bf guarded Horn simulation]\label{def:guardhornsim}\em
A \emph{guarded Horn simulation} between structures $\Amf$ and $\Bmf$ is a set $Z$ of links between $\Amf$ and $\Bmf$ such that for all $(P,\bbf)\in Z$, we have:
\begin{description}
\item[\hspace*{-1.5mm}(\atomgh)]\ for all atomic formulas $R(\xbf)$ and tuples $\bbf'$ with $[\bbf']\subseteq [\bbf]$, if $\Amf \models R(p(\bbf'))$ for all $p\in P$, then $\Bmf\models R(\bbf')$; 
\item[\hspace*{-1.5mm}(\forthgh)]\ for all sets $Y$ of guarded tuples
  and atomic formulas $R(\xbf_{0}\ybf)$, if $Y$ is an
  $R(\bbf_{0},\ybf)$-successor of $(P,\bbf)$, then there exists $(P',\bbf_{0}\bbf')\in Z$ such that $P'[\bbf_{0}\bbf'] \subseteq Y$;
\item[\hspace*{-1.5mm}(\backgh)]\ for every guarded tuple $\bbf'$ in $\Bmf$, there exists a link $(P',\bbf')$ in $Z$ such that, for any $p'\in P'$, there exists $p\in P$ with $p|_{[\bbf]\cap [\bbf']}=p'|_{[\bbf] \cap [\bbf']}$;
\item[\hspace*{-1.5mm}(\simgh)]\ there exists a guarded simulation between $(\Bmf,\bbf)$ and $(\Amf,p(\bbf))$ for every $p\in P$.
\end{description}
We write $\Amf,X \preceq_{\textit{ghsim}} \Bmf,\bbf$ if there exists a guarded Horn simulation $Z$ between 
$\Amf$ and $\Bmf$ such that $X=P[\bbf]$ for some $P$ with $(P,\bbf)\in Z$. 
\end{definition}

Lemma~\ref{lem:hornbisim} linking Horn simulations with bisimulations can be lifted to the
guarded case. In fact, any guarded bisimulation $Z$ between $\Amf$ and
$\Bmf$ defines a guarded Horn simulation \mbox{$Z'=\{
  (\{p^{-1}\},\bbf) \mid p\colon \abf \mapsto \bbf\in Z\}$} and if $Z$
  is a guarded Horn simulation with singleton $P$ for every
  \mbox{$(P,\bbf)\in Z$}, then $\{ p^{-1} \mid (\{p\},\bbf)\in Z\}$ is a
  guarded bisimulation (notice that, by~(\atomgh) and~(\simgh), $p$ is a partial
  isomorphism if $P=\{p\}$).

On the other hand, as the moves of player~1 are no longer restricted to
those along $R^{\Amf}$, $R\in \tau$, the relationship to products is 
subtler than in the Horn simulation case (Lemma~\ref{lem:product}).
\begin{example}\label{ex:ggg}\em
For $i=1,2$, let $\Amf_{i}=(\{a_{i}\},A_{1}^{\Amf_{i}},A_{2}^{\Amf_{i}})$, where
$A_1^{\Amf_1}=\{a_1\}$, $A_2^{\Amf_2}=\{a_2\}$, and $A_2^{\Amf_1}=A_1^{\Amf_2}=\emptyset$.
Then $Z = \{(\{a_{1},a_{2}\},(a_{1},a_{2}))\}$
is a Horn simulation between the disjoint union $\Amf$ of $\Amf_{1}$ and $\Amf_{2}$ and the
product $\Amf_{1}\times \Amf_{2}$, but it is not a guarded Horn simulation as $Y=\{a_{1}\}$ is a $(y=y)$-successor
of $(\{a_{1},a_{2}\},(a_{1},a_{2}))$ (with empty $\bbf_{0}$) for which there is no link with 
$Y$ as the first component in $Z$. Clearly, it is also not possible
to expand $Z$ to a guarded Horn simulation. We will revisit the relationship to products below. 
\end{example}

For $\ell<\omega$, we define the relations $\Amf,X \preceq_{\textit{ghsim}}^{\ell} \Bmf,\bbf$ in the obvious 
way following the definition Horn simulation games with $\ell$ rounds.
We write $\Amf,X \leq_{\textit{\hGF{}}}^{(\ell)} \Bmf,\bbf$ if for all formulas $\varphi$ in \hGF{} 
(of depth $\leq \ell$) the following holds: 
if $\Amf \models \varphi(\abf)$ for all $\abf\in X$, then $\Bmf \models \varphi(\bbf)$.

\begin{theorem}[\bf Ehrenfeucht-Fra\"{i}ss\'e game for \hGF{}]\label{thm:EhrenhornGF}
For any finite vocabulary $\tau$, pointed $\tau$-structures $\Amf,\abf$ and $\Bmf,\bbf$, and any 
$\ell < \omega$, we have 
$$
\Amf,\abf \leq_{\textit{hornGF}}^{\ell} \Bmf,\bbf \quad \text{ iff }
\quad \Amf,\abf \preceq_{\textit{ghsim}}^{\ell} \Bmf,\bbf.
$$
If $\Amf$ and $\Bmf$ are finite, then
$$
\Amf,\abf \leq_{\textit{hornGF}} \Bmf,\bbf \quad \text{ iff }
\quad  \Amf,\abf \preceq_{\textit{ghsim}} \Bmf,\bbf.
$$
\end{theorem}
The proof of Theorem~\ref{thm:EhrenhornGF} is similar to that of Theorem~\ref{lem:diagram}
and given in the appendix. In particular, one has to prove again a stronger version where the
tuple $\abf$ is replaced by a set $X$ of tuples. The existence of a winning strategy for player 2
in the guarded Horn simulation game is decidable in exponential time. Thus, it follows from
Theorem~\ref{thm:EhrenhornGF} that entailment and equivalence in \hGF{} are decidable in \ExpTime.
\begin{theorem}\label{thm:gf-complexity}
  In \hGF, entailment, equivalence, and CBE are in \ExpTime. Moreover,
  $\ell$-entailment, $\ell$-equivalence, and $\ell$-CBE are in
  \ExpTime for binary encoding of $\ell$ and in \PSpace for unary
  encoding.
\end{theorem}
In contrast to \hALC{}, it remains open whether the \ExpTime upper bound is tight.
Using guarded Horn simulations, we now show that the TBox $\Tmc_{\textit{guard}}$ from the
proof of Theorem~\ref{thm:guarded11} is not equivalent to any \hGF-sentence.
\begin{example}\label{ex:guardedexx}\em
Let $\Amf$ and $\Bmf$ be the structures below.
$\Bmf$ is a copy of $\Amf$ with the extra node $f$; $\Bmf$ refutes $\Tmc_{\textit{guard}}$ in $f$, 
but $\Amf$ is a model of $\Tmc_{\textit{guard}}$.
A guarded Horn simulation $Z$ between $\Amf$ and $\Bmf$ is given by
adding to the set of singleton links 
$$
\{(\{u\},u')\mid u\in \text{dom}(\Amf)\}
\cup \{(\{uv\},u'v')\mid (u,v)\in R^\Amf\cup S^\Amf\}
$$
the links
%
%
$(\{b,c\},f)$, $(\{bd,cd\},fd')$, and $(\{ba,ca\},fa')$.
\end{example}
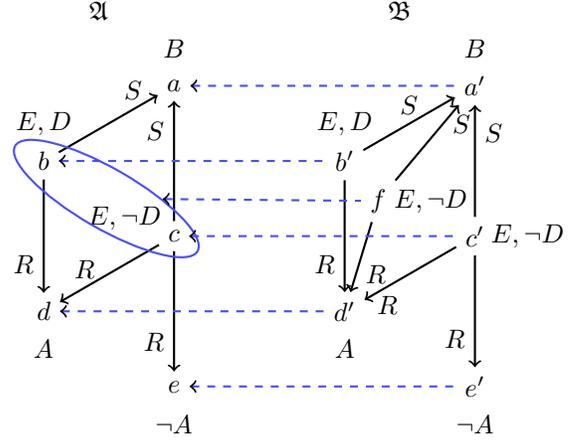
\begin{center}
    \begin{tikzpicture}
      [ every circle node/.style={draw, fill=black, inner sep=0pt,
        minimum size=.1cm},%
      xscale= 1,%
      yscale= 1,%
      ]
      \node (a)  at (0,0) {$a$};
      \draw (a) ++(-150:2cm) node (b) {$b$};
      \draw (a) ++(-90:2cm) node (c) {$c$};
      \node (d) [below of=b, node distance=2cm] {$d$};
      \node (e) [below of=c, node distance=2cm] {$e$};

      \draw [->, thick] (b) to node [above,near end] {$S$} (a);
      \draw [->, thick] (c) to node [left,near end] {$S$} (a);
      \draw [->, thick] (b) to node [left,near end] {$R$} (d);
      \draw [->, thick] (c) to node [above,near end] {$R$} (d);
      \draw [->, thick] (c) to node [left,near end] {$R$} (e);

      \node [above of=a, node distance=.5cm] {$B$};
      \node [above of=b, node distance=.5cm] {$E,D$};
      \draw (c) ++(160:.7cm ) node {$E, \neg D$};
      \node [below of=d, node distance=.5cm] {$A$};
      \node [below of=e, node distance=.5cm] {$\neg A$};

      \draw (a) ++(90:1cm) ++(-180:1cm) node {$\Amf$};

      \node (a1) [right of=a, node distance=4cm] {$a'$};
      \draw (a1) ++(-150:2cm) node (b1) {$b'$};
      \draw (a1) ++(-130:2cm) node (f1) {$f$};
      \draw (a1) ++(-90:2cm) node (c1) {$c'$};
      \node (d1) [below of=b1, node distance=2cm] {$d'$};
      \node (e1) [below of=c1, node distance=2cm] {$e'$};

      \draw [->, thick] (b1) to node [above] {$S$} (a1);
      \draw [->, thick] (b1) to node [left,near end] {$R$} (d1);

      \draw [->, thick] (c1) to node [right,near end] {$S$} (a1);
      \draw [->, thick] (c1) to node [below,near end] {$R$} (d1);
      \draw [->, thick] (c1) to node [left,near end] {$R$} (e1);

      \draw [->, thick] (f1) to node [right,near end] {$S$} (a1);
      \draw [->, thick] (f1) to node [right,near end] {$R$} (d1);

      \node [above of=a1, node distance=.5cm] {$B$};
      \node [above of=b1, node distance=.5cm] {$E,D$};
      \node [right of=f1, node distance=.7cm] {$E, \neg D$};
      \node [right of=c1, node distance=.7cm] {$E, \neg D$};
      \node [below of=d1, node distance=.5cm] {$A$};
      \node [below of=e1, node distance=.5cm] {$\neg A$};

      \draw (a1) ++(90:1cm) ++(-180:1cm) node {$\Bmf$};

      \node (fit) [ellipse, draw=blue!70!white,
      thick, rotate=-30, minimum width=2.8cm, minimum height=.8cm]
      at (-.9,-1.5) {};
      
      \draw [<-, thick, dashed, draw=blue!70!white] (a) to (a1);
      
      \draw [<-, thick, dashed, draw=blue!70!white] (b) -- (b1);

      \draw [<-, thick, dashed, draw=blue!70!white] (c) -- (c1);

      \draw [<-, thick, dashed, draw=blue!70!white] (d) -- (d1);

      \draw [<-, thick, dashed, draw=blue!70!white] (fit) -- (f1);

      \draw [<-, thick, dashed, draw=blue!70!white] (e) -- (e1);

    \end{tikzpicture}
%
\end{center}
\subsection{Expressive Completeness for \hGF{}}
Our next aim is to show that an FO-formula $\varphi$ is equivalent to a \hGF-formula just in case it is preserved under guarded Horn simulations. This statement
needs qualification, however, in two respects. First, our infinitary proof goes through only
if we require the sets $Y$ where player~1 moves to in condition~(\forthgh) to be intersections of FO-definable sets.
Second, in contrast to \hALC{}, the language \hGF{} admits equality guards and is not local
in the sense that the truth of a \hGF-formula $\varphi(\xbf)$ in $\Amf,\abf$
is not determined by some neighbourhood of $\abf$ in the Gaifman graph of $\Amf$. 
As a consequence, GF-sentences such as $\varphi=\forall x\, A_{1}(x)
\vee \forall x\, A_{2}(x)$ (with omitted equality guards)
are not equivalent to any \hGF-sentence but preserved under guarded Horn simulations.

To deal with this issue, we lift the definition of guarded Horn simulations from single structures to families
of structures. Let $\Amf_{i}$, $i \in I$, be a family of disjoint structures. 
A set $X$ of tuples in $\bigcup_{i\in I}\text{dom}(\Amf_{i})$ 
\emph{intersects with all $\Amf_{i}$, $i \in I$}, if $X$ contains at
least one tuple from each $\text{dom}(\Amf_{i})$. For an open formula $\varphi(\xbf)$, we write $(\Amf_{i}\mid i\in I) \models \varphi(\abf)$ if 
$\abf$ is a nonempty tuple in some $\text{dom}(\Amf_{i})$ and $\Amf_{i}\models \varphi(\abf)$. 
For closed $\varphi$, we write $(\Amf_{i}\mid i\in I) \models \varphi$ if $\Amf_{i}\models \varphi$ for all $i\in I$.
A set $X$ of tuples is \emph{FO$^{\infty}$-definable} in $(\Amf_{i}\mid i\in I)$ if 
there is a set $\Gamma(\xbf)$ of FO-formulas with
$
X= \{ \abf \mid \forall \varphi\in \Gamma(\xbf)\ (\Amf_{i}\mid i\in I) \models \varphi(\abf)\}.
$
\begin{definition}[\bf generalized guarded Horn simulation]\em
Let $\Amf_{i}$, $i\in I$, be a family of disjoint structures, $\Amf$ the
disjoint union of $\Amf_{i}$, $i\in I$, and $\Bmf$ a structure.
A set $Z$ of links between $\Amf$ and $\Bmf$ is a \emph{generalized guarded Horn simulation
between $(\Amf_{i}\mid i\in I)$ and $\Bmf$} if all $(P,\bbf)\in Z$ satisfy
conditions (\atomgh) and (\backgh) from Definition~\ref{def:guardhornsim} and
\begin{description}
\item[\hspace*{-2mm}(\forthggh)]\ for all sets $Y$ of guarded tuples in $\Amf$ and atomic formulas 
$R(\xbf_{0}\ybf)$, if $Y$ is an $R(\bbf_{0}\ybf)$-successor of $(P,\bbf)$ in $\Amf$ and 
\begin{enumerate}
\item[--] $\bbf_{0}$ is not empty or
\item[--] $\bbf_{0}$ is empty and $Y$ intersects with all $\Amf_{i}$, $i\in I$,
\end{enumerate}
then there is $(P',\bbf_{0}\bbf')\in Z$ with $P'[\bbf_{0}\bbf'] \subseteq Y$;
\item[\hspace*{-2mm}(\emph{sim}$^{gg}_h$)] there exists a guarded simulation between $(\Bmf,\bbf)$ and 
$(\Amf_{i},p(\bbf))$ for every $p\in P$ and $p(\bbf)$ in $\text{dom}(\Amf_{i})$.
\end{description}
$Z$ is \emph{FO-restricted} if (\forthggh) holds for all FO$^{\infty}$-definable $Y$.
We write $(\Amf_{i}\mid i \in I),X\preceq_{\textit{ghorn}}^{\textit{FO}}\Bmf,\bbf$ if there exists an
FO-restricted generalized guarded Horn simulation $Z$ between $(\Amf_{i}\mid i \in I)$ and $\Bmf$ such 
that $X=P[\bbf]$ for some $(P,\bbf)\in Z$.
\end{definition}
Note that, as we modified~(\forthgh), the set $Z$ in Example~\ref{ex:ggg}
is a generalized guarded Horn simulation. In fact, now Lemma~\ref{lem:product}
can be lifted to the guarded case: if $\Amf_{i}$, $i\in I$, is a family of structures, the set of all $(P,f_{1}\ldots f_{n})$ with 
$f_{1}\dots f_{n}$ a guarded tuple in $\prod_{i\in I}\Amf_{i}$ and $p\in P$ just in case there exists $i\in I$ such that 
$p(f_{j})= f_{j}(i)$, $1\leq j \leq n$, is a generalized guarded Horn simulation between the disjoint union of the $\Amf_{i}$ and 
$\prod_{i\in I}\Amf_{i}$.

A formula $\varphi(\xbf)$ is \emph{preserved under FO-restricted 
generalized guarded Horn simulations} if $(\Amf_{i}\mid i\in I) \models \varphi(\abf)$ for all $\abf\in X$ and $(\Amf_{i}\mid i \in I),X\preceq_{\textit{ghorn}}^{\textit{FO}}\Bmf,\bbf$ imply $\Bmf \models \varphi(\bbf)$.
%
\begin{theorem}[\bf expressive completeness: \hGF{}]\label{thm:exprhgf}
An FO-formula is equivalent to a \hGF{}-formula iff it is preserved under
FO-restricted generalized guarded Horn simulations.
\end{theorem}
\begin{proof} (sketch) The implication $(\Rightarrow)$ is straightforward. Conversely,
suppose $\varphi(\xbf_{0})$ is preserved under FO-restricted generalized guarded Horn simulations.
Let $\text{cons}(\varphi)$ be the set of all $\psi(\xbf_{0})$ in \hGF{} entailed by $\varphi(\xbf_{0})$.
By compactness, it suffices to show $\text{cons}(\varphi)\models \varphi$. Let $\Bmf$ be an $\omega$-saturated~\cite{ChangKeisler}
model satisfying $\text{cons}(\varphi)(\bbf_{0})$ for some tuple $\bbf_{0}$ in $\text{dom}(\Bmf)$. 
We show $\Bmf\models \varphi(\bbf_{0})$. For any tuple $\bbf$ and tuple $\xbf$ of variables of the same length as $\bbf$, we denote by 
$\lambda_{\Bmf,\bbf}(\xbf)$ the set of guarded existential positive $\lambda(\xbf)$ with $\Bmf\models \lambda(\bbf)$.
Let $\mathcal{C}$ be the set of all sets $\Gamma(\xbf_{0})$ of FO-formulas with $\Bmf\models \Gamma(\bbf_{0})$ and such that
$\Gamma(\xbf_{0}) \cup \{\varphi(\xbf_{0})\}$ is satisfiable and take, for any $\Gamma(\xbf_{0})\in \mathcal{C}$, an 
$\omega$-saturated structure $\Amf_{\Gamma}$ and tuple $\abf_{\Gamma}$ with $\Amf_{\Gamma}\models (\Gamma\cup \{\varphi\})(\abf_{\Gamma})$.
Let $\Amf$ be the disjoint union of $(\Amf_{\Gamma}\mid \Gamma\in \mathcal{C})$ and let
$Z$ be the set of pairs $(X,\bbf)$ such that
\begin{itemize}
\item[(a)] for any $\psi(\xbf)\in \hGF{}$, if $(\Amf_{\Gamma}\mid \Gamma\in \mathcal{C})
\models \psi(\abf)$ for all $\abf\in X$, then $\Bmf\models \psi(\bbf)$;
\item[(b)] there exists a set $\Phi(\xbf)\supseteq \lambda_{\Bmf,\bbf}$ of FO-formulas
such that $X$ is the set of all tuples $\abf$ with  
$(\Amf_{\Gamma}\mid \Gamma\in \mathcal{C}) \models \Phi(\abf)$.
\end{itemize}
Each $(X,\bbf)\in Z$ can be regarded as a link $(P,\bbf)$ with $X=P[\bbf]$. As we work with $\omega$-saturated structures, one can show that 
$Z$ is an FO-restricted generalized guarded Horn simulation between $(\Amf_{\Gamma}\mid \Gamma\in \mathcal{C})$ and $\Bmf$. 
Since $\varphi(\xbf_{0})$ is preserved under generalized guarded Horn simulations, $\Bmf\models \varphi(\bbf_{0})$.
\end{proof}

\section{Conclusion and Outlook}

We have introduced model comparison games for \hALC{} and \hGF{} and obtained Ehrenfeucht-Fra\"{i}ss\'e type definability and van 
Benthem style  expressive completeness results. For \hALC{}, our results are `complete': the characterizations hold in both classical and finite model theory
settings without any restrictions on the players' moves, and the straightforward \ExpTime upper bound for checking 
indistinguishability of models and concept learnability using our model comparison games is tight. For more expressive \hGF{}, it 
remains open whether the characterization holds in the setting of finite model theory and whether the moves of the 
players have to be restricted to `saturated' sets in the expressive completeness result. In this case, it is also open whether 
the \ExpTime upper bound for checking indistinguishability of models is tight.

A different line of open research problems arises from the fact that \hALC{} and \hGF{} do not capture the intersections of
$\mathcal{ALC}$ (respectively, GF) and Horn FO. It is thus an open problem to find out whether there exists a `neat' syntactic definition of the intersection
of $\mathcal{ALC}$ and Horn FO such that, if an $\mathcal{ALC}$-concept or TBox is equivalent to a Horn FO formula, 
then it is equivalent to a concept, or, respectively, TBox, satisfying this definition. The analogous question is
also open for \hGF{} and Horn FO. The proofs of Theorems~\ref{thm:11}
and ~\ref{thm:guarded11} suggest the following syntactic extension of \hALC{}.
\begin{example}\em
Denote by $\mathcal{ELU}_{\nabla}$ the extension of $\mathcal{ELU}$ with the 
\emph{$\nabla$-operator} defined as $\nabla R.C = \exists R.\top \sqcap \forall R.C$
%
%
and let $\hALC_{\nabla}$ be defined in the same way as \hALC{} (Definition~\ref{defHornALC}) with the exception that now $L$ is an  
$\mathcal{ELU}_{\nabla}$-concept. Then the concept $C_{\nabla}$ from the proof of Theorem~\ref{thm:11}~$(i)$ and the TBox $\Tmc_{\textit{guard}}$ from the
proof of Theorem~\ref{thm:guarded11} are clearly a $\hALC_{\nabla}$-concept and TBox, respectively. So $\hALC_{\nabla}$ captures more
from the intersection of $\mathcal{ALC}$ and Horn FO than \hALC{}. One can also show by an inductive argument that all 
\hALC$_{\nabla}$-concepts and TBoxes are indeed equivalent to Horn FO formulas (details are in the appendix). 
However, again this language does not fully capture the intersection in question as the TBox $\Tmc_{\textit{horn}}$ from
the proof of Theorem~\ref{thm:11}~$(ii)$ is not equivalent to any \hALC$_{\nabla}$-TBox. This can be shown by introducing
model comparison games for $\hALC_{\nabla}$ (obtained by replacing the simulation game for $\mathcal{ELU}$ with a game capturing $\mathcal{ELU}_{\nabla}$)
and showing that the Horn simulation from Example~\ref{exam:1}~$(ii)$ preserves $\hALC_{\nabla}$. 
\end{example}

Taking into account the examples given in this paper, we arrive at the following
lattice of languages and their intersections (modulo equivalence) where all inclusions are proper:

\medskip
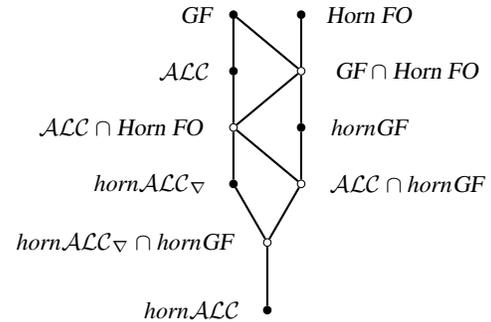
\quad\begin{tikzpicture}
  [every circle node/.style={draw, inner sep=0pt,
    minimum size=.1cm},%
  xscale= 1.2,%
  yscale= 1.2,%
  ]
  \node (a) [circle, fill=black] {};
  \node (alabel) [left of=a, node distance=1cm] {\small $\hALC$};

  \draw (a) ++(90:.75cm) node[circle] (a1) {};
  \node (a1label) [left of=a1, node distance=1.9cm] {\small
$\hALC_\nabla\cap\hGF$};

  \draw (a1) ++(120:.75cm) node[circle, fill=black] (left) {};
  \node (leftlabel) [left of=left, node distance=1.1cm]{\small $\hALC_\nabla$};

  \node (left1) [circle, above of=left, node distance=.75cm] {};
  \node (left2) [circle, fill=black, above of=left1, node distance=.75cm] {};
  \node (left3) [circle, fill=black, above of=left2, node distance=.75cm] {};

  \node (left1label) [left of=left1, node distance=1.5cm]
  {\small $\mathcal{ALC}\cap\textsl{Horn\ FO}$};
  \node (left2label) [left of=left2, node distance=.65cm]
  {\small $\mathcal{ALC}$};
  \node (left3label) [left of=left3, node distance=.5cm]
  {\small $\textsl{GF}$};

  \draw (a1) ++(60:.75cm) node[circle] (right) {};
  \node (rightlabel) [right of=right, node distance=1.4cm]{\small
$\mathcal{ALC}\cap\hGF$};

  \node (right1) [circle, fill=black, above of=right, node distance=.75cm] {};
  \node (right2) [circle, above of=right1, node distance=.75cm] {};
  \node (right3) [circle, fill=black, above of=right2, node distance=.75cm] {};

  \node (right1label) [right of=right1, node distance=.9cm]
  {\small $\hGF$};
  \node (right2label) [right of=right2, node distance=1.4cm]
  {\small $\textsl{GF}\cap\textsl{Horn\ FO}$};
  \node (right3label) [right of=right3, node distance=.9cm]
  {\small $\textsl{Horn\ FO}$};

  \draw [-, thick, draw=black] (a) -- (a1);
  \draw [-, thick, draw=black] (a1) -- (left);
  \draw [-, thick, draw=black] (a1) -- (right);
  \draw [-, thick, draw=black] (left) -- (left1);
  \draw [-, thick, draw=black] (left1) -- (left2);
  \draw [-, thick, draw=black] (left2) -- (left3);
  \draw [-, thick, draw=black] (right) -- (right1);
  \draw [-, thick, draw=black] (right1) -- (right2);
  \draw [-, thick, draw=black] (right2) -- (right3);
  \draw [-, thick, draw=black] (right) -- (left1);
  \draw [-, thick, draw=black] (left1) -- (right2);
  \draw [-, thick, draw=black] (right2) -- (left3);
\end{tikzpicture}

\bigskip
\noindent \textsc{Acknowledgements.} F.~Papacchini was supported by
the EPSRC UK grants EP/R026084 and EP/L024845.  F.~Wolter and
M.~Zakharyaschev were supported by the EPSRC grants EP/M012646 and
EP/M012670. J.C.~Jung was supported by ERC consolidator grant 647289
CODA. 

\newpage{}
\bibliographystyle{IEEEtran}



\newpage

\appendix

\newcommand{\phALC}{\ensuremath{\textit{p-horn}\mathcal{ALC}}}
\newcommand{\shALC}{\ensuremath{\textit{s-horn}\mathcal{ALC}}}
\newcommand{\nhALC}{\ensuremath{\textit{n-horn}\mathcal{ALC}}}

\section{Other Definitions of Horn Description and Modal Logics}
\label{sec:horn_relationships}
We show that the definition of $\hALC$-concepts used
in this paper is equivalent to syntactically different definitions of 
Horn-$\mathcal{ALC}$ concepts and Horn modal formulas given in the literature.

A definition for Horn DLs based on polarity is given
in~\cite{DBLP:conf/ijcai/HustadtMS05}. This definition can be
restricted to the $\mathcal{ALC}$ case as follows. We say that an $\mathcal{ALC}$-concept $C$ is a \emph{\phALC-concept} if
${\sf pol}(C) \leq 1$, where ${\sf pol}(C) = {\sf pl}^+(C)$ with
${\sf pl}^+$ defined as shown in the table below: 
\[
\begin{array}{l l l}
  C & {\sf pl}^+ & {\sf pl}^-\\
  \hline\\[-.5em]
  \top & 0 & 0\\[.5em]
  \bot & 0 & 0\\[.5em]
  A & 1 & 0\\[.5em]
  \neg C' & {\sf pl}^-(C') & {\sf pl}^+(C')\\[.5em]
  \bigsqcap_i C_i & {\sf max}_i {\sf pl}^+(C_i) & \sum_i {\sf
                                                  pl}^-(C_i)\\[.5em]
  \bigsqcup_i C_i & \sum_i {\sf pl}^+(C_i) & {\sf max}_i {\sf
                                             pl}^-(C_i)\\[.5em]
  \exists R. C' \qquad {}& {\sf max}\{1, {\sf pl}^+(C')\} \qquad {} & {\sf pl}^-(C')\\[.5em]
  \forall R. C'& {\sf pl}^+(C') & {\sf max}\{1, {\sf pl}^-(C')\}\\[.5em]
  \hline
\end{array}
\]

\begin{theorem}
  \label{sHorn_pHorn}
  Every \hALC-concept is equivalent to a
  \phALC-concept, and vice versa.
\end{theorem}

\begin{proof}
It is readily checked by induction on the construction of a \hALC-concept $C$ that ${\sf pol}(C)\leq 1$, and so $C$ is also a \phALC-concept.

  For the converse direction, we define a translation ${\sf sH}$
  of \phALC-concepts to equivalent \hALC-concepts. To ease the proof, we assume the
  \phALC-concepts to be in negation normal form (NNF). Under this assumption,
  the definition of \phALC-concepts can be
  simplified. Namely, an $\mathcal{ALC}$ concept $C$ in NNF is a
  \phALC-concept if ${\sf pol}(C)\leq 1$, where ${\sf pol}(C)$ is 
  defined as follows:
  $$
  {\sf pol}(C)=\left\{
    \begin{array}{l l}
      0 & \text{if }C= \top \mid \bot \mid \neg A\\[.5em]
      1 & \text{if }C=A\\[.5em]
      {\sf max}_i {\sf pol}(C_i) & \text{if }C = C_1\sqcap \ldots
                                        \sqcap C_n\\[.5em]
      \sum_i {\sf pol}(C_i) & \text{if }C = C_1\sqcup \ldots
                                        \sqcup C_n\\[.5em]
      {\sf max}\{1, {\sf pol}(C')\} & \text{if }C=\exists R. C'\\[.5em]
      {\sf pol}(C') & \text{if }C=\forall R. C'
    \end{array}
  \right.
  $$

\medskip
\noindent
  \textit{Claim.} Any \phALC-concept $C = C_1\sqcup
  \ldots \sqcup C_n$ in NNF is equivalent to a concept of the form $L \to D$, where $L$ is an
  $\mathcal{ELU}$-concept, and either $D = \bot$ or $D = C_j$, for some $j$,
  $1 \leq j \leq n$, and ${\sf pol}(C_j)=1$.

\smallskip
\noindent
\emph{Proof of claim}. 
As ${\sf pol}(C)\leq 1$, there exists at most one disjunct $C_j$, $1\leq j\leq n$, with ${\sf pol}(C_j) = 1$. We define a \phALC-concept 
$C_\sqcup \sqcup D$ equivalent to
$C$ in the following way. If there exists a disjunct $C_j$ of $C$ with
${\sf pol}(C_j)=1$, then
$C_\sqcup = C_1 \sqcup \ldots \sqcup C_{j-1} \sqcup C_{j+1}\sqcup
\ldots \sqcup C_n$
and $D = C_j$; otherwise $C_\sqcup = C$ and $D=\bot$. It follows that
${\sf pol}(C_\sqcup) = 0$, which implies that $C_\sqcup$ is an
$\mathcal{ALC}$-concept built using $\top$, $\bot$, $\neg A$,
$\sqcup$, $\sqcap$ and $\forall R$ only. Let $L$ be the $\mathcal{ELU}$-concept defined as ${\rm NNF}(\neg C_\sqcup)$. Then $C$ is
equivalent to $L \to D$, as required.

\smallskip

  In the following definition of the translation ${\sf sH}$, we use $L_C$ and $D_C$ for
  a $\phALC$-concept $C=C_1\sqcup \ldots \sqcup C_n$ to denote the
  equivalent concept $L_C \to D_C$:
  \[
  {\sf sH}(C)=\left\{
    \begin{array}{l l}
      C & \text{if }C ::= \top\mid \bot\mid A \mid \neg A\\[.5em]
      \bigsqcap_{i=1}^n {\sf sH}(C_i) & \text{if }C = C_1\sqcap \ldots
                                        \sqcap C_n\\[.5em]
       L_C\to {\sf sH}(D_C) & \text{if }C = C_1\sqcup \ldots
                                        \sqcup C_n\\[.5em]
      \exists R. {\sf sH}(C') & \text{if } C = \exists R . C'\\[.5em]
      \forall R . {\sf sH}(C') & \text{if } C = \forall R . C'\\
    \end{array}
    \right.
  \]
It is readily seen that ${\sf sH}(C)$ is a \hALC-concept equivalent to $C$.
\end{proof}

Horn modal formulas have been defined in various ways in the
literature~\cite{Nguyen04,Sturm00,BresolinMS16,CerroP87,ChenLin93},
and not all of the definitions are equivalent. Sturm~\cite{Sturm00} defines
Horn modal formulas with $n$-ary modal operators. We show that, when restricted
to unary modal operators, his definition is equivalent to the
definition of \hALC-concepts given in this paper.  We
rephrase the definition in~\cite{Sturm00} in the DL terms as
follows.  Let $\Hmc_b$ be defined by the grammar
$$
H,H'::= \bot \mid \neg A \mid H \sqcap H' \mid H \sqcup H' \mid \forall R.H,
$$
where $A$ is a concept name and $R$ a role name.
Then the set $\mathcal{H}$ of \emph{\shALC-concepts} is the smallest set containing $\Hmc_b \cup N_C$, closed under
$\sqcap$, $\exists R$, $\forall R$ and such that whenever $C,C' \in \mathcal{H}$ and
$C\in \Hmc_b$ or $C'\in \Hmc_b$, then $C \sqcup C' \in \mathcal{H}$. The set
$\Hmc_b$ can be seen as the set containing the negation of
$\mathcal{ELU}$-concepts, and so the equivalence with our definition can
be obtained by an argument analogous to the one in the proof of
Theorem~\ref{sHorn_pHorn}.

The remaining notions of Horn modal formulas~\cite{Nguyen04,BresolinMS16,CerroP87,ChenLin93} are rather
different from our definition of \hALC-concepts. 
To illustrate, we focus on Nguyen's definition~\cite{Nguyen04},  
rephrasing it in the DL parlance. A \emph{\nhALC-concept} $G$ is defined by 
the following grammar:
\begin{align*}
  P,P' &::= \top \mid \bot \mid A \mid P \sqcap P' \mid P \sqcup P'
         \mid \exists R . P \mid \forall R. P\\
  G,G'  &::= A \mid \neg P \mid G \sqcap G' \mid \exists R . G \mid
          \forall R. G \mid P\to G
\end{align*}
The crucial difference between \nhALC-concepts and \hALC-concepts lies in the definition of $P$,  which allows universal role restrictions $\forall R.P$. 
%
%
It is not hard to see that the \nhALC-concept 
$\forall R. A \to B$ is not preserved under products, and so is not equivalent to any Horn FO formula.

\section*{Definitions and Proofs for Section: Simulations for \hALC{}}
We first give a rigorous definition of the relation $\Amf,X \preceq_{\textit{horn}}^{\ell}\Bmf,b$ and
then supply the missing details of the proof of Lemma~\ref{lem:diagram}.






\begin{definition}[\bf $\ell$-Horn simulation]\em
Let $\Amf$ and $\Bmf$ be $\tau$-structures. Define relations $\preceq_{\textit{horn}}^{\ell}$, $\ell<\omega$, between pointed 
structures $\Amf,X$ and $\Bmf,b$ by induction:
\begin{itemize}
\item $\Amf,X \preceq_{\textit{horn}}^{0}\Bmf,b$ if $X\not=\emptyset$ and $X\subseteq A^{\Amf}$ implies $b\in A^{\Bmf}$, 
for all $A\in \tau$, and $\Bmf,b\preceq_{\textit{sim}}^{0}\Amf,a$ for all $a\in X$.
\item $\Amf,X \preceq_{\textit{horn}}^{\ell+1}\Bmf,b$ if the following conditions hold:
 \begin{itemize}
  \item $\Amf,X \preceq_{\textit{horn}}^{0} \Bmf,b$;
  \item if $X R^{\Amf\uparrow} Y$, then there exist $Y' \subseteq Y$ and $b'\in \text{dom}(\Bmf)$ such that 
        $(b,b')\in R^{\Bmf}$ and $\Amf,Y' \preceq_{\textit{horn}}^{\ell}\Bmf,b'$, for all $R\in \tau$;
  \item if $(b,b')\in R^\Bmf$,
    then there exists $Y \subseteq \text{dom}(\Amf)$ with $X R^{\Amf\downarrow} Y$ and 
    $\Amf,Y \preceq_{\textit{horn}}^{\ell} \Bmf,b'$, for all $R\in \tau$;
  \item $\Bmf,b\preceq_{\textit{sim}}^{\ell+1} \Amf,a$ for all $a\in X$.
  \end{itemize}
\end{itemize}
\end{definition}

\medskip
\noindent
{\bf Lemma~\ref{lem:diagram}}
For any finite vocabulary $\tau$, pointed $\tau$-structures $\Amf,X$ and $\Bmf,b$,
and any $\ell < \omega$,
$$
\Amf,X \leq_{\textit{horn}\mathcal{ALC}}^{\ell} \Bmf,b \quad \text{iff} \quad 
\text{$\exists X_{0}\subseteq X\ \Amf,X_{0} \preceq_{\textit{horn}}^{\ell}\Bmf,b$}.
$$
If $\Amf$ and $\Bmf$ are finite, then
$$
\Amf,X \leq_{\textit{horn}\mathcal{ALC}} \Bmf,b \quad \text{iff} \quad 
\text{$\exists X_{0}\subseteq X\ \Amf,X_{0} \preceq_{\textit{horn}}\Bmf,b$}.
$$
\begin{proof}
It remains to prove the direction from right to left of the first claim. Thus we prove the following.

\medskip
\noindent 
\emph{Claim 1}. For any $\ell < \omega$, $X\subseteq \text{dom}(\Amf)$ and $b\in \text{dom}(\Bmf)$, if
$\Amf,X \preceq_{\textit{horn}}^{\ell}\Bmf,b$, then $X\subseteq C^{\Amf}$ implies $b\in C^{\Bmf}$
for every $C \in \text{Horn}_{\ell}$.

\smallskip
\noindent
\emph{Proof of claim}. 
For $\ell=0$, Claim~1 is trivial. Suppose it has been proved for $\ell$ and show that it holds for $\ell+1$ by induction on the construction of $C$.
Thus, suppose that Claim~1 has been proved for $C', C_{1}, C_{2}\in \text{Horn}_{\ell}$, and that
$C\in \text{Horn}_{\ell+1}$ is of the form $C=\forall R.C'$, $C=\exists R.C'$, $C= C_{1}\sqcap C_{2}$, or 
$C= L\rightarrow C'$ with $L$ a $\mathcal{ELU}$ concept of depth $\leq \ell+1$.
Suppose also that $\Amf,X \preceq_{\text{horn}}^{\ell+1}\Bmf,b$ and $X\subseteq C^{\Amf}$. 

\emph{Case} $C=\forall R.C'$. Suppose $b\not\in (\forall R.C')^{\Bmf}$. Choose $b'$ with $(b,b')\in R^{\Bmf}$ and $b'\not\in C'^{\Bmf}$. By~(\backh), there is $Y \subseteq \text{dom}(\Amf)$ with $X R^{\Amf\downarrow} Y$ and $\Amf,Y \preceq_{\textit{horn}}^{\ell} \Bmf,b'$. We have $C'\in \text{Horn}_{\ell}$, and so, by IH for $\ell$, there exists $a'\in Y$ with $a'\not\in C'^{\Amf}$. Then there exists $a\in X$ with $a \not\in (\forall R.C')^{\Amf}$, which is impossible.

\emph{Case} $C = \exists R.C'$. Suppose $b\not\in (\exists R.C')^{\Bmf}$. We have $X R^{\Amf\uparrow} Y$ for 
$Y=C'^{\Amf}$. By~(\forthh) for $Y$, there exist $Y'\subseteq Y$ and $b'\in \text{dom}(\Bmf)$ with 
$(b,b')\in R^\Bmf$ and $\Amf,Y \preceq_{\textit{horn}}^{\ell}\Bmf,b'$. Since $C'\in \text{Horn}_{\ell}$ and by IH, 
$b'\in C'^{\Bmf}$. But then $b\in (\exists R.C')^{\Bmf}$, contrary to our assumption. 

\emph{Case} $C= (L \rightarrow C')$. Suppose $b\notin (L\rightarrow C')^{\Bmf}$. 
Then $b\in L^{\Bmf}$ and $b\not\in C'^{\Bmf}$.
By Theorem~\ref{thm:elu}, $X\subseteq L^{\Amf}$, and by IH, there exists 
$a\in X$ with $a\not\in C'^{\Amf}$. 
Then $X\not\subseteq (L\rightarrow C')^{\Amf}$, which is a contradiction.

The remaining case $C=C_{1}\sqcap C_{2}$ is easy and left to the reader.
%
%
%
\end{proof}

\section*{Proofs for Section~\ref{sec:complexity}}

\noindent\textbf{Lemma~\ref{lem:reductions}}. \begin{itshape}
  \begin{enumerate}

    \item[(1)] $\text{CBE}\leq_T^P\text{HornSim}$;

    \item[(2)] $\overline{\text{HornSim}}\leq_m^P\text{CBE}$;

    \item[(3)] $\text{HornSim}\leq_m^P \text{Entailment}$;

    \item[(4)] $\text{Entailment}\leq_m^P \text{HornSim}$;

    \item[(5)] $\text{Equivalence}\leq_T^P \text{Entailment}$;

    \item[(6)] $\text{Entailment}\leq_m^P\text{Equivalence}$.

  \end{enumerate}
\end{itshape}

\begin{proof}
  For (1), observe that the following are equivalent by
  Lemma~\ref{lem:diagram} for all $\Amf,P,N$:

  \begin{itemize}

    \item there is some \hALC-concept $C$ such that $P \subseteq C^\Amf$ and 
          $C^\Amf\cap N=\emptyset$,

    \item $\Amf,P\not\preceq_{\textit{horn}}\Amf,b$ for all $b\in N$.

  \end{itemize}
  
  To see ``$\Rightarrow$'', let $C$ be such a concept. By
  Lemma~\ref{lem:diagram}, we have 
  have $\Amf,P\not\preceq_{\textit{horn}} \Amf,b$, for all $b\in N$.
  
  Conversely, suppose $\Amf,P\not\preceq_{\textit{horn}}\Amf,b$ for all
  $b\in N$. By Lemma~\ref{lem:diagram}, we have
  $\Amf,P\not\leq_{\hALC{}}\Amf,b$, for all $b\in N$. Thus, there is a \hALC{}-concept $C_b$ such that
  $P\subseteq C_b^\Amf$ and $b\notin C_b^\Amf$, for all $b\in N$. Let $C$ be the
  conjunction of the concepts $C_b$, for all $b\in N$. Clearly,
  $P\subseteq C^\Amf$ and $C^\Amf\cap N=\emptyset$.

  \smallskip For (2), observe that $\Amf,X\preceq_{\textit{horn}}\Bmf,b$ iff
  in the disjoint union $\Amf\uplus \Bmf$ of $\Amf$ and $\Bmf$ the positive examples $P=X$ cannot
  be distinguished from the negative examples $N=\{b\}$.
  
  \smallskip For (3), let $\Amf,\Bmf,X,b$ be an input to HornSim. Define $\Amf'$ by adding a new
  $R$-predecessor $a$ to all nodes in $X$. Further, define $\Bmf'$ by
  taking the disjoint union of $\Amf$ and $\Bmf$ and adding a new
  $R$-predecessor $d$ to $b$, and making $d$ also a predecessor of all
  nodes in (the copies of) $X$. Then we have:
  \begin{align*}
    \Amf,X\preceq_{\textit{horn}}\Bmf,b \quad\text{iff\quad}
    \Amf',a\preceq_{\textit{horn}}\Bmf',d,
  \end{align*}
  and the latter is equivalent to $\Amf',a\leq_{\hALC} \Bmf',d$ by
  Theorem~\ref{thm:ehrenhorn}. 
  
  \smallskip Theorem~\ref{thm:ehrenhorn} implies (4).

  \smallskip For (5), observe that equivalence is just mutual entailment. 
  
  \smallskip For (6), let $\Amf,\Bmf,a,b$ be an instance of
  entailment. Construct a new structure $\Amf'$ by adding two fresh
  elements $a'$ and $b'$ to the disjoint union $\Amf\cup\Bmf$ of \Amf
  and \Bmf and making $a'$ an $R$-predecessor of
  both $a$ and $b$, and $b'$ an $R$-predecessor of $b$. It is routine
  to verify that: 
  $$\Amf,a\leq_{\hALC}\Bmf, b \quad\text{iff}\quad
  \Amf',a'\equiv_{\hALC} \Amf',b'.$$
\end{proof}

\medskip\noindent\textbf{Lemma~\ref{lem:hornsim}.}\textit{
  HornSim is \ExpTime-complete. 
}

\begin{proof}
  For the upper bound, we provide an alternating algorithm which
  essentially implements the Horn simulation game and requires only
  polynomial space. Let $(X_0,b_0)$ be the input. The algorithm
  proceeds in rounds and maintains a pair $(X,d)$ with $X\subseteq \text{dom}(\Amf)$
  and $d\in \text{dom}(\Bmf)$. At pair $(X,d)$, the algorithm proceeds as follows:
  \begin{itemize}

    \item For every $R\in \tau$, and every $Y$ with
      $X R^{\Amf\uparrow} Y$, guess non-empty $Y'\subseteq Y$ and $d'$ with
      $(d,d')\in R^\Bmf$ and proceed with the pair $(Y',d')$.

    \item For every $R\in \tau$, and every $d'$ with $(d,d')\in
      R^\Bmf$, guess non-empty $Y$ with $X R^{\Amf\downarrow} Y$ and proceed
      with the pair $(Y,d')$.

  \end{itemize}
  Note the similarity of the two points above with
  properties~\itforthh and~\itbackh, respectively. The
  algorithm rejects if $(X,d)$ does not satisfy~\itatomich
  or~\itsimb at some stage, or it fails guessing $Y'$ or $Y$,
  respectively, in the two points above; it accepts after
  $2^{|\text{dom}(\Amf)|}\cdot |\text{dom}(\Bmf)|$ rounds. It remains to observe that the algorithm
  obviously requires only polynomial space and that both~\itatomich
  and~\itsimb can be checked in polynomial time.
  
  \medskip For the lower bound, we reduce the word problem for
  polynomially space-bounded, alternating Turing machines (ATMs),
  similar to~\cite{HarelKV02}. An ATM is a tuple
  $M=(Q_e,Q_u,\Sigma,\Gamma,q_0,\mapsto,F_{\mathit{rej}},F_{\mathit{acc}})$ where
  $\Sigma$ is the input alphabet, $\Gamma$ is the tape alphabet,
  $Q_e$, $Q_u$, $F_{\mathit{rej}}$, and $F_{\mathit{acc}}$ are pairwise disjoint sets of
  existential, universal, rejecting, and accepting states,
  respectively. We denote the set of all states with $Q$, the set of
  all rejecting and accepting states with $F$, and assume that the
  initial state $q_0$ is universal, that is, $q_0\in Q_u$. The transition relation
  ${\mapsto}\subseteq Q\times\Gamma\times
  Q\times\Gamma\times\{L,R,H\}$ has binary branching degree for all
  $q\in Q_e\cup Q_u$, that is, every pair $(q,a)$ has precisely two
  successors.  Accordingly, we use $(q,a)\mapsto
  \langle(q_l,b_l,d_l),(q_r,b_r,d_r)\rangle$ to indicate that after
  reading symbol $a$ in state $q$, the TM can branch left with
  $(q_l,b_l,d_l)$ and right with $(q_r,b_r,d_r)$. We further assume
  that, for all $q\in F$ and $a\in \Gamma$, we have
  $(q,a)\mapsto(q,a,H)$, that is, once a final state $q\in F$ is
  reached, $M$ loops in the same configuration.

  Acceptance of such TMs is defined as follows. The possible
  computations of $M$ on some input word $w$ induce an AND-OR graph
  whose nodes are $M$'s configurations and whose edges correspond to
  successor configurations. With each node in the graph,
  we associate an acceptance value as follows: 
  \begin{itemize}

    \item accepting (respectively, rejecting) configurations are
      associated with value $1$ (respectively, $0$);

    \item a universal configuration is associated with the minimum
      value associated with one of its successor configurations;

    \item an existential configuration is associated with the maximum
      value associated with one of its successor configurations.

  \end{itemize}
  It is well-known that we can assume that in $M$ from every possible
  start configuration, we always reach final configurations, so the
  above is well-defined. The ATM $M$ accepts an input word $w$ iff
  the initial configuration is associated with value $1$.

  For the reduction, let $M$ be an $s(n)$-space bounded ATM, for some
  polynomial $s(n)$ and $w$ some input word of length $n$. We
  construct structures $\Amf$, $\Bmf$, a set $X\subseteq
  \text{dom}(\Amf)$, and $b\in \text{dom}(\Bmf)$ such that
  $$M\text{ accepts }w\text{\quad iff\quad $\Amf,X\preceq_{\textit{horn}}\Bmf,b$.}$$
  The structure $\Amf$ can be thought of as a disjoint union of $s(n)$
  structures $\Amf_1,\ldots,\Amf_{s(n)}$ and a single copy of \Bmf.
  Intuitively, each sub-structure $\Amf_i$ is responsible for a tape
  cell of one of $M$'s configurations on input $w$, and different tape
  cells are synchronized via the simulation conditions using different
  role names. We include a copy \Bmf in \Amf due to technical
  reasons made clear below. The role of \Bmf (as the second structure) is to
  control $M$'s computation by enforcing the right association of
  values with configurations in the AND-OR graph; it is essentially
  independent of $M$.

  The vocabulary consists of the following symbols:
  \begin{itemize}

    \item concept names $U_{0}, U_1, U_{\swarrow}, U_{\searrow}$,
      $V_1,\ldots,V_{s(n)}$, and

    \item role names $R_{id}$ and $R_{q,a,i,d}$ for every
      $q\in Q$, $a\in \Gamma$, $i\in\{1,\ldots,s(n)\}$, and
      $d\in\{\swarrow,\searrow\}$.

  \end{itemize}

  We start with \Bmf. Its domain $\text{dom}(\Bmf)$ consists of 20 elements: 16
  internal elements of shape $(\ast,l,r,val,d)$, where
  ${\ast}\in\{\wedge,\vee\}$, $l,r,val\in\{0,1\}$ satisfy $l\ast
  r\equiv val$, and $d\in \{\swarrow,\searrow\}$, and final elements
  $(val,d)$ with $val\in\{0,1\}$ and $d\in\{\swarrow,\searrow\}$. More
  precisely, we have for all $d\in \{\swarrow,\searrow\}$ the
  following elements:
  \begin{align*}
    (\wedge,0,0,0,d), (\wedge,0,1,0,d), (\wedge,1,0,0,d),
    (\wedge,1,1,1,d), \\
    (\vee,0,0,0,d), (\vee,0,1,1,d), (\vee,1,0,1,d), (\vee,1,1,1,d),
    \\
    (0,d),(1,d).
  \end{align*}
  Note that elements of the shape $(\ast,l,r,val,d)$ match the truth
  table entries for the operator $\ast$ in the sense that $l\ast
  r\equiv val$. We refer with \emph{universal} and \emph{existential}
  elements to elements of the shape $(\wedge,\cdots)$ and
  $(\vee,\cdots)$, respectively.

  For the concept names, we take
  \begin{align*}
    U_{\swarrow}^\Bmf & = \{ (\ast,l,r,val,d)\in B\mid
    d=\swarrow\}\cup\{(0,\swarrow),(1,\swarrow)\}, \\
    U_{\searrow}^\Bmf & = \{ (\ast,l,r,val,d)\in B\mid
    d=\searrow\}\cup\{(0,\searrow),(1,\searrow)\}, \\
    U_0^\Bmf & = \{(0,\swarrow),(0,\searrow)\}, \\
    U_1^\Bmf & = \{(1,\swarrow),(1,\searrow)\}, \\
    V_i^\Bmf & = \emptyset\hspace{2cm}\text{for all
      $i\in\{1,\ldots,s(n)\}$}.
  \end{align*}
  For the role names of shape $R_{q,a,i,d}$, we take
  \begin{itemize}

    \item $( (\ast,l,r,val,d),(\ast',l',r',val',d'))\in
      R_{q,a,i,d}^\Bmf$ iff $\ast\neq\ast'$ and either $val'=l$ and
      $d'={\swarrow}$ or $val'=r$ and $d'={\searrow}$. 

  \end{itemize}
  That is, we switch between existential and universal elements and require that the next
  value $val'$ coincides with the current $l$ or $r$ depending on the
  branch. Everything is independent of the values of
  $q,a,i$, but depends on $d$. We further have 
  \begin{itemize}

    \item $((\ast,l,r,val,d),(val',d'))\in R_{q,a,i,d}^\Bmf$ iff either $val'=l$ and
      $d'={\swarrow}$ or $val'=r$ and $d'={\searrow}$.

  \end{itemize}
  Finally, for the role name $R_{id}$, $R^\Bmf_{id}$ is the
  identity on the final elements, that is, $\{( b,b)\mid
    b\in \text{dom}(\Bmf)\text{ of shape $(val,d)$}\}$.

  \smallskip
  We continue with structures $\Amf_i$, for each $i\in\{1,\ldots,s(n)\}$.
  The domain $\text{dom}(\Amf_i)$ of each $\Amf_i$  consists of all possible contents
  of a cell, extended with a direction, that is,
  $$\text{dom}(\Amf_i) = (\Gamma\cup(Q\times\Gamma))\times\{\swarrow,\searrow\}.$$
  For the concept names, we take:
  \begin{align*}
    U_{\swarrow}^{\Amf_i} & = (\Gamma\cup
    (Q\times\Gamma))\times\{\swarrow\}, \\
    U_{\searrow}^{\Amf_i} & = (\Gamma\cup
    (Q\times\Gamma))\times\{\searrow\}, \\
    U_0^{\Amf_i} & =
    F_{\mathit{rej}}\times\Gamma\times\{\swarrow,\searrow\}, \\ 
    U_1^{\Amf_i} & =
    F_{\mathit{acc}}\times\Gamma\times\{\swarrow,\searrow\}, \\
    V_j^{\Amf_i} & = \begin{cases} \text{dom}(\Amf_i) & \text{if $i\neq j$} \\
      \emptyset & \text{otherwise}\end{cases} \quad\quad\text{for all
      $j\in\{1,\ldots,s(n)\}$}.
  \end{align*}
  For a role name $R_{q,a,j,d}$, intuitively
  $R_{q,a,j,d}^{\Amf_i}$ contains all pairs $(\gamma,\gamma')$ such
  that there is configuration $\alpha$ in state $q$, with head
  position $j$, reading $a$, (and previous direction $d$), and a
  successor configuration $\alpha'$ of $\alpha$, in which cell $i$
  with content $\gamma$ has been updated to $\gamma'$.  More
  precisely, we include in $R_{q,a,j,d}^{\Amf_i}$ for all $(q,a)\to
  \langle (q_l,b_l,d_l),(q_r,b_r,d_r)\rangle$:
  \begin{itemize}

    \item if $d_l=L$, the following pairs: 
      \begin{itemize}

	\item if $i=j$, then $((q,a,d),(b_l,\swarrow))$, 

	\item if $i=j+1$, then $((b,d),(q_l,b,\swarrow))$ for all
	  $b\in \Gamma$,

	\item if $i\notin\{j,j+1\}$, then $((b,d),(b,\swarrow))$
	  for all $b\in\Gamma$;
      \end{itemize}

    \item if $d_l=H$, the following pairs:
      \begin{itemize}

	\item if $i=j$, then $((q,a,d),(q_l,b_l,\swarrow))$,

	\item if $i\neq j$, then $( (b,d),(b,\swarrow))$, for all
	  $b\in\Gamma$;

      \end{itemize}

    \item if $d_l=R$, the following pairs:
      \begin{itemize}

	\item if $i=j$, then $((q,a,d),(b_l,\swarrow))$,

	\item if $i=j-1$, then $((b,d),(q_l,b,\swarrow))$ for all
	  $b\in \Gamma$,

	\item if $i\notin\{j,j-1\}$, then $( (b,d),(b,\swarrow))$ for
	  all $b\in\Gamma$;

      \end{itemize}

    \item the cases for $d_r=L$, $d_r=H$, and $d_r=R$ are
      obtained by replacing above $\swarrow$, $b_l$, and $q_l$ with $\searrow$,
      $b_r$, and $q_r$, respectively.

  \end{itemize}
  Finally, for the role name $R_{id}$, $R_{id}^{\Amf_i}$ is the 
  identity relation on the set $\Gamma\cup(F\times\Gamma)$.

    The structure $\Amf$ is now constructed by first taking the disjoint
  union of all $\Amf_i$, $i\in\{1,\ldots,s(n)\}$ and \Bmf, and then
  adding connections between elements from $\Amf_i$ and \Bmf. For all
  $c\in \text{dom}(\Amf_i)$, let us denote with $(c,i)$ the corresponding copy of $c$
  in $\text{dom}(\Amf)$. We add the following connections: 
  \begin{itemize}

    \item $((a,d,i),b)\in R^\Amf$, for every $(a,d,i)\in
      \text{dom}(\Amf)$, 
      $b\in \text{dom}(\Bmf)$, and every role name $R$;

    \item $((q,a,d,i),b)\in R^\Amf$, for every for every role name $R$
      and every $(q,a,d,i)\in \text{dom}(\Amf)$ and $b\in
      \text{dom}(\Bmf)$ such that either $q\in Q_e$ and $b$ is
      universal or $q\in Q_u$ and $b$ is existential.

  \end{itemize}
  Let us point out two important properties of the structure \Amf.
  First, by the connections of the $\Amf_i$ with the copy of \Bmf in \Amf, we have
  that:
  \begin{itemize}

    \item[(P1)] $\Bmf,b\preceq_{\textit{sim}} \Amf,x$ for every universal element $b\in\text{dom}(\Bmf)$ 
      and every element $x$ of shape $(a,d,i)\in \text{dom}(\Amf)$ and every element
      $x=(q,a,d,i)\in \text{dom}(\Amf)$ with $q\in Q_u$, and similarly, $\Bmf,b\preceq_{\textit{sim}} \Amf,x$
      for every
      existential element $b\in \text{dom}(\Bmf)$ and every $x$ of shape
      $(a,d,i)\in \text{dom}(\Amf)$ and every element $x=(q,a,d,i)\in
      \text{dom}(\Amf)$ with $q\in
      Q_e$. Even more, the witnessing simulation is of the form
      $\{( (a,d,i),b)\}\cup \{(b',b')\mid b'\in \text{dom}(\Bmf)\}$.
  
  \end{itemize}
  The second property concerns the synchronization of the different
  $\Amf_i$. For formulating the property, it is convenient to
  associate with a direction $d\in\{\swarrow,\searrow\}$ and some
  configuration $\alpha$ of $M$ a set $X_{\alpha,d}\subseteq
  \text{dom}(\Amf)$ in the following natural way. If $\alpha$ is the configuration
  $b_1\cdots b_{s(n)}$, then $X_{\alpha,d}$ is the set
  $$X_{\alpha,d}=\{(b_i,d,i)\mid i\in\{1,\ldots,s(n)\}\}.$$
  We now have the following property:
  \begin{itemize}

    \item[(P2)] For every configuration $\alpha$ in state $q$ with
      symbol $a$ at head position $i$, there are precisely two role
      names $R$ such that every $(c,j)\in X_{\alpha,d}$ (for any
      $d$) has an $R$-successor in $\Amf$, namely $R_{q,a,i,\swarrow}$
      and $R_{q,a,i,\searrow}$.  Moreover, if we move jointly to these
      successors, we arrive at $X_{\alpha_l,\swarrow}$ and
      $X_{\alpha_r,\searrow}$, respectively, where $\alpha_l,\alpha_r$
      are the two successor configurations of $\alpha$ according to
      $\mapsto$.

  \end{itemize}

  \smallskip
  Let $\alpha_0$ be $M$'s initial configuration on input $w$ and $\hat
  b$ the element $(\wedge,1,1,1,\swarrow)$ in $\Bmf$. Based on the
  insights~(P1) and~(P2) given above, we can verify
  correctness of the reduction.

  \smallskip\noindent\textit{Claim.} $M$ accepts $w$ iff
  $\Amf,X_{\alpha_0,\swarrow}\preceq_{\textit{horn}}\Bmf,\hat b$.

  \smallskip\noindent\textit{Proof of the Claim.} ``$\Rightarrow$'' We
  define a Horn simulation $Z$ guided by the AND-OR graph induced by
  the computation of $M$ on input $w$. Let $\Cmc$ be the set of all
  configurations of $M$ and denote with $\ell(\alpha)$ the associated
  acceptance value of a configuration $\alpha\in \Cmc$. Define
  $Z$ as follows: 
  \begin{itemize}

    \item $X_{\alpha,d}Z(\ast,l,r,val,d)$ for all $d\in
      \{\swarrow,\searrow\}$ such that:
      \begin{itemize}

	\item either $\alpha$ is universal and $\ast=\wedge$ or
	  $\alpha$ is existential and $\ast=\vee$,

	\item $val=\ell(\alpha)$ and $\alpha$ has a left successor
	  configuration $\alpha_l$ with $l=\ell(\alpha_l)$ and a right
	  successor configuration $\alpha_r$ with $r=\ell(\alpha_r)$.

      \end{itemize}

    \item $X_{\alpha,d}Z(val,d)$ for all $d\in\{\swarrow,\searrow\}$
      such that $\alpha$ is a final configuration with
      $val=\ell(\alpha)$.

    \item $\{b'\} Z b'$ for all $b'\in \text{dom}(\Bmf)$,

  \end{itemize}
  Note that we clearly have $X_{\alpha_0,\swarrow} Z \hat b$. It thus
  remains to verify that $Z$ is indeed a Horn simulation, that is, it
  satisfies Conditions~\itatomich, \itsimb, \itforthh, and
  \itbackh.
  \begin{itemize}

    \item Condition~\itatomich is satisfied due to the definition of
      $U^\Amf$ for the concept names $U$. 

    \item Condition~\itforthh is a consequence
      of~(P1) above.

    \item For Condition~\itforthh, consider a pair $X_{\alpha,d} Z
      (\ast,l,r,val,d)$, that is, $\alpha$ is not a final
      configuration; the other cases are similar. Further assume that
      $X_{\alpha,d} (R^\Amf)^\uparrow Y$. By~(P2), we know
      that $R$ is of shape $R_{q,a,i,d'}$,
      $d'\in\{\swarrow,\searrow\}$, where $q$ is the state of
      $\alpha$, $i$ is its head position, and $a$ is the symbol of the
      head.  Moreover, (P2) implies that
      $Y=X_{\alpha_l,\swarrow}$ or $Y=X_{\alpha_r,\searrow}$
      (depending on $d'$), where $\alpha_l$, $\alpha_r$ are left and
      right successor configurations of $\alpha$. It remains to note
      that, by definition of $Z$, there is some $b'\in
      \text{dom}(\Bmf)$ such that $YZb'$.

    \item Condition~\itbackh is also a consequence
      of~(P1). Indeed, let $X_{\alpha,d}Z(\ast,l,r,val,d)$ and
      $b$ some $R$-successor of $(\ast,l,r,val,d)$ in \Bmf.
      By~(P1), $b$ is an $R$-successor $x'$ of some (actually:
      all) $x\in X_{\alpha,d}$. By definition of $Z$, we have
      $\{b\}Zb$.

  \end{itemize}
  
  ``$\Leftarrow$'' Let $Z$ be a Horn simulation with
  $X_{\alpha_0,\swarrow}Z\hat b$. We show that $M$ accepts $w$ by
  constructing the (relevant subset of the) AND-OR graph induced by
  $M$'s computation on input $w$. We first show the following for all
  $X_{\alpha,d} Z b$: 
  \begin{enumerate}

    \item $\alpha$ is an accepting (rejecting) configuration iff $b$
      is an accepting (rejecting) element in \Bmf.

    \item if $\alpha$ is not final, then also $X_{\alpha_l,\swarrow} Z b_l$ and
      $X_{\alpha_r,\searrow} Z b_r$ for some $b_l$ and $b_r$, where
      $\alpha_l$, $\alpha_r$ are the left and right successor
      configurations of $\alpha$. 

  \end{enumerate} 
  It follows from the definition of $R_{id}^\Amf$ and $R_{id}^\Bmf$
  and the fact that $Z$ satisfies Conditions~\itforthh and~\itbackh, that
  $\alpha$ is a final configuration iff $b$ is a final element.
  Property~1) then follows from Condition~\itsimb and the definition of
  $U_{0/1}^\Bmf$ for final elements. 

  For Property~2), let $Y$ be some set with
  $X_{\alpha,d}(R^\Amf)^\uparrow Y$ for some $R$.  By~(P2), we
  know that $R$ is of shape $R_{q,a,i,d'}$,
  $d'\in\{\swarrow,\searrow\}$, where $q$ is the state of $\alpha$,
  $i$ is its head position, and $a$ is the symbol of the head.
  Moreover, $Y=X_{\alpha_l,\swarrow}$ or $Y=X_{\alpha_r,\searrow}$
  (depending on $d'$). By Condition~\itforthh, there is a
  non-empty $Y'\subseteq Y$ and some $b'$ with $(b,b')\in R^\Bmf$ such
  that $Y' Z b'$. By Condition~\itatomich and the definition of
  $V_i^\Amf$ and $V_i^\Bmf$, we have $Y'=Y$. 

  Now, $X_{\alpha_0,\swarrow} Z \hat b$ and Property~2) imply that for
  every configuration $\alpha$ reachable from $\alpha_0$ we have
  $X_{\alpha} Z b$ for some $b$. We claim that in this case the value
  $val$ of $b$ is the value associated to the configuration $\alpha$ in the
  AND-OR graph. Indeed, for final configurations this is a consequence
  of Property~1) above. For the remaining configurations, this is a
  consequence of the definition of $R_{q,a,i,d}^\Bmf$. It remains to
  note that this implies that $\alpha_0$ is associated with value
  $1$ since the value $val$ in $\hat b$ is $1$.
  This finishes the proof of the Claim.

  \medskip Since the construction of \Amf and \Bmf can be carried out in
  polynomial time, this establishes \ExpTime-hardness of HornSim.
\end{proof}

For the restricted problems $\ell$-entailment, $\ell$-equivalence, and
$\ell$-HornSim, we prove the following Lemma analogously to 
Lemma~\ref{lem:reductions}.

\begin{lemma} \label{lem:bounded-reductions}
  There are the following reductions:
  \begin{enumerate}

    \item[(1)] $\text{$\ell$-CBE}\leq_T^P\ell\text{-}\text{HornSim}$;

    \item[(2)] $\overline{\ell\text{-HornSim}}\leq_m^P\text{$\ell$-CBE}$;

    \item[(3)] $\ell\text{-HornSim}\leq_m^P \text{$\ell$-Entailment}$;

    \item[(4)] $\text{$\ell$-Entailment}\leq_m^P \ell\text{-HornSim}$;

    \item[(5)] $\text{$\ell$-Equivalence}\leq_T^P \text{$\ell$-Entailment}$;

    \item[(6)]
      $\text{$\ell$-Entailment}\leq_m^P\text{$(\ell+1)$-Equivalence}$.

  \end{enumerate}
\end{lemma}

Thus, it suffices to establish the following lemma for the complexity
of $\ell$-HornSim to finish the proof of Theorem~\ref{thm:complexity}.

\begin{lemma} \label{lem:k-hornsim}
  $\ell$-HornSim is \PSpace-complete for unary encoding of
  $\ell$ and \ExpTime-complete for binary encoding.
\end{lemma}

\begin{proof}
  For the upper bounds, observe that we can use the alternating algorithm given in the proof of
  Lemma~\ref{lem:hornsim} and run it for $\ell$ rounds.  It
  remains to observe that it is an alternating, polynomially time
  bounded procedure in case of unary encoding, and an alternating,
  polynomially space bounded procedure in case of binary encoding.
  The \PSpace and \ExpTime upper bounds follow. 
 
  For the \ExpTime-lower bound, we take the same reduction as in the
  proof of Theorem~\ref{lem:hornsim}, but add an input $\ell$
  specifying the maximum time until a final state is reached. It is
  well-known that we can assume without loss of generality that this
  happens after $2^{O(s(n))}$ steps. Binary encoding of $\ell$ yields the
  result. 

  The \PSpace lower bound follows a similar strategy. We reduce the
  word problem for polynomially \emph{time} (instead of space) bounded
  ATMs. There is a fixed such ATM $M$ with a \PSpace-hard word
  problem; let $p(n)$ be the polynomial bound on the time of $M$.  The
  reduction is now as in Lemma~\ref{lem:hornsim} except
  that we replace $s(n)$ with $p(n)$, and on input $w$, we
  additionally set $\ell=p(|w|)$.
\end{proof}

We finish this section with TBox distinguishability. Let us denote
with \emph{TBox entailment} and \emph{TBox equivalence}, the problems
whether $\Amf\leq_{\hALC}\Bmf$ and $\Amf\equiv_{\hALC}\Bmf$,
respectively, for given structures $\Amf,\Bmf$. Moreover, denote with
\emph{GHornSim} the problem whether $\Amf \preceq_{\textit{horn}}
\Bmf$. Finally, we refer with TBox $\ell$-entailment and TBox
$\ell$-equivalence to the restricted versions.

\begin{theorem} \label{thm:complexity-global}
  TBox entailment and equivalence are \ExpTime-complete. Moreover,
  TBox $\ell$-entailment and TBox $\ell$-equivalence are \ExpTime-complete for
  binary encoding of $\ell$ and \PSpace-complete for unary encoding. 
\end{theorem}

As we have the following relationships, it suffices to prove
Lemma~\ref{lem:ghornsim} below. 

\begin{lemma} \label{lem:reductions-tbox}
  \begin{enumerate}

    \item[(1)] TBox entailment is equivalent to GHornSim;

    \item[(2)] TBox $\ell$-entailment is equivalent to $\ell$-GHornSim;

    \item[(3)] TBox equivalence $\leq_{T}^P$ TBox entailment;

    \item[(4)] TBox $\ell$-equivalence $\leq_{T}^P$ TBox $\ell$-entailment;

    \item[(5)] TBox entailment $\leq_m^P$ TBox equivalence.

    \item[(6)] TBox $\ell$-entailment $\leq_m^P$ TBox $\ell$-equivalence.

  \end{enumerate}
\end{lemma}

\begin{proof}
  Properties~(1) and~(2) are a consequence of
  Theorem~\ref{thm:ehrenhorn-tbox}.  Properties~(3) and~(4) is due to
  the fact that equivalence is just mutual entailment.  Properties~(5)
  and~(6) follow from the fact that $\Amf\leq_{\hALC}^{(\ell)}\Bmf$
  iff $\Amf\cup\Bmf\equiv^{(\ell)}_{\hALC}\Bmf$, for the disjoint
  union $\Amf \cup \Bmf$ of $\Amf$ and $\Bmf$.
\end{proof}

\begin{lemma} \label{lem:ghornsim}
  GHornSim is \ExpTime-complete, and $\ell$-GHornSim is
  \ExpTime-complete if $\ell$ is given in binary, and \PSpace-complete
  if $\ell$ is given in unary. 
\end{lemma}

\begin{proof}
  
  The upper bound follows from the definition of $\Amf \preceq_{\textit{horn}}\Bmf$
  and the complexity of Horn simulations and restricted Horn
  simulations, see Lemma~\ref{lem:hornsim} and~\ref{lem:k-hornsim},
  respectively. 

  For the \ExpTime-lower bound for GHornSim, observe that the proof of
  Lemma~\ref{lem:hornsim} is easily adapted so as to show
  \ExpTime-hardness also in this case. Indeed, we can mark $\hat b$
  and all elements in $X_{\alpha_0}$ with a fresh concept name $S^*$.
  Now, for every element $d\in \text{dom}(\Bmf)\setminus\{\hat b\}$
  there is a copy $d'$ of $d$ in \Amf, for which trivially
  $\Amf,d'\leq \Bmf,d$. It remains to note that, by construction, the
  only candidate set for $\hat b$ is in fact $X_{\alpha_0}$. 

  Finally, \ExpTime and \PSpace-hardness for the restricted problem is
  obtained from the above result in the same way as in
  Lemma~\ref{lem:k-hornsim}.
\end{proof}

\section*{Proofs for Section: Expressive Completeness for \hALC{}}
We provide proofs for the implications (3) $\Rightarrow$ (4) of Theorems~\ref{thm:mainconcept}
and \ref{thm:maintbox}.
\begin{lemma}
If an $\mathcal{ALC}$-concept $C$ of depth $\leq \ell$ is preserved under $\ell$-Horn simulations, 
then $C$ is equivalent to a \hALC{}-concept of depth $\leq \ell$ \textup{(}also in the setting of finite model theory\textup{)}.
\end{lemma}
\begin{proof}
Let $C$ be an $\mathcal{ALC}$ concept of depth $\leq \ell$ that is preserved under $\ell$-Horn simulations. 
We use the set Horn$_{\ell}$ of concepts and the concept $\rho_{\Bmf,\ell,b}$ defined in the proof of Lemma~\ref{lem:diagram}. 
Let $D$ be the conjunction of all $H\in \text{Horn}_{\ell}$ with $\emptyset \models C \sqsubseteq H$. 
It suffices to show that $\emptyset \models D \sqsubseteq C$ establishing that $D$ and $C$ are equivalent. 
To this aim, observe that $D$ is equivalent to $\rho_{\Amf_{u},\ell,X}$, where $\Amf_{u}$ is the disjoint union of all 
finite structures (up to isomorphisms) and $X=C^{\Amf_{u}}$. To see this, suppose  first that $H\in \text{Horn}_{\ell}$ and 
$\emptyset \models C \sqsubseteq H$. Then $X=C^{\Amf}\subseteq H^{\Amf}$ for all structures $\Amf$. Thus, $H$ is a conjunct of 
$\rho_{\Amf_{u},\ell,X}$. Conversely, suppose $H$ is a conjunct of $\rho_{\Amf_{u},\ell,X}$. Then $X\subseteq H^{\Amf_{u}}$. 
Thus, $C^{\Amf_{u}}\subseteq H^{\Amf_{u}}$. But then $\emptyset \models C \sqsubseteq H$ since the latter is equivalent to 
the condition that $C^{\Amf}\subseteq H^{\Amf}$ for all \emph{finite} structures $\Amf$ and $\mathcal{ALC}$ has the finite model property.

Now suppose $\Bmf,b$ is a pointed structure with $b\in D^{\Bmf}$. Then $b\in \rho_{\Amf_{u},\ell,X}^{\Bmf}$. By Lemma~\ref{lem:diagram}, 
there exists $X_{0}\subseteq X$ such that $\Amf_{u},X_{0} \preceq_{\textit{horn}}^{\ell} \Bmf,b$. Then $b\in C^{\Bmf}$ follows from 
the assumption that $C$ is preserved under $\ell$-Horn simulations.
The finite model theory version follows directly using the finite model property of $\mathcal{ALC}$.
\end{proof}

\begin{lemma}
Let $\Tmc$ be an $\mathcal{ALC}$ TBox of depth $\ell$ preserved under global $\ell$-Horn simulations.
Then $\Tmc$ is equivalent to a \hALC{} TBox of depth $\leq \ell$ (also in the setting of finite model theory).
\end{lemma}
\begin{proof}
Let $\Tmc'$ be the set of \hALC{} CIs $\top \sqsubseteq H$ with $H\in \text{Horn}_{\ell}$
such that $\Tmc\models \top \sqsubseteq H$. We show that $\Tmc'\models \Tmc$.
Take a model $\Bmf$ of $\Tmc'$.
Take for every \hALC{} CI $\top \sqsubseteq H$ of depth $\leq \ell$ and $b\in \text{dom}(\Bmf)\setminus H^{\Bmf}$
a model $\Amf_{b,H}$ of $\Tmc$ and $a_{b,H}\in \text{dom}(\Amf_{b,H})$ with $a_{b,H}\not \in H^{\Amf_{b,H}}$.
Such a model exists since $\mathcal{T}' \not\models \top \rightarrow H$ and so, by the definition of $\mathcal{T}'$, 
$\mathcal{T} \not\models \top \rightarrow H$.

Let $\Amf$ be the disjoint union of all $\Amf_{b,H}$ and let $X_{b}$ be the set of all $a_{b,H}$.
Then $X_{b} \subseteq H^{\Amf}$ implies $b\in H^{\Bmf}$, for every \hALC{}-concept $H$
of depth $\leq \ell$. Thus, by Lemma~\ref{lem:diagram}, there exists $Y_{b}\subseteq X_{b}$ such that
$\Amf,Y_{b}\preceq_{\text{horn}}^{\ell} \Bmf,b$. Then $\Bmf$ is model of $\Tmc$ since $\Amf$ is a model of
$\Tmc$ and $\Tmc$ is preserved under global $\ell$-Horn simulations.

In the finite model theory setting it suffices to consider finite models $\Bmf$ of $\Tmc'$ since $\mathcal{ALC}$
has the finite model property for TBox reasoning. For the same reason one can always choose finite $\Amf_{b,H}$.
\end{proof}

\section*{Proofs for Section on Guarded Fragment}

\subsection{Basic Properties of \hGF}

\begin{theorem} \label{thm:horngf-complexity}
  Satisfiability in \hGF$[\tau]$ is \ExpTime-complete if the arity
  of predicate names in $\tau$ is bounded, and \TwoExpTime-complete if
  not.
\end{theorem}

\begin{proof}
  The upper bounds follow from the guarded fragment
  GF~\cite{DBLP:journals/jsyml/Gradel99}.  The lower bound in the
  bounded arity case is inherited from \hALC{}~\cite{DBLP:journals/tocl/KrotzschRH13}.
  For the case of unbounded arity, the lower bound is obtained as a
  straightforward adaptation of the \TwoExpTime lower bound for GF
  in~\cite{DBLP:journals/jsyml/Gradel99}. 
\end{proof}

Let us fix a vocabulary $\tau$. A \emph{tuple-generating dependency (tgd)
over $\tau$} is an FO formula of the form 
$$
\forall \xbf\forall \ybf\,(\psi(\xbf,\ybf)\rightarrow \exists \zbf\, \varphi(\xbf,\zbf))
$$ 
where $\psi(\xbf,\ybf)$ and $\varphi(\xbf,\zbf)$ are conjunctions of atoms
over $\tau$ and are called \emph{body} and \emph{head} of the tgd,
respectively. A tgd is \emph{guarded} if there is an atom in the body
$\psi(\xbf,\ybf)$ which contains all variables $\xbf\cup\ybf$. We
allow as special cases empty body and head and denote this as
$\top\rightarrow \exists \zbf\,\varphi(\zbf)$ and
$\forall\xbf\forall\ybf\, \psi(\xbf,\ybf)\rightarrow\bot$, respectively.

The following theorem captures the relation between
\hGF{} and guarded tgds; essentially, satisfiability and query
answering can be interreduced in polynomial time. 

\begin{lemma} \label{lem:horngf-gtgds}
  For every set $\Sigma$ of guarded tgds over some vocabulary $\tau$,
  we can compute in polynomial time a \hGF{} formula
  $\varphi_\Sigma$ over a larger vocabulary $\tau'\supseteq\tau$ such
  that $\Sigma$ and $\varphi_{\Sigma}$ are equisatisfiable and, for
  every conjunctive query $q$ and database $D$ over $\tau$, we have
  $\Sigma\cup D\models q$ iff $\{\varphi_\Sigma\}\cup D\models q$.
  
  Conversely, every \hGF{} formula $\varphi$
  over $\tau$ can be translated in polynomial time into a set of
  guarded tgds $\Sigma_\varphi$ over a larger vocabulary
  $\tau'\supseteq\tau$ such that $\varphi$ and $\Sigma_{\varphi}$ are
  equisatisfiable and, for every conjunctive query $q$ and database
  $D$ over $\tau$, we have $\{\varphi\}\cup D\models q$ iff
  $\Sigma_\varphi\cup D\models q$.
\end{lemma}

\begin{proof}
  The first statement follows from the arguments
  in~\cite{DBLP:journals/jair/CaliGK13}. We repeat the idea for the
  sake of completeness. Every guarded tgd $\forall \xbf\forall \ybf\,
  (\psi(\xbf,\ybf)\rightarrow \exists \zbf\, \varphi(\xbf,\zbf))$ in
  $\Sigma$ is replaced by two guarded tgds
  \begin{align*}
    \forall \xbf\forall \ybf\, (\psi(\xbf,\ybf) & \rightarrow \exists
    \zbf\, R(\xbf,\zbf)), \\
    \forall \xbf\forall \zbf\, (R(\xbf,\zbf) & \rightarrow
    \varphi(\xbf,\zbf)),
  \end{align*}
  for a fresh predicate name $R$ of appropriate arity. Obviously, the resulting set is a set
  of \hGF{} sentences; their conjunction satisfies the requirements of
  the statement.

  \smallskip Conversely, we give a translation of \hGF{} sentences
  into guarded tgds which can be used to reduce satisfiability and
  query answering.  Let $\varphi_0$ be a \hGF{} sentence.  Let us
  denote with $\varphi[\ybf]$ the fact that a formula $\varphi$ has
  precisely free variables $\ybf$. Further, if $\varphi$ is a
  subformula of $\varphi_0$, the \emph{guard of $\varphi$ in
  $\varphi_0$} is the atom $G(\xbf\ybf)$ that guards all free
  variables of $\varphi$; in case such an atom does not exist, we
  assume that the guard is $\top$. Introduce a fresh predicate name $R_\varphi$
  (of matching arity) for every complex subformula $\varphi$ of $\varphi_0$,
  and set $R_\alpha=\alpha$ whenever $\alpha$ is $\bot$, $\top$, or
  an atomic formula.
   
  Now, let $\Sigma_0$ be the set consisting of the following sentences:
  \begin{itemize}

    \item $\top\to R_{\varphi_0}$,

    \item for every subformula $\psi[\xbf]=\varphi[\ybf]\wedge
      \varphi'[\ybf']$ of $\varphi_0$, include 
      \begin{align*}
	& \forall \xbf\, (R_{\psi}(\xbf)\to
	R_{\varphi}(\ybf)),\text{ and}\\
	& \forall \xbf\, (R_{\psi}(\xbf)\to R_{\varphi'}(\ybf')),
      \end{align*}

    \item for every subformula $\psi[\xbf]=\exists
      \ybf(G(\xbf,\ybf)\wedge \varphi[\zbf])$ of $\varphi_0$, include 
      \begin{align*}
	& \forall\xbf\, (R_{\psi}(\xbf)\to \exists \ybf\, R(\xbf,\ybf)),
	\\
	& \forall \xbf\forall \ybf\, (R(\xbf,\ybf)\to G(\xbf,\ybf)),
	\text{ and}\\
	& \forall \xbf\forall \ybf\, (R(\xbf,\ybf)\to R_{\varphi}(\zbf)),
      \end{align*}
      where $R$ is a fresh predicate name of appropriate arity,

    \item for every subformula $\psi[\xbf]=\forall
      \ybf(G(\xbf\ybf)\to \varphi[\zbf])$ of $\varphi_0$, include 
      $$\forall\xbf\forall \ybf\, (G(\xbf\ybf)\wedge R_{\psi}(\xbf)\to R_{\varphi}(\zbf)),$$

    \item for every subformula $\psi[\xbf]=\lambda[\ybf]\to
      \varphi[\ybf']$ of $\varphi_0$, 
      include $$\forall \xbf\, (R_{\psi}(\xbf)\wedge
      R_{\lambda}(\ybf)\to R_{\varphi}(\ybf')).$$
      Moreover, for every subformula
      $\lambda'[\vbf]=\theta[\zbf]\wedge\theta'[\zbf']$
      of $\lambda$ with guard $G(\zbf_0)$, include 
      $$\forall \xbf\, (G(\zbf_0)\wedge R_{\theta}(\zbf)\wedge
      R_{\theta'}(\zbf')\to R_{\lambda'}(\vbf)),$$
      for every subformula
      $\lambda'[\vbf]=\theta[\zbf]\vee\theta'[\zbf']$
      of $\lambda$ with guard $G(\zbf_0)$, include 
      \begin{align*}
	& \forall \xbf\, (G(\zbf_0)\wedge R_{\theta}(\zbf)\to
	R_{\lambda'}(\vbf)), \text{ and}\\
	& \forall \xbf\, (G(\zbf_0)\wedge R_{\theta'}(\zbf')\to R_{\lambda'}(\vbf)), 
      \end{align*}
      for every subformula $\lambda'[\vbf]=\exists \zbf\, H(\vbf,\zbf) \wedge \theta[\zbf']$
      of $\lambda$ with guard $G(\zbf_0)$, include 
      \begin{align*}
	& \forall \vbf\forall \zbf\, (H(\vbf,\zbf) \wedge
	R_{\theta}(\zbf')\to
	R(\vbf)), \text{ and}\\
	& \forall \zbf_0\, (G(\zbf_0) \wedge R(\vbf)\to
	R_{\lambda'}(\vbf)),
      \end{align*}
      for some fresh predicate name $R$ of appropriate arity.

  \end{itemize}
  It is routine to verify that 
  \begin{itemize} 

    \item[$(\dagger)$] $\Sigma_0$ is satisfiable iff $\varphi_0$ is
      satisfiable and query answering relative to $\varphi_0$ is
      the same as query answering relative to $\Sigma_0$ (over databases
      in the vocabulary of $\varphi_0$).  
  
  \end{itemize}

  Notice that the sentences in $\Sigma_0$ are guarded tgds whenever,
  $\varphi_0$ is equality-free, that is, it does not contain atoms of
  the form $x=y$. Obtain $\Sigma_1$ from $\Sigma_0$ by removing all
  occurrences of $x=y$ on the left-hand side of some rule and, in such
  a case, replacing all occurrences of $y$ with $x$ in the rule.
  Obviously, $\Sigma_1$ still satisfies $(\dagger)$. To remove equality
  atoms from the right-hand side of the rules in $\Sigma_1$, we observe
  first that there is only one atom on the right-hand side of the
  rules in $\Sigma_1$. Obtain a set of rules $\Sigma_2$ by replacing every
  $x=y$ on the right-hand side by $E(x,y)$, for a new predicate name
  $E$, and adding the following guarded tgds, for every predicate name $R$
  appearing in $\Sigma_1$, $\xbf=x_1,\ldots,x_{r_R}$, and every
  $i,j,k$ with $1\leq i<j<k\leq r_R$:
  \begin{align*}
    & \forall \xbf\, (R(\xbf)\wedge
    E(x_i,x_j) \to E(x_j,x_i)), \\
    & \forall \xbf\, (R(\xbf)\wedge
    E(x_i,x_j)\wedge E(x_j,x_k) \to E(x_i,x_k)),\\
    & \forall \xbf\, (R(\xbf)\wedge
    E(x_i,x_j)\to
    R(\xbf[x_i/x_j])\wedge R(\xbf[x_j/x_i])).
  \end{align*}
  Intuitively, these guarded tgds enforce that $E$ behaves like a
  congruence relation under each possible guard $R$. It is routine to
  verify that $\varphi_0$ is satisfiable iff $\Sigma_2$ is satisfiable,
  and that query answering relative to $\varphi_0$ is the same as
  query answering relative to $\Sigma_2$.\footnote{The latter
    equivalence only holds in
  the case without the unique name assumption (UNA). With UNA, the
  equivalence holds under the additional condition that
$\Sigma_2\cup D\not\models E(a,b)$ for any distinct
$a,b\in\mn{dom}(D)$. }
\end{proof}

Consequently, we have: 

\begin{theorem}\label{thm:hgf-ptime-query}
  Query answering in \hGF{} is in \PTime data complexity. 
\end{theorem}

\subsection{Ehrenfeucht-Fra\"{i}ss\'{e} games}

The proof of Theorem~\ref{thm:gfexists} is standard and omitted. It is useful
to define an analogue of the formula $\lambda_{\Amf,\ell,a}$ (in $\mathcal{ELU}$) 
for GF$^{\exists}$. For a finite vocabulary $\tau$, $\ell<\omega$, and variables $\xbf = x_{1}\ldots x_{m}$
we fix a finite set GF$_{\ell}^{\exists}(\tau,m)$ of formulas in GF$^{\exists}[\tau]$
of depth $\leq \ell$ and free variables among $x_{1},\ldots,x_{m}$ such that for every 
GF$^{\exists}[\tau]$ formula $\lambda$ of depth $\leq \ell$ and free variables among 
$x_{1},\ldots,x_{m}$ there exists a 
formula $\lambda'$ in GF$_{\ell}^{\exists}(\tau,m)$
such that $\lambda$ and $\lambda'$ are equivalent.
For a structure $\Amf$ and tuple $\abf$ of length $m$ we denote by $\lambda_{\Amf,\ell,\abf}$
the conjunction of all formulas $\lambda$ in GF$_{\ell}^{\exists}(\tau,m)$ such that $\Amf\models \lambda(\abf)$.
We then have for all pointed structures $\Amf,\abf$ and $\Bmf,\bbf$ and $\ell< \omega$:
\begin{equation}\label{eq:guarded1}
\Bmf\models \lambda_{\Amf,\ell,\abf}(\bbf) \quad \text{iff} \quad \Amf,\abf \preceq_{\textit{gsim}}^{\ell}\Bmf,\bbf.
\end{equation}
We now come to the proof of Theorem~\ref{thm:EhrenhornGF}. As announced, we are going
to prove a stronger version in which the tuple $\abf$ is replaced by a set $X$ of tuples.
\begin{lemma}\label{lem:EhrenhornGF}
For any finite vocabulary $\tau$, pointed $\tau$-structures $\Amf,X$ and $\Bmf,\bbf$, and any 
$\ell < \omega$, we have
$$
\Amf,X \leq_{\textit{hornGF}}^{\ell} \Bmf,\bbf \quad \text{ iff }
\quad \exists X_{0}\subseteq X: \Amf,X_{0} \preceq_{\textit{ghsim}}^{\ell} \Bmf,\bbf
$$
If $\Amf$ and $\Bmf$ are finite, then
$$
\Amf,X \leq_{\textit{hornGF}} \Bmf,\bbf \quad \text{ iff }
\quad \exists X_{0}\subseteq X: \Amf,X_{0} \preceq_{\textit{ghsim}} \Bmf,\bbf
$$
\end{lemma}
\begin{proof}
We first introduce an analogue of the formula $\rho_{\Amf,\ell,X}$ (for \hALC{}) in \hGF{}.
For a finite vocabulary $\tau$, $\ell < \omega$, and variables $\xbf=x_{1}\ldots x_{m}$
we fix a finite set \hGF$_{\ell}(m)$ of formulas in \hGF{} of depth $\leq \ell$ and free variables among 
$x_{1},\ldots,x_{m}$ such that for every 
\hGF{} formula $\varphi$ of depth $\leq \ell$ and free variables among $x_{1},\ldots,x_{m}$ there exists a 
formula $\varphi'$ in \hGF$_{\ell}(m)$ such that $\varphi$ and $\varphi'$ are equivalent.
For a structure $\Amf$ and set $X$ of tuples of the same length $m$ we denote by $\rho_{\Amf,\ell,X}$
the conjunction of all formulas $\varphi$ in \hGF$_{\ell}(m)$ such that $\Amf\models \varphi(\abf)$ for all 
$\abf\in X$. In what follows, when using the formulas $\lambda_{\Amf,\ell,\abf}$ and $\rho_{\Amf,\ell,X}$,
the number $m$ of variables used will always be clear from the context and not be mentioned explicitly.

As the second claim follows directly from the first claim, we prove the first one only.

For the proof of the implication from left to right, let the pointed structures 
$\Amf,X_{0}$ and $\Bmf$, $\bbf_{0}$ be given with $X_{0}$ and $\bbf_{0}$ guarded tuples of the same
length. Assume that $\Amf,X_{0} \leq_{\textit{hornGF}}^{\ell} \Bmf,\bbf_{0}$. We show that
there exists $X_{0}'\subseteq X_{0}$ such that $\Amf,X_{0}' \preceq_{\textit{ghsim}}^{\ell} \Bmf,\bbf_{0}$.

%
Define sets $Z_{k}$ of pairs $(X,\bbf)$ by setting
$(X,\bbf)\in Z_{k}$ if $X\not=\emptyset$ is a set of guarded tuples in $\Amf$ and $\bbf$ 
is a guarded tuple in $\Bmf$ of the same length as the tuples in $X$ such that
\begin{itemize}
\item[(a)] $\Bmf\models \rho_{\Amf,k,X}(\bbf)$;
\item[(b)] $\Amf \models \lambda_{\Bmf,k,\bbf}(\abf)$ for all $\abf\in X$.
\end{itemize}
Observe that if $(X,\bbf)\in Z_{0}$, then $(P,\bbf)$ with 
$$
P= \{p_{\abf}:\bbf\mapsto \abf \mid \abf \in X\}
$$
is a well-defined link (as $\Amf\models \lambda(\abf)$ implies $\Bmf\models \lambda(\bbf)$ for all $\abf\in X$
and $\lambda(\xbf)$ in GF$_{0}^{\exists}$). 
In what follows we do not distinguish between $(X,\bbf)$ and the corresponding link.
We show the following claim by induction over $k<\omega$.

\medskip
\noindent
\emph{Claim 1}. For all $k<\omega$, non-empty sets $X$ of tuples in $\Amf$, and tuples $\bbf$ in $\Bmf$: if $(X,\bbf)\in Z_{k}$,
then $\Amf,X \preceq_{\textit{ghsim}}^{k} \Bmf,\bbf$.

\medskip 
\noindent
For $k=0$, Claim~1 is proved as follows. Assume $(X,\bbf)\in Z_{0}$ and let $(P,\bbf)$ be the corresponding link. 
For Condition~(\atomgh), observe that if $\Amf\models p(b_{i})=p(b_{j})$ for all $p\in P$,
then $\Bmf\models b_{i}= b_{j}$ since $\Bmf\models \rho_{\Amf,0,P[\bbf]}(\bbf)$ and $(x=y)$ are in \hGF$_{0}$.
Similarly, if $R\in \tau$ and $\Amf\models R(p(\bbf))$ for all $p\in P$,
then $\Bmf\models R(\bbf)$ since $\Bmf\models \rho_{\Amf,0,P[\bbf]}(\bbf)$ and $R(\xbf)$ is in \hGF$_{0}$.
Condition~(\simgh) follows from Theorem~\ref{thm:gfexists}. 

Assume Claim~1 has been proved for $k$ and let $(X,\bbf)\in Z_{k+1}$. Denote by $(P,\bbf)$ the corresponding link.
Condition~(\atomgh) can be proved in the same way as for $k=0$.
For Condition~(\forthgh), assume $Y$ is an $R(\bbf_{0},\ybf)$-successor of
$(P,\bbf)$. We have to construct a tuple $\bbf'$ and set of tuples $Y'\subseteq Y$ with $(Y',\bbf_{0}\bbf')\in Z_{k}$.
By definition,
$$
\Amf\models \exists \ybf \big(R(p(\bbf_{0})\ybf)\wedge \rho_{\Amf,k,Y}(p(\bbf_{0})\ybf)\big),
$$  
for all $p\in P$. By Condition (a) and since $\exists \ybf (R(\xbf_{0}\ybf)\wedge \rho_{\Amf,k,Y}(\xbf_{0}\ybf))$
is of depth $\leq k+1$,
$$
\Bmf \models \exists \ybf \big(R(\bbf_{0}\ybf)\wedge \rho_{\Amf,k,Y}(\bbf_{0}\ybf)\big).
$$
Thus, there exists a tuple $\bbf'$ such that 
$$
\Bmf\models R(\bbf_{0}\bbf')\wedge \rho_{\Amf,k,Y}(\bbf_{0}\bbf').
$$
Let 
$$
Y' = \{ \abf \in Y \mid \Amf\models \lambda_{\Bmf,k,\bbf_{0}\bbf'}(\abf)\}
$$
To show that $(Y',\bbf_{0}\bbf')\in Z_{k}$ it suffices to show that $Y'$ is non-empty and satisfies Condition~(a)
(Condition~(b) holds by definition).
\begin{itemize}
\item Assume $Y'=\emptyset$. Then $\Amf\not\models \lambda_{\Bmf,k,\bbf_{0}\bbf'}(\abf)$
for any $\abf \in Y$.
Then $(\lambda_{\Bmf,k,\bbf_{0}\bbf'}\rightarrow \bot)$ is equivalent to a conjunct
of $\rho_{\Amf,k,Y}$. Then 
$$
\Bmf\models (\lambda_{\Bmf,k,\bbf_{0}\bbf'}(\bbf_{0}\bbf')\rightarrow \bot)
$$
by construction of $\bbf_{0}\bbf'$ and we have derived a contradiction.
\item For Condition~(a), assume that $\psi$ is in \hGF{} and of depth $\leq k$ and 
$\Amf \models \psi(\abf)$ for all $\abf\in Y'$. 
We have to show that $\Bmf\models \psi(\bbf_{0}\bbf')$. We have 
$$
\Amf \models (\lambda_{\Bmf,k,\bbf_{0}\bbf'} \rightarrow \psi) (\abf),
$$
for all $\abf\in Y$. Then $(\lambda_{\Bmf,k,\bbf_{0}\bbf'} \rightarrow \psi)$ is a conjunct of
$\rho_{\Amf,k,Y}$ and so 
$$
\Bmf\models (\lambda_{\Bmf,k,\bbf_{0}\bbf'} \rightarrow \psi)(\bbf_{0}\bbf').
$$
From
$$
\Bmf\models \lambda_{\Bmf,k,\bbf_{0}\bbf'}(\bbf_{0}\bbf')
$$
we obtain
$$
\Bmf\models \psi(\bbf_{0}\bbf'),
$$
as required.
\end{itemize}
To show Condition~(\backgh), assume that a guarded tuple $\bbf'$ is given. Consider w.l.o.g. an atomic formula $R(\xbf')$ with 
$\Bmf\models R(\bbf')$. Let $\bbf'= \bbf_{0}\bbf_{1}$ where $\bbf_{0}$ and $\bbf_{1}$ are such that
$[\bbf_{0}]\subseteq [\bbf]$ and $[\bbf_{1}]\cap [\bbf]=\emptyset$.
Take for every $\psi(\xbf_{0}\xbf_{1})$ in \hGF$_{k}$ with $\Bmf\models\neg\psi(\bbf_{0}\bbf_{1})$
a $p_{\psi}\in P$ and tuple $\abf_{\psi}$ such that $\Amf\models (\lambda_{\Amf,k,\bbf'}\wedge \neg\psi)(p_{\psi}(\bbf_{0}),\abf_{\psi})$.
They exist because otherwise 
$$
\Amf\models \forall \xbf_{1}(R(\xbf_{0}\xbf_{1}) \rightarrow 
(\lambda_{\Amf,k,\bbf'}\rightarrow \psi))(p(\bbf_{0}),\xbf_{1})
$$
for all $p\in P$ and so, by definition of $Z_{k+1}$,:
$$
\Bmf\models \forall \xbf_{1}(R(\xbf_{0}\xbf_{1}) \rightarrow 
(\lambda_{\Amf,k,\bbf'}\rightarrow \psi))(\bbf_{0},\xbf_{1}))
$$
and we have derived a contradiction.
Now let $(P',\bbf')$ consist of all $p':\bbf' \mapsto p_{\psi}(\bbf_{0})\abf_{\psi}$ with $p'|_{[\bbf_{0}]}=p|_{[\bbf_{0}]}$.
Then $(P',\bbf')$ is as required for Condition~(\backgh).

Condition~(\simgh) follows from Theorem~\ref{thm:gfexists} and so Claim~1 is proved.

\medskip

It remains to prove that there exists $X_{0}'\subseteq X_{0}$ with
$(X_{0}',\bbf_{0})\in Z_{\ell}$.
Set 
$$
X_{0}' = \{ \abf \in X_{0} \mid \Amf\models \lambda_{\Bmf,\ell,\bbf_{0}}(\abf)\}
$$
Then $(X_{0}',\bbf_{0})$ satisfies Conditions~(a) and (b) by definition.
$X_{0}'\not=\emptyset$ can be proved in the same way as the proof of $Y'\not=\emptyset$ 
given above. 

\medskip

For the proof of the implication from right to left, we show the following claim
by induction over $k<\omega$.

\medskip
\noindent 
\emph{Claim 2.} For all $k< \omega$, if $(P,\bbf) \in Z_{k}$, then $\Amf\models \varphi(p(\bbf))$ for all $p\in P$
implies $\Bmf\models \varphi(\bbf)$, for all formulas $\varphi$ in \hGF{} of depth $\leq k$.

\medskip
For $k=0$, Claim~2 follows from Condition~(\atomgh). Assume Claim~2 has been proved for $k$.
We prove Claim~2 for $k+1$ by induction over the construction of $\varphi$.
Thus, assume that Claim~2 has been proved for $\varphi', \varphi_{1},\varphi_{2}\in \hGF_{k}$, and that
$\varphi\in \hGF_{k+1}$ is of the form $\varphi = \forall \ybf(G(\xbf\ybf) \to \varphi'(\xbf\ybf))$,
$\varphi=  \exists \ybf (G(\xbf\ybf)\wedge \varphi'(\xbf\ybf))$,
$\varphi=\varphi_{1}\wedge \varphi_{2}$,
or $\varphi= \lambda\rightarrow \varphi'$, where $\lambda$ is in GF$^{\exists}_{k+1}$.
Then we prove Claim~2 for $\varphi$. Assume $(P,\bbf) \in Z_{k+1}$ and $\Amf\models \varphi(p(\bbf))$ for all $p\in P$. 

\begin{itemize}
\item Assume $\varphi = \forall \ybf(G(\xbf\ybf) \to \varphi'(\xbf\ybf))$ and for a proof by contradiction that 
$\Bmf\not\models \varphi(\bbf)$.
Choose $\bbf'$ with $\Bmf\models G(\bbf\bbf') \wedge \neg \varphi'(\bbf\bbf')$. 
By Condition~(\backgh), there exists a link $(P',\bbf\bbf')$ in $Z_{k}$ such that for all $p'\in P'$ there exists $p\in P$
with $p|_{[\bbf]}=p'|_{[\bbf]}$. $\varphi'$ has depth $\leq k$ and thus,
as we assume that Claim~2 has been proved for $k$, there exists $p'\in P'$
such that $\Amf\not\models \varphi'(p'(\bbf\bbf'))$. 
There exists $p\in P$ with $p|_{[\bbf]}=p'|_{[\bbf]}$
and so $\Amf\not\models \varphi(p(\bbf))$, and we have derived a contradiction.
\item Assume $\varphi=  \exists \ybf (G(\xbf\ybf)\wedge \varphi'(\xbf\ybf))$
and for a proof by contradiction that $\Bmf\not\models \exists \ybf (G(\bbf\ybf)\wedge \varphi'(\bbf\ybf))$.
Let 
$$
Y= \{ \abf \mid \Amf\models (G \wedge \varphi')(\abf)\}
$$
Then $Y$ is a $G(\bbf,\ybf)$-successor of
$(P,\bbf)$ and so by Condition~(\forthgh) there exists $(P',\bbf\bbf')\in Z_{k}$ such that $P'[\bbf\bbf'] \subseteq Y$.
Thus, by induction hypothesis, $\Bmf\models (G(\bbf\bbf') \wedge \varphi'(\bbf\bbf'))$.
But then $\Bmf\models \varphi(\bbf)$ and we have derived a contradiction.
\item Assume $\varphi= \varphi_{1}\wedge \varphi_{2}$ and for a proof by contradiction that 
$\Bmf\not\models (\varphi_{1}\wedge \varphi_{2})(\bbf)$. We
may assume w.l.o.g. that $\Bmf\not\models \varphi_{1}(\bbf)$. But $\Amf\models \varphi(p(\bbf))$
for all $p\in P$, and so $\Amf\models \varphi_{1}(p(\bbf))$, for all $p\in P$.
By IH, $\Bmf\models \varphi_{1}(\bbf)$, and we have derived a contradiction.
\item Assume $\varphi = \lambda \rightarrow \varphi'$
and for a proof by contradiction $\Bmf\not\models (\lambda \rightarrow \varphi')(\bbf)$.
Then $\Bmf\models \lambda(\bbf)$ and $\Bmf\not\models \varphi'(\bbf)$.
By Theorem~\ref{thm:gfexists}, $\Amf\models \lambda(p(\bbf))$ for all $p\in P$, and, 
by IH, there exists $p\in P$ such that $\Amf\not\models \varphi'(p(\bbf))$.
Then there exists $p\in P$ with $\Amf\not\models \varphi(p(\bbf))$ and we have derived a contradiction.
\end{itemize}
This finishes the proof.%
\end{proof}

\subsection{Model indistinguishability in \hGF}

We introduce decision problems entailment, equivalence, CBE, and
GuardedHornSim similar to Section~\ref{sec:complexity}. For example,
entailment is the problem of deciding $\Amf,\abf\leq_{\hGF}\Bmf,\bbf$
on input $\Amf,\Bmf,\abf,\bbf$, and GuardedHornSim is the following
problem:
\begin{itemize}

  \item \textbf{Input}: structures $\Amf$, $\Bmf$ and a set $X$ of
    guarded tuples in \Amf, and \bbf a guarded tuple in \Bmf

  \item \textbf{Question:} Is $\Amf,X\preceq_{\textit{ghsim}}\Bmf,\bbf$?

\end{itemize}

The main theorem here is the following: 

\medskip\noindent\textbf{Theorem~\ref{thm:gf-complexity}.}\textit{
  In \hGF, entailment, equivalence, and CBE are in \ExpTime. Moreover,
  $\ell$-entailment, $\ell$-equivalence, and $\ell$-CBE are in
  \ExpTime for binary encoding of $\ell$ and in \PSpace for unary
  encoding.
}

\begin{proof}
  
Since entailment, equivalence, and CBE can be reduced to
GuardedHornSim, it suffices to show that GuardedHornSim can be decided
in \ExpTime.

To this end, we start with noting that, for a given guarded tuple $\bbf$, the
number of mappings $p:\bbf\mapsto p(\bbf)$ is bounded by
$|\Amf|$: since $\bbf$ is a guarded tuple, there is some $R\in \tau$
and a tuple \abf such that $\abf\in R^{\Bmf|_{[\bbf]}}$ and
$[\abf]=[\bbf]$. Every possible mapping $p$ maps this atom to a
different atom $R(\abf')$ in \Amf, thus the number of mappings is
bounded by $|\Amf|$. Hence, the size of a witnessing $P$ with
$X=P[\bbf]$ is bounded by $|\Amf|$ as well, and we can try (in
exponential time) all possible~$P$.

It thus remains to check on input $\Amf,\Bmf,(P,\bbf)$ whether there
is a guarded Horn simulation $Z$ between \Amf and \Bmf with $(P,\bbf)\in
Z$. This can be realized using an alternating algorithm which
implements the guarded Horn simulation game. We need an auxiliary notion
and claim. Fix structures \Amf, \Bmf
over some vocabulary $\tau$ and a link $(P_0,\bbf_0)$ between $\Amf$
and \Bmf. Denote with $\mn{ar}(\tau)$ the maximal arity of symbols in
$\tau$. We call a guarded Horn simulation $Z$ \emph{normal for
$(P_0,\bbf_0)$} if, for every $(P,\bbf)\in Z$ with $(P,\bbf)\neq
(P_0,\bbf_0)$, we have $|\bbf|\leq \mn{ar}(\tau)$.  The following is
routine to show:

\smallskip\noindent{\textit{Claim.}} There is a guarded Horn simulation
$Z$ with $(P_0,\bbf_0)\in Z$ iff there is a guarded Horn simulation $Z$
with $(P_0,\bbf_0)\in Z$ which is normal for $(P_0,\bbf_0)$.

\smallskip Based on this claim, we can devise the alternating
algorithm.  Let $\Amf,\Bmf,(P_0,\bbf_0)$ be the input. The algorithm
maintains links and proceeds in rounds. At a link $(P,\bbf)$ it does
the following: 
\begin{enumerate}

  \item Reject if $(P,\bbf)$ does not satisfy Condition~(\atomgh).

  \item Reject if $(P,\bbf)$ does not satisfy Condition~(\simgh). 

  \item For every set of guarded tuples $Y$ and every atomic formula
    $R$ such that $Y$ is an $R(\bbf_0,\ybf)$ successor of $(P,\bbf)$:
    \begin{enumerate}

      \item guess a link $(P',\bbf_0\bbf)$ with $|\bbf_0\bbf|\leq
	\mn{ar}(\tau)$, 

      \item reject if $P'[\bbf_0\bbf']\not\subseteq Y$,

      \item replace $(P,\bbf)$ with $(P',\bbf_0\bbf)$.

    \end{enumerate}

  \item for every guarded tuple $\bbf'$ with $|\bbf'|\leq
    \mn{ar}(\tau)$: 
    \begin{enumerate}

      \item guess a link $(P',\bbf')$ with $\bbf'\leq
	\mn{ar}(\tau)$,

      \item reject if for some $p'\in P'$, there is no $p\in P$ with
	$p|_{[\bbf]\cap[\bbf']}=p'|_{[\bbf\cap[\bbf']]}$;

      \item replace $(P,\bbf)$ with $(P',\bbf')$.

    \end{enumerate}

\end{enumerate}
The algorithm accepts after $N:=2^{|\Amf|}\cdot
(|\text{dom}(\Bmf)|^{\mn{ar}(\tau)}+1)+1$ where $|\Amf|$ is length of the
  representation of \Amf, that is, roughly the number of ground atoms.
  Note that steps~1) and~2) correspond to Conditions~(\atomgh)
  and~(\simgh), respectively, and steps~3) and~4) correspond
  to~(\forthgh) and~(\backgh), respectively. Thus, to establish
  correctness it suffices to prove that after $N$ steps, we know that
  there is a normal guarded Horn simulation $Z$ with $(P_0,\bbf_0)$. This is
  a consequence of the fact that there are at most $N-1$ links:
  \begin{itemize}

    \item the number of possible tuples $\bbf$ is bounded by
      $\text{dom}(\Bmf)^{\mn{ar}(\tau)}+1$, by normality;

    \item as argued in the beginning of the proof, for a fixed guarded
      tuple $\bbf$ the number of possible sets $P$ is bounded by
      $2^{|\Amf|}$.

  \end{itemize}
  Thus, after $N$ steps, the algorithm has visited a link twice and
  can stop. This argument also shows that the size of each link is
  bounded by some polynomial in the size of the representations of
  \Amf and \Bmf. Moreover, we can count (in binary) up to $N$ in
  polynomial space. Overall, the algorithm is an alternating \PSpace
  algorithm which suffices to show an \ExpTime upper bound.

  The upper bounds for the restricted problems are obtained in the
  same way from the alternating algorithm as discussed in
  Lemma~\ref{lem:k-hornsim} for \hALC.
\end{proof}

\subsection{Expressive Completeness}

We remind the reader of the definition of $\omega$-saturated structures. Let $\tau$ be a finite vocabulary and
assume that $\Amf$ is a $\tau$-structure. Then $\Amf$ is \emph{$\omega$-saturated} if for all tuples 
$\abf$ in $\text{dom}(\Amf)$ and all sets $\Gamma(\xbf\ybf)$ of FO$[\tau]$-formulas with $\ybf$ and $\abf$ of 
the same length, if $\Amf\models \exists \xbf (\bigwedge_{\varphi\in \Gamma'}\varphi(\xbf\abf))$ for all 
finite subsets $\Gamma'$ of $\Gamma$, then there exists a tuple 
$\bbf$ in $\text{dom}(\Amf)$ such that $\Amf\models \Gamma(\bbf\abf)$.
Every satisfiable set of FO$[\tau]$-formulas is satisfiable in an $\omega$-saturated structure~\cite{ChangKeisler}. 

\medskip
\noindent
{\bf Theorem~\ref{thm:exprhgf}}
A FO-formula is equivalent to a \hGF-formula just in case it is preserved under FO-restricted
generalized guarded Horn simulations.

\medskip
\noindent
\begin{proof} The direction from left to right is straightforward. Conversely,
suppose $\varphi(\xbf_{0})$ is preserved under generalized guarded Horn simulations.
Let $\text{cons}(\varphi)$ be the set of all $\psi(\xbf_{0})$ in \hGF{} entailed by $\varphi(\xbf_{0})$.
By compactness, it suffices to show $\text{cons}(\varphi)\models \varphi$. Let $\Bmf$ be an $\omega$-saturated
model satisfying $\text{cons}(\varphi)(\bbf_{0})$ for some tuple $\bbf_{0}$ in $\text{dom}(\Bmf)$. 
We show $\Bmf\models \varphi(\bbf_{0})$. For any tuple $\bbf$ and tuple $\xbf$ of variables of the same length as $\bbf$, we denote by 
$\lambda_{\Bmf,\bbf}(\xbf)$ the set of guarded existential 
positive $\lambda(\xbf)$ with $\Bmf\models \lambda(\bbf)$.
Let $\mathcal{C}$ be the set of all sets $\Gamma(\xbf_{0})$ of FO-formulas with $\Bmf\models \Gamma(\bbf_{0})$ and such that
$\Gamma(\xbf_{0}) \cup \{\varphi(\xbf_{0})\}$ is satisfiable and take, for any $\Gamma(\xbf_{0})\in \mathcal{C}$, an $\omega$-saturated structure $\Amf_{\Gamma}$ and tuple $\abf_{\Gamma}$ with $\Amf_{\Gamma}\models (\Gamma\cup \{\varphi\})(\abf_{\Gamma})$.
Let $\Amf$ be the disjoint union of $(\Amf_{\Gamma}\mid \Gamma\in \mathcal{C})$ and let
$Z$ be the set of pairs $(X,\bbf)$ such that
\begin{itemize}
\item[(a)] for any $\psi(\xbf)\in \hGF{}$, if $(\Amf_{\Gamma}\mid \Gamma\in \mathcal{C})
\models \psi(\abf)$ for all $\abf\in X$, then $\Bmf\models \psi(\bbf)$;
\item[(b)] there exists a set $\Phi(\xbf)\supseteq \lambda_{\Bmf,\bbf}$ of FO-formulas
such that $X$ is the set of all tuples $\abf$ with  
$(\Amf_{\Gamma}\mid \Gamma\in \mathcal{C}) \models \Phi(\abf)$.
\end{itemize}
Each $(X,\bbf)$ in $Z$ can be regarded as a link of the form $(P,\bbf)$ with $X=P[\bbf]$.
We show that $Z$ is an FO-restricted guarded Horn simulation between $(\Amf_{\Gamma}\mid \Gamma\in \mathcal{C})$ and $\Bmf$. 

Assume $(X,\bbf)$ is given. Let $\Gamma_{X}$ be a set of FO-formulas that defines $X$ in 
$(\Amf_{\Gamma}\mid \Gamma\in \mathcal{C})$. Let $(P,\bbf)$ be the link corresponding to $(X,\bbf)$. 
We check the conditions.

\medskip
\noindent
The Condition~(\atomgh) follows from Condition~(a).

\medskip
\noindent
To show that Condition~(\forthggh) holds, 
assume that $Y$ is a set of guarded tuples such that there is a set $\Gamma_{Y}$ of FO-formulas with free
variables among $\xbf_{0}\ybf$ defining $Y$ in $(\Amf_{\Gamma}\mid \Gamma\in \mathcal{C})$.
Assume first that $Y$ is an $R(\bbf_{0}\ybf)$-successor of
$(P,\bbf)$ and $\bbf_{0}$ is not empty. We have to show that there exists $(P',\bbf_{0}\bbf')\in Z$ such that 
$P'[\bbf_{0}\bbf'] \subseteq Y$. Let $\rho_{(\Amf_{\Gamma}|\Gamma\in\mathcal{C}),Y}$ be the set of all formulas $\rho$ in \hGF{} with free variables 
among $\xbf_{0}\ybf$ such that $(\Amf_{\Gamma} \mid \Gamma\in \mathcal{C})\models \rho(\abf_{0}\abf)$ for all $\abf_{0}\abf\in Y$.
For every $\rho(\xbf_{0}\ybf)$ in $\rho_{(\Amf_{\Gamma}|\Gamma\in\mathcal{C}),Y}$, and $p\in P$:
$$
(\Amf_{\Gamma}\mid \Gamma\in \mathcal{C})\models \exists \ybf (R(p(\bbf_{0}),\ybf) \wedge \rho(p(\bbf_{0}),\ybf)).
$$ 
Thus 
$$
\Bmf\models \exists \ybf (R(\bbf_{0}\ybf) \wedge \rho(\bbf_{0}\ybf)),
$$ 
for every $\rho(\xbf_{0}\ybf)$ in $\rho_{(\Amf_{\Gamma}|\Gamma\in\mathcal{C}),Y}$.
By $\omega$-saturatedness of $\Bmf$, there exists a tuple $\bbf'$ such that
$$
\Bmf\models R(\bbf_{0}\bbf'), \quad \Bmf \models \rho_{(\Amf_{\Gamma}|\Gamma\in\mathcal{C}),Y}(\bbf_{0}\bbf').
$$
Now we set 
$$
Y'= \{\abf \in Y \mid (\Amf_{\Gamma}\mid \Gamma\in \mathcal{C})\models \lambda_{\Bmf,\bbf_{0}\bbf'}(\abf)\}
$$
One can show that $(Y',\bbf_{0}\bbf')$ is as required. 
To show that $(Y',\bbf_{0}\bbf')\in Z$ it suffices to show that 
$Y'$ is not empty and satisfies Conditions~(a) and (b).
\begin{itemize}
\item Assume $Y'=\emptyset$. Then 
$$
(\Amf_{\Gamma}\mid \Gamma\in \mathcal{C})\not\models \lambda_{\Bmf,\bbf_{0}\bbf'}(\abf)
$$
for any $\abf \in Y$. Then $(\Gamma_{Y}\cup \lambda_{\Bmf,\bbf_{0}\bbf'})(\xbf_{0}\ybf)$ is not satisfied in any $\Amf_{\Gamma}$.
By compactness and $\omega$-saturatedness of every $\Amf_{\Gamma}$, 
there exist finite subsets $\Gamma_{0}$ of $\Gamma_{Y}$ and
$\lambda_{\Bmf,\bbf_{0}\bbf'}'$ of $\lambda_{\Bmf,\bbf_{0}\bbf'}$ such that $\Gamma_{0} \cup\lambda_{\Bmf,\bbf_{0}\bbf'}'$
is not satisfied in any $\Amf_{\Gamma}$. But then 
$$
(\Amf_{\Gamma}\mid \Gamma\in \mathcal{C})\models(\lambda_{\Bmf,\bbf_{0}\bbf'}'\rightarrow \bot)(\abf),
$$
for all $\abf\in Y$. $(\lambda_{\Bmf,\bbf_{0}\bbf'}'\rightarrow \bot)$ is then in \hGF{} and so
$$
\Bmf\models (\lambda_{\Bmf,\bbf_{0}\bbf'}'\rightarrow \bot)(\bbf_{0}\bbf'),
$$
by construction of $\bbf_{0}\bbf'$, and we have derived a contradiction.
\item For Condition~(a), assume that $\rho$ is in \hGF{} and 
$\Amf \models \rho(\abf)$ for all $\abf\in Y'$. 
We have to show that $\Bmf\models \rho(\bbf_{0}\bbf')$. But this can be shown similarly to the
non-emptiness proof above.
\item For Condition (b), observe that 
$$
\Phi = (\Gamma_{Y} \cup \lambda_{\Bmf,\bbf_{0}\bbf'})(\xbf_{0}\ybf)
$$
is as required.
\end{itemize}
Assume next that $Y$ is an $R(\ybf)$-successor of
$(P,\bbf)$ and $Y$ intersects with all $\Amf_{\Gamma}$.
We have to show that there exists $(P',\bbf')\in Z$ such that 
$P'[\bbf'] \subseteq Y$.
As $Y$ intersects with every $\Amf_{\Gamma}$, it follows that 
$$
(\Amf_{\Gamma}\mid \Gamma\in \mathcal{C})\models \exists \ybf (R(\ybf) \wedge \rho(\ybf))
$$ 
for every $\rho(\ybf)$ in $\rho_{(\Amf_{\Gamma}|\Gamma\in\mathcal{C}),Y}$.
Thus,
$$
\Bmf\models \exists \ybf (R(\ybf) \wedge \rho(\ybf)),
$$
for every $\rho(\ybf)$ in $\rho_{(\Amf_{\Gamma}|\Gamma\in\mathcal{C}),Y}$.
By $\omega$-saturatedness of $\Bmf$, there exists a tuple $\bbf'$ such that
$$
\Bmf\models R(\bbf'), \quad \Bmf \models \rho_{(\Amf_{\Gamma}|\Gamma\in\mathcal{C}),Y}(\bbf')
$$
Consider $\lambda_{\Bmf,\bbf'}(\ybf)$. We let 
$$
Y'= \{\abf \in Y \mid (\Amf_{\Gamma}\mid \Gamma\in \mathcal{C})\models \lambda_{\Bmf,\bbf'}(\abf)\}
$$
One can now show similarly to the previous case that $(Y',\bbf')$ is as required.

\medskip

To show Condition~(\backgh), assume that a guarded tuple $\bbf'$ is given. Consider w.l.o.g. an atomic formula 
$R(\xbf')$ with $\Bmf\models R(\bbf')$. Let $\bbf'= \bbf_{0}\bbf_{1}$, where $\bbf_{0}$ and $\bbf_{1}$ are such that
$[\bbf_{0}]\subseteq [\bbf]$ and $[\bbf_{1}]\cap [\bbf]=\emptyset$. Let $X'$ denote the set of all tuples $\abf_{0}\abf_{1}$ such that
there exists $\Gamma\in \mathcal{C}$ with $\abf_{0}\abf_{1}$ in $\Amf_{\Gamma}$ and 
such that
\begin{itemize}
\item $\Amf_{\Gamma}\models \lambda_{\Amf,\bbf'}(\abf_{0}\abf_{1})$ and
\item there exists a tuple $\abf''$ with $\Amf_{\Gamma}\models \Gamma_{X}(\abf_{0}\abf'')$.
\end{itemize}
Define $\Gamma_{X'}$ as the set of all formulas
$$
\exists \xbf_{1} (\lambda(\xbf_{0}\xbf_{1}) \wedge \psi(\xbf_{0}))
$$
where $\lambda(\xbf_{0}\xbf_{1})\in \lambda_{\Amf,\bbf'}(\xbf_{0}\xbf_{1})$ and
$\psi(\xbf_{0})= \exists \xbf''\psi'(\xbf_{0}\xbf'')$ for some $\psi(\xbf_{0}\xbf'')\in
\Gamma_{X}(\xbf_{0}\xbf'')$. Then, by $\omega$-saturatedness of all $\Amf_{\Gamma}$ with $\Gamma\in \mathcal{C}$,
we have that $\Gamma_{X'}$ defines $X'$:
$$
X' = \{\abf_{0}\abf_{1} \mid (\Amf_{\Gamma}\mid \Gamma\in \mathcal{C})\models \Gamma_{X'}(\abf_{0}\abf_{1})\}
$$
We show that $(X',\bbf')$ is as required. Condition~(b) is satisfied by definition.
For Condition~(a), assume $\rho(\xbf_{0}\xbf_{1})$ is in \hGF{} and $\Bmf\not\models\rho(\bbf_{0}\bbf_{1})$.
It suffices to show that there exists $\abf_{0}\abf_{\rho}\in X'$ such that 
$(\Amf_{\Gamma}\mid \Gamma\in \mathcal{C})\models \neg\rho (\abf_{0}\abf_{\rho})$.
But $\abf_{0}\abf_{\rho}$ exists because otherwise there exists $\lambda\in\lambda_{(\Amf_{\Gamma}\mid \Gamma\in \mathcal{C}),\bbf'}$
such that   
\begin{equation}\label{first}
(\Amf_{\Gamma}\mid \Gamma\in \mathcal{C})\models \forall \xbf_{1}(R(\xbf_{0}\xbf_{1}) \rightarrow 
(\lambda \rightarrow \rho))(p(\bbf_{0})\xbf_{1})
\end{equation}
for all $p\in P$ and so, by definition of $Z$,
$$
\Bmf\models \forall \xbf_{1}(R(\xbf_{0}\xbf_{1}) \rightarrow 
(\lambda \rightarrow \rho))(\bbf_{0},\xbf_{1}),
$$
and we have derived a contradiction. To prove \eqref{first}, assume
\eqref{first} does not hold. Then, for every $\lambda'\in\lambda_{(\Amf_{\Gamma}\mid \Gamma\in \mathcal{C}),\bbf'}$,
there exists $p\in P$ such that 
$$
(\Amf_{\Gamma}\mid \Gamma\in \mathcal{C})\models \exists \xbf_{1}(\lambda'\wedge \neg\rho)(p(\bbf_{0}),\xbf_{1}).
$$
By definition, it follows that for every $\lambda'\in\lambda_{(\Amf_{\Gamma}\mid \Gamma\in \mathcal{C}),\bbf'}$ there exists $\Gamma\in\mathcal{C}$ such that 
$\Amf_{\Gamma}$ realizes
$$
\Gamma_{X}(\xbf_{0}\xbf'')\cup \{(\lambda'\wedge \neg\rho)(\xbf_{0}\xbf_{1})\}
$$ 
But then, by $\omega$-saturatedness, compactness, and the definition of $(\Amf_{\Gamma}\mid \Gamma\in \mathcal{C})$,
there exists $\Gamma\in \mathcal{C}$ such that $\Amf_{\Gamma}$ realizes
$$
\Gamma_{X}(\xbf_{0}\xbf'')\cup \lambda_{(\Amf_{\Gamma}\mid \Gamma\in \mathcal{C}),\bbf'}(\xbf_{0}\xbf_{1})
\cup \{\neg\rho(\xbf_{0}\xbf_{1})\}
$$
which implies that there exists $\abf_{0}\abf_{\rho}\in X'$ such that 
$(\Amf_{\Gamma}\mid \Gamma\in \mathcal{C})\models \neg\rho (\abf_{0}\abf_{\rho})$.

\medskip

Finally, Condition~(\emph{sim}$^{gg}_h$) can be proved using $\omega$-saturatedness of $\Bmf$ and all $\Amf_{\Gamma}$.
\end{proof}

\section*{Extending \hALC{} with the $\nabla$-operator}

Denote by $\mathcal{ELU}_{\nabla}$ the extension of $\mathcal{ELU}$ with the $\nabla$-operator defined by
setting
$$
\nabla R.C ~=~ \exists R.\top \sqcap \forall R.C.
$$
Thus, \emph{$\mathcal{ELU}_{\nabla}$-concepts} are given by the grammar
$$
C,D ~::=~ A \ \mid \ \top \ \mid \  C \sqcap D \ \mid \  C \sqcup D \ \mid \  \exists R.C \ \mid \ \nabla R.C.
$$
We define $\hALC_{\nabla}$ in the same way as \hALC{} (see Definition~\ref{defHornALC}) with the exception that now $L$ is an arbitrary $\mathcal{ELU}_{\nabla}$-concept. The following lemma shows that,  modulo the standard translation, $\hALC_{\nabla}$ is a fragment of Horn FO.

\begin{lemma}\label{Lnabla}
Every \hALC$_{\nabla}$-concept is equivalent to a Horn formula.
\end{lemma}
\begin{proof}
We proceed by induction on the construction of \hALC$_{\nabla}$-concepts. The only non-trivial case is $L \to C$, where $L$ is an $\mathcal{ELU}_{\nabla}$-concept and $C^\dag$ is equivalent to a Horn formula. This case follows from the claim below.

\medskip
\noindent
{\it Claim}. If $H(\zbf)$ is any Horn formula and $L^\dag(x)$ is the standard translation of an $\mathcal{ELU}_{\nabla}$-concept $L$, then $L^{\dag}(x)\rightarrow H(\zbf)$ is equivalent to a Horn formula.

\smallskip
\noindent
{\it Proof of claim} is by induction on the construction of $L$. The basis of induction and the cases when $L$ is $A$, $\top$ and $C \sqcup D$ are obvious. 

\emph{Case} $(C^\dag \land D^\dag)(x) \to H(\zbf)$. This formula is equivalent to $C^\dag(x) \to (D^\dag(x) \to H(\zbf))$. By IH, $D^\dag(x) \to H(\zbf)$ is equivalent to some Horn formula $H'(x,\zbf)$, and by the same reason $C^\dag(x) \to H'(x,\zbf)$ is also equivalent to a Horn formula.

\emph{Case} $(\exists R.C)^{\dag}(x) \to H(\zbf)$. This formula is equivalent to $\forall y\, ((R(x,y) \wedge C^{\dag}(y)) \to H(\zbf))$ with a fresh $y$, 
which is clearly equivalent to a Horn formula, by IH.

\emph{Case} $(\nabla R.C)^{\dag} \rightarrow H(\zbf)$. Its standard translation is equivalent to
$$
\exists y\, R(x,y) \to [\forall y (R(x,y) \rightarrow C^{\dag}(y)) \rightarrow H(\zbf)].
$$
One can check that this formula has the same models as 
$$
\exists y\, R(x,y) \to \exists y [R(x,y) \wedge (C^{\dag}(y) \rightarrow H(\zbf))],
$$
which is equivalent to a Horn formula, by IH. 
\end{proof}

\begin{definition}[\bf $\mathcal{ELU}_{\nabla}$-simulation]\em
 An \emph{$\mathcal{ELU}_{\nabla}$-simulation} between $\tau$-structures $\Amf$ and $\Bmf$ is a relation 
$Z \subseteq \text{dom}(\Amf)\times \text{dom}(\Bmf)$ if the
following conditions hold:
\begin{description}
\item[\hspace*{-1.5mm}(\emph{atom})] for any $A\in \tau$, if $(a,b) \in Z$ and $a\in A^{\Amf}$, then 
$b\in A^{\Bmf}$;
 
\item[\hspace*{-1.5mm}(\emph{forth})] for any $R\in \tau$, if $(a,b) \in Z$ and $(a,a')\in R^{\Amf}$, then there exists
$b'\in \text{dom}(\Bmf)$ with $(b,b')\in R^{\Bmf}$ and $(a',b') \in Z$;
       
\item[\hspace*{-1.5mm}(\emph{back})] for any $R\in \tau$, if $(a,b) \in Z$, $a \in (\exists R.\top)^{\Amf}$, 
and $(b,b')\in R^{\Bmf}$, then there is $a'\in \text{dom}(\Amf)$ with $(a,a')\in R^{\Amf}$ and $(a',b')\in Z$.
%
\end{description}
We write $\Amf,a \preceq_{\nabla}\Bmf,b$ if there exists a $\mathcal{ELU}_{\nabla}$-simulation $Z$ between $\Amf$ and 
$\Bmf$ such that $(a,b) \in Z$.
\end{definition}

\begin{theorem}[\bf Ehrenfeucht-Fra\"{i}ss\'e game for $\mathcal{ELU}_{\nabla}$]\label{thm:elunabla}
For any finite vocabulary $\tau$, pointed $\tau$-structures $\Amf,a$ and $\Bmf,b$, and any 
$\ell < \omega$, we have
$$
\Amf,a \leq_{\mathcal{ELU}_{\nabla}}^{\ell} \Bmf,b \quad \text{iff} \quad 
\Amf,a \preceq_{\nabla}^{\ell}\Bmf,b.
$$
If $\Amf$ and $\Bmf$ are finite, then
$$
\Amf,a \leq_{\mathcal{ELU}_{\nabla}} \Bmf,b \quad \text{iff} \quad 
\Amf,a \preceq_{\nabla}\Bmf,b.
$$
\end{theorem}

\begin{definition}[\bf Horn$_{\nabla}$-simulation]\em
 A \emph{Horn$_{\nabla}$ simulation} $Z$ between $\tau$-structures $\Amf$ and $\Bmf$ is a Horn simulation
 such that, in addition, 
\begin{description}
\item[\hspace*{-1.5mm}(\emph{sim}$_{\nabla}$)] if $(X,b) \in Z$, then $\Bmf,b\preceq_{\nabla}\Amf,a$ for every $a\in X$.
\end{description}
We write $\Amf,X \preceq_{\textit{horn}_{\nabla}}\Bmf,b$ if there exists a Horn$_{\nabla}$ simulation $Z$ between 
$\Amf$ and $\Bmf$ such that $(X,b) \in Z$.
\end{definition}
\begin{theorem}[\bf Ehrenfeucht-Fra\"{i}ss\'e game for \hALC$_{\nabla}$]\label{thm:ehrenhornnabla}
For any finite vocabulary $\tau$, pointed $\tau$-structures $\Amf,a$ and $\Bmf,b$,
and any $\ell < \omega$, we have
$$
\Amf,a \leq_{\textit{horn}\mathcal{ALC}_{\nabla}}^{\ell} \Bmf,b \quad \text{iff} \quad 
\Amf,a \preceq_{\textit{horn}_{\nabla}}^{\ell}\Bmf,b.
$$
If $\Amf$ and $\Bmf$ are finite, then
$$
\Amf,a \leq_{\textit{horn}\mathcal{ALC}_{\nabla}} \Bmf,b \quad \text{iff} \quad \Amf,a \preceq_{\textit{horn}_{\nabla}}\Bmf,b.
$$
\end{theorem}

\end{document}